\documentclass{article}
\usepackage{amssymb}
\usepackage{graphicx}
\usepackage{amsmath}
\DeclareMathAlphabet{\mathpzc}{OT1}{pzc}{m}{it}
\usepackage{bm}
\usepackage{mathrsfs}
\usepackage{hyperref}
\usepackage[toc,page]{appendix}
\newtheorem{theorem}{Theorem}

\newtheorem{corollary}[theorem]{Corollary}

\newtheorem{definition}[theorem]{Definition}
\newtheorem{example}[theorem]{Example}

\newtheorem{lemma}[theorem]{Lemma}

\newtheorem{proposition}[theorem]{Proposition}

\newenvironment{proof}[1][Proof]{\textbf{#1.} }{\ \rule{0.5em}{0.5em}}
\DeclareMathAlphabet{\mathpzc}{OT1}{pzc}{m}{it}

\newcommand {\End}{\mathrm{End}}
\newcommand{\pr}[1]{\frac{\partial}{\partial #1}}

\newcommand {\Mat}{\mathrm{Mat}}

\newcommand {\susy}{\mathfrak{susy}}

\newcommand {\sd}{\partial}

\newcommand {\red}{\mathrm{red}}

\newcommand {\id}{\mathrm{id}}
\newcommand {\Hom}{\mathrm{Hom}}
\newcommand{\Imm}{\mathrm{Im}}
\newcommand{\cobar}{\mathrm{cobar}}

\newcommand {\F}{{\cal F}}

\newcommand {\C}{\mathcal{C}}

\renewcommand {\L}{\mathcal{L}}
\newcommand {\Ss}{\mathcal{S}}

\newcommand {\tr}{\mathrm{tr}}
\newcommand {\U}{\mathrm{U}}
\newcommand {\V}{\mathbb{V}}

\newcommand {\Spin}{\mathrm{Spin}}

\newcommand {\SO}{\mathrm{SO}}
\newcommand {\so}{\mathfrak{so}}

\newcommand {\GL}{\mathrm{GL}}

\newcommand {\vol}{\mathrm{vol}}
\newcommand {\dL}{\delta\mathcal{L}}
\newcommand {\ddL}{\delta'\mathcal{L}}
\newcommand {\ddef}{\mathrm{def}}
\newcommand {\M}{\mathcal{M}}
\newcommand {\N}{\mathcal{N}}
\newcommand {\W}{\mathcal{W}}

\newcommand {\BEE}{\mathcal{E}^{\bullet}}

\newcommand {\TQ}{\widetilde Q}
\newcommand {\bz}{z}

\newcommand{\rS}{\mbox{{\rm S}}}

\renewcommand {\O}{{\cal O}}
\newcommand {\YM}{\mathcal{YM}}
\newcommand {\TYM}{\mathcal{TYM}}

\newcommand {\btheta}{\bm{\theta}}
\newcommand {\bD}{\bm{D}}

\newcommand {\bchi}{\bm{\chi}}

\newcommand{\g}{\mathfrak{g}}
\newcommand{\p}{\mathfrak{p}}
\newcommand{\vvv}{\mathfrak{v}}
\renewcommand{\Im}{\mbox{{\rm Im}}}

\newcommand{\UN}{\mathrm{U}(N)}

\newcommand{\Sym}{\mathrm{Sym}}
\newcommand {\Ker}{\mbox{{\rm Ker}}}

\newcommand{\dbar}{\bar \partial}

\newcommand{\pp}{\mathpzc{p}}

\newcommand{\ity}{_{\infty}}

\renewcommand{\tt}{I}

\newcommand{\rH}{\mathrm{H}}
\newcommand{\rHH}{\mathrm{HH}}
\newcommand{\rHC}{\mathrm{HC}}

\begin{document}
\title {Supersymmetric Deformations of Maximally Supersymmetric  Gauge Theories.} 
\author{M. V. Movshev\\Stony Brook University\\Stony Brook, NY 11794-3651, USA
\\ A. Schwarz\thanks{The work of both  authors was partially supported by NSF 
grant No. DMS 0505735 and  by grants DE-FG02-90ER40542  and PHY99-0794}\\ Department of Mathematics\\ 
University of 
California \\ Davis, CA 95616, USA}

\date{\today}
\maketitle
\tableofcontents
\begin{abstract}
We study supersymmetric and super Poincar\'e invariant deformations of ten-dimensional super Yang-Mills theory and of 
its dimensional reductions.
We describe all infinitesimal super Poincar\'e invariant deformations of equations of motion of ten-dimensional super Yang-Mills theory and its reduction to a point; we discuss the extension of them to formal deformations.  Our methods are based on
homological algebra, in particular, on the theory
of L-infinity and A-infinity algebras. The exposition of this theory as well as of some basic facts about Lie algebra 
homology and Hochschild homology is given in appendices.
\end{abstract}

\section{Introduction}\label{S:introduction}

The superspace technique is a very powerful tool of construction of supersymmetric theories. However this technique 
does not work for theories with large number of supersymmetries. It is possible to apply methods of homological algebra 
and formal non-commutative geometry to prove existence of supersymmetric deformations of gauge theories and give 
explicit construction of them. In this paper we discuss results obtained by such methods in the analysis of supersymmetric (SUSY) 
deformations of 10-dimensional SUSY YM-theory (SYM theory) and its dimensional reductions.

These  deformations are quite important from the
viewpoint of string theory. The D-brane action in the first approximation is given by dimensional reduction of  ten-dimensional SYM theory; taking 
into account the $\alpha '$ corrections we obtain SUSY deformation of this theory. (More precisely, we obtain a power 
series with respect to $\alpha '$ specifying a formal deformation of the theory at hand.) 
 
 Our approach is closely related to pure spinors techniques; it seems that it could be quite useful to understand better the 
pure spinor formalism in string theory constructed by Berkovits \cite {Berko}.

 Recall that  a $\UN$  gauge potential $ A_i(x)$ and a chiral spinor $\chi^{\alpha}$ are the fields of the SYM-theory.
In component form the action functional of SYM-theory looks as follows: 
\begin{equation}\label{E:SYM}
S_{SYM}(A,\chi)=\int \L_{SYM}d^{10}x=\int \tr\bigg(\frac{1}{4}F_{ij}F_{ij}+\frac{1}{2}\Gamma_{\alpha\beta}^i\chi^{\alpha}
\nabla_i\chi^{\beta}\bigg)d^{10}x 
\end{equation} 
$\nabla_i=\frac{\partial}
{\partial x^i}+ A_i(x)$ are covariant derivatives , $\chi^{\alpha}$ are chiral spinors with values in the adjoint representation, 
$F_{ij}=[\nabla_i,\nabla_j]$ is the curvature.\footnote{In this text by default small Roman indices $i,j$ run over $1,\dots,10$, Greek 
indices $\alpha,\beta,\gamma$ run over $1,\dots,16$

In the above formula $A_1,\cdots A_{10}$ and $\chi ^1,\cdots ,\chi ^{16}$ denote fields on ten-dimensional space $\mathbb{R}^{10}$ taking values in complex $N\times N$-matrices. (We do not impose the hermiticity condition; in our approach this condition appears only in the choice of real slice in the space of fields that we should use in the definition of the functional integral of the quantum theory.)

.  We shall freely raise and lower Roman indices with  ten-dimensional metric $(dx^i)^2$. We use Einstein summation convention over repeated indices.} 

Let $Y$ be a linear combination of  products of covariant derivatives of the curvature  $F_{ij}$ and  spinor fields $\chi^
{\alpha}$. By construction $Y$ is  is a section of the adjoint bundle. The function $A_i,\chi^{\alpha}\rightarrow Y(A_i,\chi^{\alpha})$ commutes with the gauge group action. 

We shall refer to such $Y$ as a {\it gauge-covariant local expression}. A {\it local gauge-invariant expression} is by definition $\tr(Y)$ with local gauge-covariant $Y$. 
In this paper we shall  consider deformation that can be described by  action functionals of the form
\begin{equation}\label{E:cvcjdfj}
\int \L_{SYM}+\sum_{i\geq1} \int \tr(Y_i)\epsilon^i
\end{equation}
 where $\tr Y_i$ are  gauge invariant, $\epsilon$ is a formal deformation parameter.



 In our approach the integrals are invariant with respect to some field transformation iff 
the variation of the integrand is a total derivative.
 We consider only deformations that can be applied simultaneously to gauge theories with all gauge groups $\UN$ where 
$N$ is an arbitrary positive integer. This  means that  we miss some 
of the deformations (e.g. certain  deformations of the abelian gauge theory) that are defined  only for a finite range of $N$.

Our methods can be applied to supersymmetric deformations of  dimensional reductions of SYM theory.
 A SYM deformation after reduction defines a deformation of the corresponding reduced 
theory. However not all deformations of the  reduced theory are of this kind. We shall give a complete description of SUSY-
deformations of  the full SYM theory and of its reduction to D=0 . 

Let us list explicit formulas for the action of supersymmetry generators and translations on fields of SYM-theory.
The Lie algebra of supersymmetry acts on gauge equivalence classes of solutions of equations
\begin{align}
&\sum_{i=1}^{10}\nabla_{i}F_{im}-\frac{1}{2}\sum_{\alpha\beta=1}^{16}\Gamma_{\alpha\beta}^m[\chi^{\alpha},\chi^{\beta}]=0 \quad m=1,\dots,
10 \label{E:relssfsyy1}\\
&\sum_{\beta=1}^{16}\sum_{i=1}^{10}\Gamma_{\alpha\beta}^i\nabla_i\chi^{\beta}=0\quad  \alpha=1\dots 16 
\label{E:relssfsyy3}
\end{align}
 The supersymmetry operators $\theta_{\alpha}$ are defined by formulas
\begin{equation}\label{E:susyform}
\begin{split}
&\theta_{\alpha}A_i=\Gamma_{\alpha\beta i}\chi^{\beta}\\
&\theta_{\alpha}\chi^{\beta}=\Gamma_{\alpha}^{\beta ij}F_{ij}
\end{split}
\end{equation}
 Denote by $D_i$ the lift of the space-time translation $\partial/\partial x^i$ to the space of  fields $(A_i,\chi^{\alpha})$.
 The lift is defined only up to gauge transformation.
 We fix the gauge freedom in a choice of $D_i$ requiring that

\begin{equation}\label{E:transl}
\begin{split}
&D_iA_j=F_{ij}\\
&D_i\chi^{\alpha}=\nabla_i\chi^{\alpha}
\end{split}
\end{equation}
 Infinitesimal symmetries $\theta_{\alpha}$ satisfy
\begin{equation}\label{E:hfdfjfhjd}
\begin{split}
&[\theta_{\alpha},\theta_{\beta}]=\Gamma_{\alpha\beta}^iD_i\\
&[\theta_{\alpha},D_i]A_k=- \Gamma_{\alpha\beta i}\nabla_k\chi^{\beta}\\
&[\theta_{\alpha},D_i]\chi^{\gamma}=\Gamma_{\alpha\beta i}[\chi^{\beta},\chi^{\gamma}]
\end{split}
\end{equation}
if  fields are solutions of equations of motion of $S_{SYM}$.

We see that on shell (on the space of solutions of the equations of motion
where gauge equivalent solutions are identified)  supersymmetry transformations commute with space-time translations:
\begin{equation}\label{E:supergaugetr}
\begin{split}
[ \theta_{\alpha},D_i]=0 \mbox{ on shell}.
\end{split}
\end{equation}
Hence on shell $\theta_{\alpha}, D_i$ generate the standard supersymmetry algebra (this is wrong off shell).
To reduce the theory to a point we should  consider  $A_1,\cdots A_{10}$ and $\chi ^1,\cdots ,\chi ^{16}$ as constant matrices. 

We will introduce the Lie algebra  $YM$ in such a way that an $N$-
dimensional representation of the algebra $YM$ gives a classical solution of  the  reduced SYM theory. More formally, we can define 
$YM$ as a Lie algebra having even generators $\bm{D}_1,\dots,\bm{D}_{10}$ and odd generators $\bm{\chi}^{1},\dots,
\bm{\chi}^{16}$ obeying relations 
\begin{align}
\sum_{i=1}^{10}[\bm{D}_{i},[\bm{D}_{i},\bm{D}_{m}]]- 
\frac{1}{2}\sum_{\alpha\beta=1}^{16}\Gamma_{\alpha\beta}^m[\bm{\chi}^{\alpha},\bm{\chi}^{\beta}]=0 \quad m=1,\dots,
10 \label{E:relssfs1}\\
\sum_{\beta=1}^{16}\sum_{i=1}^{10}\Gamma_{\alpha\beta}^i[\bm{D}_i,\bm{\chi}^{\beta}]=0\quad  \alpha=1\dots 16 
\label{E:relssfs3}
\end{align}

Universal enveloping algebra $U(YM)$ can be regarded as a unital associative algebra defined by the same relations as $YM.$

The algebra $YM$ is closely related to the graded  Lie algebra $L=\sum L$ with 
generators $\bm{\theta}_{1},\dots, \bm{\theta}_{16}$ of 
degree  one subject to relations
\begin{equation}\label{E:nvpri}
\Gamma^{\alpha\beta}_{i_1,\dots,i_5}[\bm{\theta}_{\alpha},\bm{\theta}_{\beta}]=0.
\end{equation} 

Namely, one can prove that $YM$ is isomorphic to $\bigoplus_{i\geq 2} L^i\subset L$
(the generator $\bm{D}_{m}$ is defined by the formula $[\bm{\theta}_{\alpha},\bm{\theta}_{\beta}]=\Gamma_{\alpha\beta}^i\bm{D}_i$  and the generator $\bm{\chi}^{\alpha}$ by the formula $\Gamma _{\alpha \beta i}\bm{\chi}^{\beta}=
[\bm {\theta} _{\alpha}, \bm {D}_i]$
).

We can define also the Lie algebra $TYM$  as follows
$$TYM\overset{\ddef}{=}\bigoplus_{i\geq 3} L^i\subset L.$$
This algebra is generated by     as a Lie algebra  (and its universal enveloping algebra 
$U(TYM)$ as associative algebra) 
 by expressions $\nabla _ {i_1}\cdots \nabla _{i_n}\Phi$
 where $\Phi$ is  either  ${\bm F}_{kl}$ or $ \bchi ^{\alpha}$ and  $\nabla _i (x)= [\bD_i,x]$
 (we use the notation  ${\bm F}_{ij}=[\bD_i,\bD_j]$).
This means that an element of $U(TYM)$ specifies a  gauge-covariant local expression.

In the following by  SUSY-deformations we understand deformations of the Lagrangian density  
and simultaneous deformation of  the sixteen supersymmetries. 

The SYM-theory has  additional 16 trivial supersymmetries  - the constant shifts of fermion fields. 
Analysis of deformations that preserves these symmetries was left out of scope of the present paper.


%
 As a first approximation to the problem we would like to solve we shall study
 infinitesimal  supersymmetric  deformations of equations of motion of ten-dimensional SUSY Yang-Mills theory. 
We translate this problem to a question  in homological algebra. The homological reformulation leads to a highly nontrivial, 
but solvable problem. 
We shall analyze also super  Poincar\'e invariant 
(= supersymmetric +Lorentz invariant)  infinitesimal deformations.  We shall prove that all of them are Lagrangian 
deformations of equations of motion (i.e. the deformed equations come from a deformed Lagrangian).

One of our  tools  is the theory of A$\ity$ and L$\ity$ algebras \cite{LadaStasheff}. The theory of L$\ity$ algebras is closely 
related to BV formalism (see e.g. \cite{StasheffBV},\cite{SchwarzBV}). The theory of L$\ity$ algebras with invariant odd inner product is equivalent to 
classical BV-formalism if we are working at  formal level \cite{AKSZ}. (This means that we are considering all functions at hand as 
formal power series). The theory of A$\ity$ algebras arises if we would like to consider Yang-Mills theory for all gauge 
groups $\UN$ at the same time (see Appendix A). 

Recall \cite{SchwarzBV} that in classical BV-formalism the space of solutions to the equations of motion (EM) can be characterized as zero locus 
$Sol$ of the odd vector field $Q$ that satisfies $[Q,Q]=0$. {\footnote {We use a unified notation $[\cdot,\cdot]$ for the 
commutators and super-commutators.}}
It is convenient to work with the space $Sol/\thicksim$ obtained from zero locus $Sol$ after identification of physically 
equivalent solutions (see Sec. \ref {S:BV} for detail).

One can consider $Q$ as a derivation of the algebra of functionals 
on the space of fields $M$.
One of the pieces of the input data  an odd symplectic structure on the space $M$; vector field $Q$ preserves the odd symplectic form.  The 
corresponding derivation can be written in the form $Qf=\{S,f\}$ where $\{\cdot, \cdot\}$ stands for the  odd Poisson 
bracket and $S$ plays
the role of the  action functional in the  BV formalism. 

A vector field $q_0$ on $M$ is an infinitesimal symmetry of EM in BV-formalism if $[Q,q_0]=0$; we  disregard trivial symmetries (symmetries of the form $q_0=[Q,q'_0]$). Let  $f_{\tau_1\tau_2}^{\tau_3}$ be the structure constants of a Lie algebra $\g$.  The infinitesimal symmetries $q_{\tau}$ in BV formalism define a $\g$-action if 
\begin{equation}\label{E:L2}
[q_{\tau_1},q_{\tau_2}]=f_{\tau_1\tau_2}^{\tau_3}q_{\tau_3}+[Q,q_{\tau_1\tau_2}] .
\end{equation}
Here the index $\tau$ labels a basis in $\g$.
We say in this case that $\g$ acts weakly on the space of fields. 
It is more convenient to work with  L$\ity$ actions of $\g$. To define an L$\ity$ action we complete the sequence of vector fields $q_{\tau},q_{\tau_1\tau_2}$ by their higher analogs $q_{\tau_1,...,\tau_k}, k=3,\dots$ and impose  relations
generalizing (\ref {E:L2}). These relations are easy to write if we introduce a generating function $q=\sum \frac{1}{k!}q_{\tau_1,...,\tau_k}c^{\tau_1}\dots c^{\tau_k}$: 
\begin{equation}\label{E:laction}d_{\g}q+[Q,q]+\frac{1}{2}[q,q]=0.
\end{equation} 
 Here 
\begin{equation}
d_{\g}=\frac{1}{2}(-1)^{|c^{\tau_1}|}f_{\tau_2\tau_1}^{\tau_3}c^{\tau_1}c^{\tau_2} \frac{\partial}{\partial c^{\tau_3}}
\end{equation} stands for the differential calculating the Lie algebra cohomology of $\g$,
$c^{\tau}$ are the ghosts corresponding to the Lie algebra (cf. \cite{HenneauxTeitelboim} where the sign of the operator is opposite to ours). Lagrangian  BV version of Equation (\ref{E:laction}) contains  Poisson bracket of functionals instead of  the supercommutators of vector fields.

Deformations of a theory that preserve a Lie algebra of symmetries in BV language become deformations of a solution of equation (\ref {E:L2}).
 It is important to emphasize that we can start  with an arbitrary BV formulation of the  given theory 
and the classification of deformations  does not depend on our choices.  Classification of infinitesimal deformations is a homological problem. This problem concerns  cohomology of the differential
$d_{\g}+[q, \cdot]$ that acts in the space of  ghost-dependent vector fields. 

 In the present paper
we apply  such homological methods to the ten-dimensional SUSY Yang-Mills (SYM) theory and to its dimensional reductions.  In particular, we describe  all infinitesimal super Poincar\'e invariant  deformations of ten-dimensional SYM  theory and its reduction to a point.  
We show that almost all of them are given by a simple general formula. We analyze the extension of  SUSY-invariant   
infinitesimal deformations  to formal SUSY-invariant  deformations. A formal deformation of a Lagrangian is a 
deformation that can be written as a formal power series with respect to  some parameter; in string theory the role of this parameter is played by string tension $\alpha'$. In the context of Yang-Mills theory we require  
that the Taylor coefficients of the series   are derivatives of  finite order  of coefficients of the gauge potential and matter fields.

The paper will be organized in the following way:
Preliminaries (Section \ref{S:Preliminaries}) contains some mathematical information needed in our constructions and 
proofs.  We do not use essentially the material of this section in sections   \ref{S:fdgdfgmjl} and  \ref{S:homologicalapp}.Therefore one can skip this section at the first reading  and start reading with Section  \ref{S:fdgdfgmjl} returning to Section \ref
{S:Preliminaries} as necessary.
 In Section  \ref{S:fdgdfgmjl} we give a complete description of  infinitesimal SUSY deformations of ten-dimensional SYM theory and its reduction to a point.  
 We give a very explicit formula that works  almost for  all  deformations. \footnote { One can  show that exceptional deformations are related to the 
homology of  SUSY Lie algebra.}
  In Section \ref{E:Formal} we sketch the proof of the fact that  infinitesimal SUSY deformations can be extended to formal deformations.   
In Section \ref{S:homologicalapp} we reduce the computation of the infinitesimal SUSY deformations to a homological 
problem. In Section \ref{S:calc} and in Appendix \ref{AppendixE} we sketch the solution of this problem. In Section \ref{S:BV}  we 
approach to the problem of infinitesimal deformations from the point of view of BV formalism.  This approach leads to 
another homological formulation of our problem; this is the formulation sketched in the introduction.

In Appendix  \ref{AppendixC} we relate this formulation to formulation of Section \ref{S:homologicalapp}. The approach 
based on BV formalism works in more  general situation than the formulation in Section \ref{S:homologicalapp}.
   
  The reader who is more interested in methods rather  in concrete results can choose Appendices 
\ref{AppendixA}, \ref{AppendixB} as a starting point. These Appendices  contain  a brief exposition of the theory  of L$\ity$ and A$\ity$ algebras and of the duality of 
differential associative algebras, that play an important role in our calculations. 
  
  In Appendix H we consider
  deformations of a $d$-dimensional reduction of
  ten-dimensional SYM theory  when
  $0\leq d\leq 10$.  We show that many statements obtained  for $d=0$ and $d=10$  can be generalized to any $d.$

 The present paper concludes the series of papers devoted to the analysis of deformations
 of  SYM theories \cite {MSch}, \cite {MSch2}, \cite {M3}, \cite {M4}, \cite {M4}. It contains a review of most important 
results of these papers as well as
 some new constructions.  The paper is mostly self-contained, but contain occasional references to the other papers in the series when it comes to proofs we intend to skip.
 
The supersymmetric deformations of maximally supersymmetric gauge theories were studied in numerous papers mostly in superspace approach or/and in pure spinor formalism (see, for example,\cite{BerkovitsHowe},\cite{CederwallNilssonTsimpis},\cite{CNT2},\cite{CNT3},\cite{CNT4}  It is not always easy to compare the results of these papers with our results, but when this comparison is possible the results agree.
 
 {\bf Acknowledgment}
 Both authors would like to thank  IHES (Bures-sur-Ivette) and MPIM (Bonn) and the first author wishes to thank MPIM (Leipzig)
   for hospitality and excellent working conditions.
 \section{Preliminaries}\label{S:Preliminaries}
\subsection {Basic algebras} \label{2.1}
By now it became a common place (see e.g. \cite{AFG},\cite{Howe},\cite{Berkovits}) that pure spinors can be used to formulate super Yang-Mills equations in dimension ten  in manifestly supersymmetric manner. In this section we define  objects relevant to such formulations. 

The algebra $\Ss=\bigoplus_{k\geq 0}\Ss_k$ has generators 
\begin{equation}\label{E:basisSdual}
\lambda^1,\dots,\lambda^{16}
\end{equation}
and relations 
\begin{equation}\label{E:pure}
\Gamma_{\alpha\beta}^i\lambda^{\alpha}\lambda^{\beta}=0,i=1,\dots,10
\end{equation}
where $\Gamma_{\alpha\beta}^i$ are ten-dimensional $\Gamma$-matrices (see \cite{Deligne} for mathematical introduction). $(\lambda^{\alpha})$ are coordinates on sixteen dimensional spinor representation in a basis $(\phi_{\alpha})$. Any spinor $\phi$ is equal to $\lambda^{\alpha}\phi_{\alpha}$. Occasionally we shall identify a spinor $\psi$ in a basis-dependent fashion with an array of coordinates $\lambda=(\lambda^1(\phi),\dots,\lambda^{16}(\phi))$.  $\Ss$ is an algebra of 
polynomial functions on the space $\mathcal{C}$ of pure  spinors (spinors obeying $\Gamma_{\alpha\beta}^i\lambda^{\alpha}
\lambda^{\beta}=0$, cf. \cite{Chevalley}). Components $\Ss_k$ of $\Ss$ are spaces of polynomial functions  of degree $k$. Projectivization of  $\mathcal{C}\backslash\{0\}$ , i.e. result of identification of proportional pure spinors in $\mathcal{C}\backslash\{0\}$ ,  is  an Isotropic Grassmannian $\mathcal{Q}$ or , which is the same, projective space of pure spinors.
The linear  space $\Ss_k$  is the space  of holomorphic sections of the line bundle $\O(k)$ over $\mathcal{Q}$ (see, for example,  \cite{GriffithsHarris} for explanation of standard notations of algebraic geometry). The group $\Spin (10)$ acts on $\Ss$ in natural way; it is easy to check that $\Ss_k$ is an irreducible representation of $\Spin (10)$ with Dynkin label $[0,0,0,k,0]$. Using this fact one can calculate the Poincar\'e series of $\Ss$:
$$\Ss(t)=\frac{1+5t+5t^2+t^3}{(1-t)^{11}}$$
(see \cite{BerNek}).

 The reduced Berkovits algebra $B_0$ is a differential graded commutative algebra. It is generated by  even $\lambda^
{\alpha}$ obeying pure spinor relations (\ref{E:pure}) and odd $ \psi^{\alpha}$. The   differential $d$ satisfies $d(\psi^
{\alpha})=\lambda^{\alpha}, d(\lambda^{\alpha})=0$.

One can also give a description of $B_0$ in terms of  functions on $\mathcal{C}$. Its elements are   polynomial functions 
depending on pure spinor $\lambda$ and odd spinor  $\psi$. We can interpret $\psi^{\alpha}$  as coordinates on odd 
spinor space $\Pi S$. The differential is represented by the odd vector field $\lambda^{\alpha}\pr{\psi^{\alpha}}$.

The (unreduced) Berkovits algebra $B$ can be defined as the algebra of polynomial functions of pure spinor $
\lambda$, odd spinor $\psi$ and $x\in \mathbb{R}^{10}$. 
Sometimes it is convenient to modify this definition considering an algebra $B^{\infty}$ consisting of functions that are polynomial in  $
\lambda$ and $\psi$ but smooth as  functions of  $x\in \mathbb{R}^{10}$. 
The differential is 
defined as the derivation 
\begin{equation}\label{E:differential}
\lambda^{\alpha}\bigg(\pr{\psi^{\alpha}}+\Gamma_{\alpha\beta}^i\psi^{\beta}\pr{x^{i}}\bigg).
\end{equation}

 $\Ss,B_0,B$ are quadratic algebras, i.e. they are described by generators obeying quadratic relations. The reader is referred to  \cite{Qalg} for a comprehensive account of  quadratic algebras  and references to the original publications. 
We shall be using construction of a quadratic dual algebra, which we remind the reader presently. Let $A=
\bigoplus_{i\geq 0}A_i$  be  a  unital not necessarily commutative  quadratic graded algebra $A$. We assume that $A_0=\mathbb{C}$. and that $A$ is generated by elements $w_1,\dots,w_n\in A_1$ obeying  quadratic relations 
$r_k=r_k^{ij}w_iw_j=0, k=1,\dots,l$. In more invariant terms we say that $A$ is generated by the  linear span $W=<w_1,\dots,w_n>=A_1$. 
Relations  $(r_k)$, in turn, span a linear subspace $R\subset W\otimes W$.    
The {\it quadratic  dual}  $A^!=\bigoplus_{i\geq 0}A^!_i$ is an algebra  generated by dual linear 
space $W^*=<w^{*1},\dots,w^{*n}>$, where $\langle w_{i},w^{*j}\rangle =\delta^j_i$. It relations are generating the 
subspace $R^{\perp}\subset W^*\otimes W^*$. In other words $R^{\perp}$ has a basis $s^m=\sum_{ij}s^m_{ij}w^{*i}w^{*j},m=1,\dots,n^2-l$. The matrices $s^m_{ij}$ form a 
basis in the space of solutions of the linear system $\sum_{ij}r_k^{ij}s_{ij}=0$ (when some of $w_j$ have different parities $|w_j|$ Koszul  sign rule is in force: $\sum_{ij}(-1)^{|w_j||w^{*i}|}r_k^{ij}s_{ij}=0$). In addition $A^{!!}=A$.

This duality preserves (graded) tensor product of algebras. The dual to polynomial algebra $\mathbb{C}[s]$ is the exterior algebra on one generator $\Lambda[\psi]$. A combination of these two facts gives an isomorphism $\mathbb{C}[s_1,\dots,s_k]^{!}\cong\Lambda[\psi^1,\dots,\psi^k]$.

Another interesting example is an algebra of polynomial functions on a quadric $A=\mathbb{C}[s_1,\dots,s_k]/s_1^2+\cdots+s^2_k$.
The dual algebra $A^!$ is generated by $\psi^1,\dots,\psi^k$ subject to relations $\psi^i\psi^j+\psi^j\psi^i=0$ and $(\psi^i)^2=(\psi^j)^2, i\neq j$. These are commutation relations in D=1,N=k supersymmetry Lie algebra. The element
\begin{equation}\label{E:H}
H=(\psi^i)^2
\end{equation} is central. The algebra $A^!$ is closely related to the Clifford algebra $Cl_n$.

Quadratic duality  has some useful properties in case of {\it Koszul algebras} which we about to define. 
 The tensor product  $A_1\otimes A^!_1=W\otimes W^*$ is a subspace of  $A\otimes A^!$. The tensor 
\begin{equation}\label{E:dif}
e=\sum_i w_i\otimes w^{*i}\in W\otimes W^* 
\end{equation} satisfies
$e^2=0$ and defines a differential on any left $A\otimes A^!$-module $K$:
\[d(m)=em,m\in K.\] The  module 
$K=A\otimes A^{!^*}$  contains a subspace $\mathbb{C}=A_0\otimes A^{!^*}_0 $ which generates nontrivial subspace in 
cohomology $H(A\otimes A^{!^*})$. Quadratic  algebra $A$ is a Koszul algebra  if this subspace exhausts the cohomology. 

For example if $A=\mathbb{C}[s_1,\dots,s_k]$, then $A\otimes A^{!^*}$ coincides with $\mathbb{C}[s_1,\dots,s_k]\otimes \Lambda[\xi_1,\dots,\xi_k]$, equipped with the acyclic  differential $s_i\frac{\sd}{\sd \xi_i}$. This is why the polynomial algebra is Koszul.
If $A=\mathbb{C}[s_1,\dots,s_k]/s_1^2+\cdots+s^2_k$, then $A^{!^*}$ coincides with $\Lambda[\xi_1,\dots,\xi_k]\otimes \mathbb{C}[h]$, The generators $h$ is dual to (\ref{E:H}). The differential $e$ in this case acts by 
\[s_i\frac{\sd}{\sd \xi_i}+s_i\xi_i\frac{\sd}{\sd h}\]
The differential is acyclic and $\mathbb{C}[s_1,\dots,s_k]/s_1^2+\cdots+s^2_k$ is Koszul. This differential is reminiscent of (\ref{E:differential}).

One of the properties of Koszul algebras is that $A$ is Koszul if  and only if $A^!$ is. The Poincar\'{e} series $A(t)=\sum_{i\geq 0}\dim A_i t^i$ and  $A^!(t)=\sum_{i\geq 0}\dim A^!_i t^i$ of Koszul, quadratically dual algebras  are related:
\begin{equation}\label{E:dualseries}
A(t)A^!(-t)=1
\end{equation}

Idea to use Koszul duality for classification of deformations of SYM has been  proven to be fruitful \cite{M4}. This is why we look at $\Ss,B_0, B$ from the perspective of quadratic duality.
Most of the facts pertinent to  this can be also found in   \cite{MSch} and  \cite{MSch2}.

An obvious starting point in this direction would be the calculation of  quadratic duals. Algebras $B_0, B$ carry a differential. Quadratic duals in our setting will be algebras with a differential. At the initial stage of computation we simply ignore the differentials.  After all the underlying algebras are found we 
take care of
 the differentials in dual pairs of objects. The differential in $B_0$ is linear. It defines a differential in $B_0^!$ in obvious way. The differential in $B$ is more complicated- it contains linear and quadratic parts. The linear part defines a differential in  $B^!$, the quadratic part defines a nontrivial commutator in $B^!$ (see below). The algebras $\Ss^!,B_0^!$ and $B^!$ are dual to $\Ss,B_0$ and $B$ in the sense of Definition \ref{D:dual} (Appendix B).

The use of the negative grading indices in linear spaces is unavoidable in a systematic quadratic-duality theory. In particular if all gradings of an algebra are  positive, the dual object has negative gradings. Negative indices from any other standpoints appear unnatural. Our attempt to reconcile these points of view  is to use the following convention :\[N_i=N^{-i}.\]

The Koszul dual algebra to $\Ss $  is a graded algebra $U(L)$ on generators $\bm{\theta}_{1},\dots, \bm{\theta}_{16}$ of 
degree  one subject to
\begin{equation}\label{E:nvpri}
\Gamma^{\alpha\beta}_{i_1,\dots,i_5}[\bm{\theta}_{\alpha},\bm{\theta}_{\beta}]=0.
\end{equation}  
The element $e$ (see (\ref{E:dif})) is $\lambda^{\alpha}\bm{\theta}_{\alpha}$.
The algebra $U(L)$ is a universal enveloping algebra of the graded Lie algebra 
\begin{equation}\label{E:ldef}
L=\sum_{n\geq1} L^n
\end{equation} that is defined by 
 generators $(\bm{\theta}_{\alpha})$  obeying
   the same relations (\ref {E:nvpri}). {\footnote { The reader should be warned that the currently used grading in $L$ is different from the  grading of the general theory outlined in Appendix \ref{AppendixB}: in the notations of Appendix  \ref{AppendixB} the generators of $L$ have degree minus one and generators of $\Ss$ have degree two.}}
(This is a particular case of general statement: the quadratic dual of commutative quadratic  algebra is a universal enveloping algebra of Lie algebra.)

 A basis $(\bm{\theta}_{\alpha})$ in $L^1=S$ is dual to the basis $(\lambda^{\alpha})$ (\ref{E:basisSdual}).  The Poincar\'e series 
\begin{equation}\label{E:PseriesUL}
U(L)(t)=\frac{(1+t)^{11}}{1-5t+5t^2-t^3}
\end{equation} follows from (\ref{E:dualseries}).  
 The second graded component $L^2$ of $L$ has dimension ten. This follows from the formula (\ref{E:PseriesUL}) for Poincar\'{e} series $U(L)(t)$.  It is easy to see that one can find a basis $\bm{D}_1,\dots,\bm{D}_{10}$ in $L^2$ 
obeying  

\begin{equation}\label{E:commutator}
[\bm{\theta}_{\alpha},\bm{\theta}_{\beta}]=\Gamma_{\alpha\beta}^i\bm{D}_i.
\end{equation}


$\Ss^!$ is significantly more complex than the Koszul dual to $\mathbb{C}[s_1,\dots,s_k]/s_1^2+\cdots+s^2_k$. Elements $\bm{D}_i\in L$ - the analogues of $H\in Cl_n$ are not central. 

The algebra $U(L)$ admits derivations $\frac{\sd}{\sd \theta_{\alpha}}$, that act by the formula 
\begin{equation}\label{E:diff}
\frac{\sd}{\sd \theta_{\alpha}}\theta_{\beta}=\delta_{\alpha}^{\beta}
\end{equation} and are compatible with the relations (\ref{E:nvpri})


 The dual algebra $B^!_0 $ to $B_0 $ is a tensor product $U(L)\otimes \mathbb{C}[s_1,\dots, s_{16}]$. The polynomial factor is dual to $\Lambda[\psi^1,\dots,\psi^{16}]$.  The adjoint  to $\lambda^{\alpha}\pr{\psi^{\alpha}}$ acts on the generating set $\{\bm{\theta}_{\alpha},s_{\alpha}\}$ in $B_0 ^!$ by the formula 
\begin{equation}\label{E:differentialred}
\bm{\theta}_{\alpha} \rightarrow s_{\alpha}
\end{equation}
 and defines a differential $d_{B^!_0}$ in $B^!_0 $.  The algebra $B^!_0 $ is a universal enveloping of a  Lie algebra $H$.
It is a direct sum $H=\bigoplus_{i\geq 0}H^i=L+S$, where $S=H^0=<s_{\alpha}>$ is an  abelian Lie algebra in degree  zero. The differential  (\ref{E:differentialred}) has degree minus one in  $H$. The duality in the sense of the definition 28 is established by means of the element $e=\lambda^{\alpha}\bm{\theta}_{\alpha}+\psi^{\alpha}s_{\alpha}$.

Let us introduce a Lie algebra $YM$ on even generators $\bm{D}_1,\dots,\bm{D}_{10}$ and odd $\bm{\chi}^{1},\dots,
\bm{\chi}^{16}$ obeying relations 
\begin{align}
\sum_{i=1}^{10}[\bm{D}_{i},[\bm{D}_{i},\bm{D}_{m}]]- \notag \\
&-\frac{1}{2}\sum_{\alpha\beta=1}^{16}\Gamma_{\alpha\beta}^m[\bm{\chi}^{\alpha},\bm{\chi}^{\beta}]=0 \quad m=1,\dots,
10 \label{E:relssfs1}\\
\sum_{\beta=1}^{16}\sum_{i=1}^{10}\Gamma_{\alpha\beta}^i[\bm{D}_i,\bm{\chi}^{\beta}]=0\quad  \alpha=1\dots 16 
\label{E:relssfs3}
\end{align}

The relations give a formal algebraic abstraction of  equations of motion of D=10 SYM  theory reduced to a point. An $N$-
dimensional representations of the algebra $YM$ gives a classical solution of  the  reduced SYM theory (of IKKT model ). 
  It is easy to construct a homomorphism of the Lie algebra $YM$ into $L$ (or, more precisely, into $\bigoplus_{k\geq 2}
L^k$).  Namely, we should send its generators into $\bm {D}_i$, defined by (\ref {E:commutator}) and into $\bm{\chi}^
{\beta}$ defined by the formula
\begin{equation}
\label{cc}
\Gamma _{\alpha \beta i}\bm{\chi}^{\beta}=
[\bm {\theta} _{\alpha}, \bm {D}_i]
\end{equation} 
It turns out that 
\begin{equation}\label{E:superchi}
\bm{\theta}_{\alpha}\bm{\chi^{\beta}}=\Gamma_{\alpha}^{\beta ij}[\bm{A}_i,\bm{A}_j]
\end{equation}
in the Lie algebra $L$.
\begin{proposition}
The algebra $YM$ is isomorphic to $\bigoplus_{k\geq 2}L^k$.  The obvious map  $\bigoplus_{k\geq 2}L^k\rightarrow (L
+S,d)$ is a quasi-isomorphism. Similarly $U(YM)$ is quasi-isomorphic to $B^!_0$.
\end{proposition}

Recall that a homomorphism of differential algebras (modules) is called a {\it quasi-isomorphism} if it induces 
an isomorphism in homology.

Let us denote by $\theta_{\alpha}$ the derivation of the algebra $YM$  acting on generators by the formula (\ref{cc},\ref{E:superchi})  The derivations $\theta _{\alpha}$ can be interpreted as supersymmetry transformations.
It is easy to check that they can be interpreted as commutators with generators of $L$:
$$\theta _{\alpha} x=[\bm{\theta}_{\alpha},x].$$

  The dual  to $B $ is the universal enveloping algebra of  a graded semi-direct product $L\ltimes \Lambda$ of $L$ with  abelian  $\Lambda$.  The generators $s_{1},\dots,s_{16} $, 
$\varsigma_{1},\dots ,\varsigma_{10}$ of $\Lambda$ have degree zero and  one respectively. The nontrivial commutation 
relations between $L$ and $\Lambda$ are \[[\bm{\theta}_{\alpha},s_{\beta}]=\Gamma_{\alpha\beta}^i\varsigma_i.\]
The action of the nontrivial components of differential are given by the following formulas  \[d(\bm{\theta}_{\alpha})=s_{\alpha},d(\bD_i)=\varsigma_i.\]
The  duality is established by means of the element $e=\lambda^{\alpha}\bm{\theta}_{\alpha}+\psi^{\alpha}s_{\alpha}+x^i\varsigma_i$.


It is clear that ${\bm F}_{ij}=[\bD_i,\bD_j]$ and 
$\bchi ^{\alpha}$  belong to $TYM\overset{\ddef}{=}\bigoplus_{i\geq 3} L^i\subset L$. Moreover, they generate $TYM$ as an ideal of $YM$. More precisely, as an algebra 
$TYM$ is generated 
 by expressions $\nabla _ {i_1}\cdots \nabla _{i_n}\Phi$
 where $\Phi$ is  either  ${\bm F}_{kl}$ or $ \bchi ^{\alpha}$ and  $\nabla _i (x)= [\bD_i,x].$ In the framework of ten 
dimensional Yang-Mills theory
 we can interpret these expressions as covariant derivatives of  field strength and spinor field.
 Thus the elements of $U(TYM)$ are  algebraic abstractions of  gauge covariant local expressions.
 
\begin{proposition}\label{P:proposition3}
Inclusion  $TYM \subset (\tilde L,d)$ is a quasi-isomorphism. Likewise $U(TYM)$ and  $B^!
$ is a pair of quasi-isomorphic algebras.
\end{proposition}

One can prove that all quadratic algebras we use are Koszul algebras. 

Notice $YM$ and $TYM$ are ideals in $L$, therefore  elements $l\in L$ specify derivations $\alpha_l(x)=[l,x]$ on $YM$ and $TYM$.  This allows us to realize elements of $L$ as vector fields on the  space $Sol=Sol_N$ of solutions  of Yang-Mills equations with the gauge group $\U(N)$. (Notice that we do not identify gauge equivalent solutions.) More precisely, $L$ is realized  is a Lie subalgebra of the LIe algebra of vector fields generated by supersymmetries.  The universal enveloping algebras $U(TYM)$,  $U(YM)$ and  $U(L)$ become associative subalgebras in algebras of differential operators $Diff$ on the space of solutions.

\subsection {Calculation of the Lie algebra cohomology}\label{S:wweqe}
Computation of the Hochschild and the Lie algebra cohomology {\footnote {Definitions   of Lie algebra cohomology and
Hochschild cohomology  are given in Appendix A.  Appendix B contains a sketch of the proof of the results formulated in  present section. The book  \cite{Qalg}  contains a detailed treatment of the material presented in this section.  The book  \cite{Weibel} is a  modern elementary introduction to homological algebra that could be a good starting point for nonspecialists. }}
is seldom done with the standard complexes . The reason is that the spaces of chains  in these complexes are extremely big  and redundant . In practical computations we are interested in more manageable complexes, that  still have the same cohomology. Koszul duality theory provides us with very economic complexes. The appropriate constructions will be spelled out in this section.

   Let us consider  a graded commutative Koszul algebra $\C$ and its dual algebra $\C^!=\bigoplus_{i\geq 0}\C^!_{-i}=\bigoplus_{i\geq 0}\C^{! i}=U(\g)$ where $\g$ is a graded Lie 
algebra. Let $N=\bigoplus_{i\geq 0}N_{-i}=\bigoplus_{i\geq 0}N^{ i}$ be a graded $\g$-module (representation of $\g$). 

\begin{proposition}\label{P:tqydxc1iww} \cite{Qalg}
The  cohomology $\rH^{\bullet}(\g,N)$ is equal to the  cohomology of the complex $N_c\overset{\ddef}=N\otimes \C$ (the 
$\C$-grading  defines the  cohomological grading in the tensor product). The subscript $c$ in $N_c$ stands for cohomology. The differential $d$ is defined by the left multiplication on 

\begin{equation}\label{E:uwhswtyrt1}
e=w^{*i}\otimes w_{i}\in \C^{!1}\otimes \C_1 \subset \C^{!}\otimes \C
\end{equation}
The basis elements  $w_i \in \C_1$ act
on $\C$ by means of multiplication from the left and the action of the elements of the dual basis $w^{*i}\in (\C^1)^*\subset 
\C^!$ is defined by means of representation of $\g$ on $N$.

The subspaces $N^{\bullet}_{c\  \bm{m}}=\bigoplus_{i+j=\bm{m}}N_j\otimes \C_i$ are $d$-invariant {\footnote {To avoid a 
possible confusion of cohomological and internal homogeneous grading  we reserve  the bold index for the latter.}} and 
 $N^{\bullet}_{c}$ is the direct sum $\bigoplus_{\bm{m}}N^{\bullet}_{c\  \bm{m}}$ of subcomplexes.

  The component $\rH^{k,\bm{m}}(\g,N)$ of $k$-th cohomology group of homogeneity   $\bm{m}$ coincides with   $\rH^k(N^
{\bullet}_{c\ \bm{m}})$.
\end{proposition}
There exists a similar statement for Lie algebra homology. The complex $N^{\bullet}_{h}=N\otimes \C^*$ is the direct sum 
of subcomplexes $N^{\bullet}_{h}=\bigoplus_{\bm{m}}N^{\bullet}_{h\ \bm{m}}$. The homological grading on $N^{\bullet}_
{h\ \bm{m}}$ is defined as follows:
\begin{equation}\label{E:homol2}
\begin{split}
&N^{\bullet}_{h\ \bm{m}}=(N_{m_0}\otimes \C^*_{m-m_0}\overset{d^{m_0}}\rightarrow \dots N_{0}\otimes \C^*_{m}
\overset{d^0}\rightarrow  \dots N_{m-1}\otimes \C^*_{1} \overset{d^{m-1}}\rightarrow N_{m}\otimes \C^*_{0})\\
&(N\otimes \C^*)_{\bm{m}}= N_{m}\otimes \C^*_{0}\leftarrow N_{m-1}\otimes \C^*_{1}\leftarrow \cdots
\end{split}
\end{equation}

\begin{proposition}\label{P:iqwwstiwww2} \cite{Qalg}
There is an isomorphism $\rH_{\bullet}(\g,N)\cong \rH_{\bullet}(N\otimes \C^*)$ 
and its refinement \[\rH_{k,\bm{m}}(\g,N)\cong \rH_k((N\otimes \C^*)_{\bm{m}}).\]

\end{proposition}

Let $\C$ be $\Ss$, $\C^!$ be $U(L)$ and $N$ be an $L$-module. Propositions \ref{P:tqydxc1iww},\ref{P:iqwwstiwww2}  give us  alternative complexes for computation of    $\rH^{\bullet}(L,N)$, $\rH_
{\bullet}(L,N)$.
\begin{corollary}\label{P:tqydxc1} \cite{Qalg}
The  cohomology $\rH^{\bullet}(L,N)$ is equal to the  cohomology of the complexes $N_c\overset{\ddef}=N\otimes \Ss$. 
The differential is a multiplication by
\begin{equation}\label{E:e}
e=\lambda^{\alpha}\bm{\theta}_{\alpha}.
\end{equation} 
The cohomological grading coincides with the grading of $\Ss$-factor. The total degree is preserved by $d$.
 The complex $N_c$ splits according to degree: 
\begin{equation}\label{E:fdgisevi}
\begin{split}
&N^{\bullet}_{c\  \bm{m}}=(N_{m}\otimes \Ss_0\rightarrow N_{m+1}\otimes \Ss_1 \rightarrow\dots )\\
&(N\otimes \Ss)^{\bm{m}}= N^{m}\otimes \Ss_0\rightarrow N^{m+1}\otimes \Ss_1 \rightarrow\cdots  
\end{split}
\end{equation}
The complex $N^{\bullet}_{c\  \bm{m}}$ is defined for positive and negative $\bm{m}$, we assume that $N_m=0$ if 
$m<m_0$.
Then $\rH^{k,\bm{m}}(L,N)=\rH^k(N^{\bullet}_{c\ \bm{m}})$.
\end{corollary}
There is also a degree decomposition in homology $\rH_k(L,N)=\bigoplus_{\bm{m}} \rH_{k,\bm{m}}(L,N)$.
\begin{corollary}\label{P:iqwwst2} \cite{Qalg}
The homology $\rH_{\bullet}(L,N)$ is equal to the cohomology of the complex $N_h\overset{\ddef}=N\otimes \Ss^*$. 
The space $\Ss^*=\bigoplus_{n\geq 0} \Ss^*_n$ is an $\Ss$-bimodule dual to  $\Ss$. 
The differential is a multiplication by $e$ (\ref{E:e}). 
The homological degree coincides with the grading of  $\Ss^*$-factor. 
The complex $N_h$  splits :
\begin{equation}\label{E:homol21}
\begin{split}
&N^{\bullet}_{h\ \bm{m}}=N_{m_0}\otimes \Ss^*_{m-m_0}\overset{d}\rightarrow \dots N_{0}\otimes \Ss^*_{m}\overset{d}
\rightarrow  \dots N_{m-1}\otimes \Ss^*_{1} \overset{d}\rightarrow N_{m}\otimes \Ss^*_{0}
\end{split}
\end{equation}

and $\rH_{k,\bm{m}}(L,N)=\rH^{m-k}(N^{\bullet}_{h\ \bm{m}})$.

\end{corollary}
 
The Propositions \ref{P:proposition3},\ref{P:tqydxc1iww} are particular cases of more general statements formulated in terms of Hochschild cohomology and homology (see Appendix \ref{AppendixB} ,  Propositions \ref{P:tqydxc},\ref{P:iqwwst})

\subsection {The group $\Spin(10, \mathbb{C})$ and the space of pure spinors} \label {2.3}

The complex group $\Spin(10, \mathbb{C})$ acts transitively on the projective space of pure spinors  $\mathcal{Q} $ ; 
the stabilizer   of a point is a parabolic subgroup and parabolic subgroup corresponding to different points are conjugated.  To describe the
Lie algebra $\mathfrak{p}$ of one of such parabolic subgroups $P $ we notice that the Lie algebra $
\mathfrak{so}(10, \mathbb{C}) $ of $\SO(10, \mathbb{C}) $ can be identified
with ${\Lambda}^2V $. This is the same as a  the space of antisymmetric tensors ${\rho}_{a
b} $ where $a, b=1, \dots, 10$, if a  basis $(v_1, \dots, v_{10})$ in $V$ is fixed. The vector representation $V$ of $\SO(10, 
\mathbb{C}) $ splits upon restriction to $\GL(5, \mathbb{C}) \subset \SO(10, \mathbb{C})$
into the direct sum $W+W'$  of vector and
representation $W'$ isomorphic to covector representations $W^*$. We shall not make a distinction between $W'$ and $W^*$.
The Lie algebra of $\SO(10,\mathbb{C})$ can be decomposed into a sum $\Lambda^2 W + \mathfrak{p} $ of vector spaces.  Subspace $\mathfrak{p }=W \otimes W^* + \Lambda^2W^* $
is the Lie algebra of $P$. 
Suppose that $(v_1, \dots, v_{5})$ is a basis $W$ and $(v_6, \dots, v_{10})$ in $W^*$.
Using the language of generators we
can say that the Lie algebra $\mathfrak{so}(10, \mathbb{C}) $ is generated by
skew-symmetric tensors $m_{ab}, n^{ab} $ and by $k_a^b $ where $
a, b=1, \dots, 5 $. These are blocks of ${\rho}_{a
b} $  defined by the chosen partition of the basis. The subalgebra $\mathfrak{p}$ is generated by $k_a^b $
and $n^{ab} $. Corresponding commutation relations are 
\begin{equation}\label{E:adjoint}
\begin{split}
& [m, m^{\prime}]=[n, n^{\prime}]=0 \\
&[m, n]_a^b=m_{ac}n^{cb} \\
&[m, k]_{ab}=m_{ac}k_b^c+m_{cb}k_a^c \\
&[n, k]_{ab}=n^{ac}k_c^b+n^{cb}k_c^a
\end{split}
\end{equation}
Homomorphism $ \Spin(10, \mathbb{C})\rightarrow \SO(10)$ defines cover in two sheets. Let $\widetilde{P}$ be the preimage of $P$ in $ \Spin(10, \mathbb{C})$ under this projection. The group $\widetilde{P}$ contains a double sheeted cover $\widetilde{\GL}(5, \mathbb{C})$ of $\GL(5, \mathbb{C})$.
There exists one-to-one correspondence between $\Spin(10, \mathbb{C})$-homogeneous holomorphic vector bundles over $
\mathcal{Q}$ and complex representations of $\widetilde{P}$ (see e.g. \cite{bott}) in the fiber over the $P$-fixed point. 
 A  one-dimensional representation of $\widetilde{P}$ corresponding to the line bundle 
$\O (k)$ over $\mathcal{Q}$ will be denoted $\mu
_k$. In this representation $g\in \widetilde{\GL}(5, \mathbb{C})$ acts by multiplication on  $\sqrt{\det}^k(g)$.

 The space of spinors can be 
embedded into Fock space $\mathcal{F}$(see \cite{Chevalley} for mathematical account).
This Fock space is a representation of canonical anti-commutation relations $a_ia_j+a_ja_i=0, a^*_ia^*_j+a^*_ja^*_i=0, a_ia^*_j+a^*_ja_i=\delta _{ij}
$ $i,j=1,\dots,5$. The cone $\mathcal{C}\backslash\{0\}$  can be realized as the orbit of Fock vacuum with respect to the action of the group of linear 
canonical transformations
(transformations preserving anti-commutation relations), that preserve chirality.  For every vector $x \in \mathcal{F}$ 
we consider the subspace $W^*(x)$ of the space $V$ of linear combinations $A= \sum \rho ^ia_i+
\sum \tau ^j a^*_j$ obeying $Ax=0$. For $x\in \mathcal{Q}$
the subspace $W^*(x)$ is five-dimensional. The subspaces $W^*(x)$ specify a $\Spin(10)$-invariant vector bundle over 
$\mathcal{Q}$ that will be denoted by $\W^*$;  corresponding representation of $P$ will be denoted by $W^*$.
The bundle over $\mathcal{Q}$ having fibers $V/W^*(x)$ will be denoted by $\W$; corresponding representation of $P$ will be 
denoted by $W$. 

Notice that $\Spin(10)$-representation contents of first two components of $L$ is 
\begin{align}
& L^1=[0,0,0,1,0]\label{E:repa}\\
& L^2=[1,0,0,0,0]\label{E:repb}\\
& L^3=[0,0,0,0,1]\label{E:repc}\\
& L^4=[0,1,0,0,0]\label{E:repd}\\
&\dots
\end{align}

And this is how they split as  $\widetilde{\GL}(5)$-representations :
\begin{align}
&L^1=\mu _{-1} +\Lambda^2(W)\otimes \mu _{-1}+\Lambda^4(W)\otimes \mu _{-1}
\label{E:decompa}\\
&L^2=W^*+W  \label{E:decompb}\\
&L^3=\Lambda^4(W^*)\otimes \mu _1+\Lambda^2(W^*)\otimes \mu _1+\mu _1\cong\notag\\
&\cong W\otimes \mu_{-1}+\Lambda^3(W)\otimes \mu_{-1}+\mu _1
\label{E:decompc}\\
&L^4=\Lambda^2(W)+\Lambda^2(W^*)+W\otimes W^*\label{E:decompd}\\
&\dots
\end{align}
The above formulas are written in such a way that the leftmost  summand in every line is a representation of $\widetilde{P}$; the same is 
true for the sum of first two leftmost summands.

\subsection {Euler characteristics}\label{S:Euler}

 Statements formulated in  Section \ref{S:wweqe} permit us to calculate the Euler characteristics of $\rH^{\bullet}(L,N)$ and $\rH_{\bullet} 
(L,N)$. 

Suppose  a complex of vector spaces
$K=\sum K^k$  has  finite number of finite-dimensional cohomology groups $\rH=\sum \rH^k $. Then  the Euler characteristic  can be defined as $\chi =\chi (\rH)=\sum (-1)^k \dim \rH^k$. If  $K$ has 
finite number of  finite-dimensional  components $K^k$ then   
\begin{equation}
\label{ch}
\chi (\rH)=\sum (-1)^kd_k
\end{equation}  where $d_k=\dim K^k.$

A Lie group $G$ action on  complex $K^i$  by operators $T_{K^i}(g), g\in G$ descends  to  
cohomology. The Euler characteristic $\chi_K(g)$ is defined as the alternating sum of characters $\sum (-1)^k \tr T_{\rH^k}(g)$.
The relation (\ref {ch}) remains correct in this more general situation after appropriate modifications, namely
$d_k$ gets replaced by $\tr T_{K^i}(g)$ .{\footnote { Using the notion of virtual representations ( see Appendix I) we can define the Euler characteristic of  graded representation $K=\sum K^i$ as an alternating sum $\chi_ K=\sum (-1)^iK^i.$}}

A bit of terminology:  if a group $G$ acts on a graded vector space $A=\bigoplus_{i\geq n_0}A_i$  by linear automorphisms $T_{A_i}(g)$ then $G$-equivariant  Poincar\'{e} series  is defined as a generating function of characters:
$$A(g,t)=\sum_{i\geq n_0}\tr T_{A_i}(g)t^i.$$  (In the case of trivial group action this definition is equivalent to the standard definition:  $A(t)=\sum_{i\geq n_0}\dim A_it^i$.) 

In our applications module $N$ over Lie algebra $L$ is a graded  $\Spin (10)$ representation. The  $\Spin (10)$-action on $N$ is compatible with the action  on $L$. 
We will give an expression of  $\chi _{\rH^{\bullet} (L,N)}(g,t)$ and of $\chi_{ \rH_{\bullet} 
(L,N)}(g,t)$,  in terms of   Poincar\'e series $N(g,t)$, $\Ss(g,t)$.  

Let us sketch the derivation in more general case when $N$ is a graded module over graded Lie algebra $\g$ and a group $K$ acts on $\g$ and $N$ in compatible way. (The Poincar\'e series of $\g$ and $N$  and Euler characteristic of $\rH^k(\g, N)$  depend on $g\in K$ and on $t.$)

The Lie algebra cohomology $\rH^k(\g, N)$ coincides with Hochschild cohomology $\rH^k(U(\g),N)$ and therefore can be computed with the normalized Hochschild complex $C^k(U(g),N)=\Hom(I^{\otimes k},N)$ where
$I$ is the set of elements of the universal enveloping $U(\g)$ without the constant  term.  Formally we can express the Euler characteristic in terms of  Poincar\'e series for $N$ and $U(\g)$:
\[\chi_{\rH^{\bullet}(\g, N)}=\frac{N(g,t)}{U(\g)(g^{-1},t^{-1})}\]
In the derivation of this relation we use the formula 
$$\frac{1}{u}=\frac{1}{1+(u-1)}=\sum (-1)^k(u-1)^k.$$
Likewise homology $\rH_k(\g, N)$ can be computed with the normalized Hochschild complex $C_k(U(g),N)=I^{\otimes k}\otimes N$ and 
\[\chi_{\rH_{\bullet}(\g, N)}=\frac{N(g,t)}{U(\g)(g^{},t^{})}\]

For us this means that 
\begin{equation}\label{E:ourcharacters}
\begin{split}
&\chi_{\rH^{\bullet}(L, N)}(g,t)=\frac{N(g,t)}{U(L)(g^{-1},t^{-1})},\\
&\chi_{\rH_{\bullet}(L, N)}(g,t)=\frac{N(g,t)}{U(L)(g,t)}.
\end{split}
\end{equation}

The expression for the Poincar\'e series ${U(L)(g,t)}$ will be derived at the end of this section.

We have mentioned already that the $n$-th graded component of $\Ss$ is an irreducible representation of $\Spin (10)$ with Dynkin label $[0,0,0,n,0]$. Using this fact one can express $\Ss(g,t)$ in terms of the character $V(g)$ of vector representation $V$ and the character $S(g)$ of the spinor representation $S$ labelled by $[0,0,0,1,0]$. The character of the algebra   $\Ss$ is given by the formula
\begin{equation}\label{E:characterSs}
 \Ss(g,t)=(1-V(g)t^2+S(g^{-1})t^3-S(g^{})t^5+V(g)t^6-t^8)\Sym S(g,t).
 \end{equation}
There are several ways to prove this important formula.
The  homological approach of \cite{CortiReid} is based on construction of a resolution of $\Ss$ as $\Sym S$-module. Another way of tackling the problem is to use fixed point formula \cite {BerNek}. Standard monomial theory (see  \cite{StrMov} ) can be utilized to arrive at the same result.
The information about $\Ss$ permits us to
analyze the structure of Koszul dual algebra $U(L)$ using the following general statement.

 \begin{proposition}\label{P:genfun}
 Let $A$ be a Koszul algebra, equipped with an action of a group $G$. Then the group $G$ also acts on $A^!$ and there is an equality 
 \begin{equation}\label{E:Koszuldual}
 A(g,t)A^!(g^{-1},-t)=1
 \end{equation}
  \end{proposition}
\begin{proof}
It is a trivial adaptation of the proof  \cite{Priddy} to  the case of algebra with $G$-action.  
The complex $A\otimes (A^!)^*$ has trivial cohomology  by 
the definition of Koszul algebra. It decomposes into a direct sum of acyclic complexes $K_n=\bigoplus_{i+j=n} A_i\otimes 
(A^{!}_j)^*$. The generating function of Euler characteristics of $K_n$ is equal to the constant function $1$. But it also 
equals to the product of the generating functions $A(g,t)A^!(g^{-1},-t)$. (We use here the fact that the character of the 
dual representation is expressed in terms of the character $\rho (g)$ of original representation as $\rho (g^{-1})$.)
\end{proof}

Note that $V(g)=V(g^{-1})$, because the ten-dimensional tautological  representation of $\SO(10)$ is self-dual. For spinors we have  $S^*(g)=S(g^{-1})$.
\begin{corollary}
The Poincar\'e series $U(L)(t,g)$ of the universal enveloping $U(L)$ is equal to 
\[U(L)(g,t)=\frac{\Lambda S^* (g,t)}{1-V(g^{-1})t^2-S(g^{-1})t^3+S(g)t^5+V(g)t^6-t^8}\]
\end{corollary}
\begin{proof}
Apply  Proposition \ref{P:genfun} to  algebras $\Ss$, $\Ss^!\cong U(L)$   and $\Sym(S)^{!}\cong \Lambda
(S^*)$.
\end{proof}
Notice that instead of working  with characters we could work with virtual representations and representation rings (see Appendix I). Then the above formula looks nicer:
$$U(L)=\frac{\lambda_t (S^*)}{1-Vt^2-S^*t^3+St^5+Vt^6-t^8.}$$
\section{Infinitesimal SUSY Deformations of $\L_{SYM}$}\label{S:fdgdfgmjl}

Let us consider an infinitesimal deformation $\dL$ of a Lagrangian $\L$. If an infinitesimal deformation $\ddL$ is obtained 
from $\dL$ by means of a field redefinition then  the action functionals corresponding to infinitesimal
deformations $\dL,\ddL$ coincide on solutions of EM for $\L$.  The converse statement is also true. 
Therefore we will identify infinitesimal deformations $\dL$ and $\ddL$ of $\L_{SYM}$ if   $\ddL=\dL$  on 
the solutions of EM for $\L_{SYM}.$
{\footnote {To reach a better understanding of the above statements we will discuss a finite-dimensional analogy.

Any function
\begin{equation}\label{E:pgerm}
f:\mathbb{C}^n\rightarrow \mathbb{C}
\end{equation}
 can be deformed by adding an arbitrary function $g$ multiplied by an infinitesimal parameter. It is not true however that 
the space of deformations of $f$ coincides with the space $\hat\O$ of all $g$. The reason is that there are trivial 
deformations of $f$ obtained by a change of parametrization of $\mathbb{C}^n$. A vector field  $\xi$ on $\mathbb{C}^n$ 
defines an infinitesimal change of coordinates, under which $f$ transforms to $f_{\xi}=\xi ^i \frac{\partial f}{\partial xi}$. 
The space $Vect f$ of functions $f_{\xi}$ forms a subspace of  trivial infinitesimal  deformations. The quotient $\hat\O/
Vect f$ is the formal tangent space to the space of nontrivial deformations of  (\ref{E:pgerm}).

Under some conditions of regularity one can identify $\hat\O/Vect f$ with the algebra of functions on the set of critical 
points of $f$. (If this set is considered as a scheme the conditions of regularity are not necessary.)

In field theories this identification corresponds to identification of off-shell classes of deformations of an action functional 
with the deformations considered on shell (on the solutions of EM). }}

We will be interested in deformations of SYM that are defined simultaneously for all gauge groups $\UN$.  Let us 
consider first the  Lagrangian $\L_{SYM}$ reduced to a point.  The deformation of the kind we are interested in   are 
single-trace deformations: they can be represented in  the form $\tr \Lambda $, where $\Lambda$ is an arbitrary 
 non-commutative polynomial  in terms of the fields of the reduced theory.  The fields form an array of $N\times N$  matrix variables  $A_1,\dots,A_{10},\chi^{1},\dots,\chi^{16}$  of suitable parity in the theory with the gauge group 
$\UN$. Reality conditions are left out of scope of  our analysis, and we simply consider fields as complex 
matrices , i.e. as elements  of $\Mat_N$. ( This means that the gauge group $\UN$ is  extended to its complexification $GL(N)$.)   We are working with all of these groups simultaneously, hence we consider the fields as formal 
non-commuting variables, i.e. as generators of free graded associative algebra .  More precisely, we consider free graded associative algebra $\mathcal{A}$ generated by symbols of  fields 
\begin{equation}\label{E:variables}
\bm{D_i}, \bm{\chi}^{\alpha}.
\end{equation}

The space of single-trace deformations can be identified with $\mathcal{A}/[\mathcal{A},\mathcal{A}]$ .  {\footnote { We can identify this space also with the space of cyclic words in the alphabet where letters 
correspond to the fields.}}  (A non-commutative polynomial of matrix variables  $A_1,\dots,A_{10},\chi^{1},\dots,\chi^{16}$ can be considered as an element of $\mathcal{A}$. However, the (super)trace of (super) commutator vanishes, therefore the trace of this polynomial can regarded as an element of $\mathcal{A}/[\mathcal{A},\mathcal{A}]$ .)The natural map  $\mathcal{A}\to\mathcal{A}/[\mathcal{A},\mathcal{A}]$ will be denoted by $\tr$ .) However, we should take into account that the deformations can be equivalent (related by a 
change of variables). As we have seen this means that the deformations coincide  on shell. In other words,  the 
space of equivalence classes of deformations can be identified with $U(YM)/[U(YM),U(YM)]$ where  $U(YM)$ can be 
interpreted
as an associative algebra generated by the fields of SYM theory reduced to a point  with  relations coming from the 
equations  of motion (see Section \ref{S:introduction} and Section \ref{S:Preliminaries} for more detail). 
\footnote{This space has also interpretation in terms of Hochschild  homology $ \rH_0(U(YM),U(YM)) $ or in terms of cyclic homology.} 

Similar results are true for non-reduced SYM theory.   In this case we consider the Lagrangian as a gauge invariant local expression (a trace of gauge covariant local expression). We are working with all groups $\UN$ simultaneously, hence we are writing a Lagrangian in the form $\tr H$ where $H\in \U(TYM)$  and $\tr$ stand for the natural homomorphism $\tr : \U(TYM)\to [ \U(TYM),  \U(TYM)].$ Notice, that in this case the Lagrangian is defined up to a total derivative.{\footnote {We say that a function on $
\mathbb{R}^n$ is a total derivative if it can be represented in the form $\pr{x^i}H^i$. In more invariant way one can say 
that the differential form of degree $n$ corresponding to this function should be exact.  
Saying that Lagrangian is defined up to a total derivative we have in mind that adding $\pr{x^i}H^i$ where $\rH^i$ are gauge invariant local expressions we obtain an equivalent Lagrangian.}}  

The supersymmetry transformations $\theta_{\alpha}$ act in natural way on all spaces we have considered.

We are saying that infinitesimal deformation $\dL$ of reduced SYM theory is supersymmetric if $\theta_{\alpha}\dL$  vanishes  
 on equations of motion of $\L_{SYM}$.  For non-reduced theory a  deformation is supersymmetric if $\theta_{\alpha}\dL$  is a total derivative 
 on equations of motion of $\L_{SYM}$.

 Poincar\'e 
invariance is defined in a similar way.

There exists an infinite number of infinitesimal super Poincar\'e  invariant deformations. Most of them are given by a 
simple general formula  below, but  for the non-reduced theory there are three  exceptional deformations which do not fit into this formula. The first  
was discussed earlier in \cite{Berg}. 

\begin{equation}\label{E:dvsdjvh}
\begin{split}
&\dL_{16}(\nabla,\chi)=\tr\Bigg(  \frac{1}{8} F_{mn}F_{nr}F_{rs}F_{sm}- \frac{1}{32} \left( F_{mn}F_{mn}\right)^2\\ 
&+i\frac{1}{4} \chi^{\alpha} \Gamma_{m\alpha \beta} ({\nabla}_n\chi^{\beta}) F_{mr}F_{rn}  \\
&-i\frac{1}{8} \chi^{\alpha} \Gamma_{mnr\alpha \beta} ({\nabla}_s\chi^{\beta}) F^{mn} F_{rs}\\
&+ \frac{1}{8} \chi^{\alpha} \Gamma^m_{\alpha \beta} ({\nabla}_n\chi^{\beta})\chi^{\Gamma} \Gamma_{m\Gamma\delta} 
({\nabla}_n\chi^{\delta}) \\
&- \frac{1}{4} \chi^{\alpha} \Gamma^m_{\alpha\beta} ({\nabla}_n\chi^{\beta})\chi^{\gamma} \Gamma_{n\gamma\delta} 
({\nabla}_m\chi^{\delta}) \Bigg)
\end{split}
\end{equation}
It is convenient to introduce a grading on the space of fields (\ref{E:cvcjdfj}). We suppose that grading is multiplicative and $
\deg(\nabla_i)=2$, $\deg(\chi^{\alpha})=3$. This grading is related to the grading  with respect to $\alpha'$, that comes 
from string theory,  by the formula 
\begin{equation}
\label{alp}
\deg_{\alpha'}=\frac{\deg-8}{4}.
\end{equation} 

Subscript in $\dL_{16}$ in the formula (\ref{E:dvsdjvh}) stands for  the grading of infinitesimal Lagrangian. Lagrangian $
\dL_{16}$ is a  super Poincar\'e invariant deformation  of  lowest possible degree.\footnote{We will  treat the  truly lowest 
order deformation $\dL=\L_{SYM}$ as trivial.} The next linearly independent infinitesimal  super Poincar\'e invariant 
deformation (of degree 20) was found in \cite{Collinucci}. It has the following Lagrangian

\begin{eqnarray}\label{E:dvsdjvh111}
\dL_{20}(\nabla,\chi) &=& f^{XYZ}f^{VWZ}\bigg[2\,F_{ab}{}^{X}F_{cd}{}^{W}
\nabla_eF_{bc}{}^{V}\nabla_eF_{ad}{}^{Y}
-2\,F_{ab}{}^{X}F_{ac}{}^{W}
\nabla_dF_{be}{}^{V}\nabla_dF_{ce}{}^{Y}
\nonumber\\
&&\qquad +F_{ab}{}^{X}F_{cd}{}^{W}
\nabla_eF_{ab}{}^{V}\nabla_eF_{cd}{}^{Y}
\nonumber\\
&&\qquad -4\,F_{ab}{}^{W}
\nabla_cF_{bd}{}^{Y}\chi^{\alpha X}\Gamma_{\alpha \beta a}\,\nabla_d\nabla_c\chi^{\beta V}
-4\,F_{ab}{}^{W}
\nabla_cF_{bd}{}^{Y}\chi^{\alpha X}\Gamma_{\alpha \beta d}\,\nabla_a\nabla_c\chi^{\beta V}
\nonumber\\
&&\qquad +2\,F_{ab}{}^{W}
\nabla_cF_{de}{}^{Y}\chi^{\alpha X}\Gamma_{\alpha \beta ade}\,\nabla_b\nabla_c\chi^{\beta V}
+2\,F_{ab}{}^{W}
\nabla_cF_{de}{}^{Y}\chi^{\alpha X}\Gamma_{\alpha \beta abd}\,\nabla_e\nabla_c\chi^{\beta V}
\bigg]\ +
\nonumber
\end{eqnarray}
\begin{eqnarray}
&& +\ f^{XYZ}f^{UVW}f^{TUX}\bigg[4\,
  F_{ab}{}^{Y}F_{cd}{}^{Z}F_{ac}{}^{V}F_{be}{}^{W}F_{de}{}^{T}
+2\,F_{ab}{}^{Y}F_{cd}{}^{Z}F_{ab}{}^{V}F_{ce}{}^{W}F_{de}{}^{T}
\nonumber\\
&&\qquad -11\,F_{ab}{}^{Y}F_{cd}{}^{Z}F_{cd}{}^{V}\chi^{\alpha T}\Gamma_{\alpha \beta a}\,
\nabla_b\chi^{\beta W}
+22\,F_{ab}{}^{Y}F_{cd}{}^{Z}F_{ac}{}^{V}\chi^{\alpha T}\Gamma_{\alpha \beta b}\,
\nabla_d\chi^{\beta W}
\nonumber\\
&&\qquad +18\,F_{ab}{}^{Y}F_{cd}{}^{V}F_{ac}{}^{W}\chi^{\alpha T}\Gamma_{\alpha \beta b}\,
\nabla_d\chi^{\beta Z}
+12\,F_{ab}{}^{T}F_{cd}{}^{Y}F_{ac}{}^{V}\chi^{\alpha Z}\Gamma_{\alpha \beta b}\,
\nabla_d\chi^{\beta W}
\nonumber\\
&&\qquad +28\,F_{ab}{}^{T}F_{cd}{}^{Y}F_{ac}{}^{V}\chi^{\alpha W}\Gamma_{\alpha \beta b}\,
\nabla_d\chi^{\beta Z}
-24\,F_{ab}{}^{Y}F_{cd}{}^{V}F_{ac}{}^{T}\chi^{\alpha W}\Gamma_{\alpha \beta b}\,
\nabla_d\chi^{\beta Z}
\nonumber\\
&&\label{L3}\qquad +8\,F_{ab}{}^{T}F_{cd}{}^{Y}F_{ac}{}^{Z}\chi^{\alpha V}\Gamma_{\alpha \beta b}\,
\nabla_d\chi^{\beta W}
-12\,F_{ab}{}^{T}F_{ac}{}^{Y}
\nabla_bF_{cd}{}^{V}\chi^{\alpha Z}\Gamma_{\alpha \beta d}\,\chi^{\beta W}
\\
&&\qquad -8\,F_{ab}{}^{Y}F_{ac}{}^{T}
\nabla_bF_{cd}{}^{V}\chi^{\alpha Z}\Gamma_{\alpha \beta d}\,\chi^{\beta W}
+22\,F_{ab}{}^{V}F_{ac}{}^{Y}
\nabla_bF_{cd}{}^{T}\chi^{\alpha Z}\Gamma_{\alpha \beta d}\,\chi^{\beta W}
\nonumber\\
&&\qquad -4\,F_{ab}{}^{Y}F_{cd}{}^{T}
\nabla_eF_{ac}{}^{V}\chi^{\alpha Z}\Gamma_{\alpha \beta bde}\,\chi^{\beta W}
+4\,F_{ab}{}^{Y}F_{ac}{}^{T}
\nabla_cF_{de}{}^{V}\chi^{\alpha Z}\Gamma_{\alpha \beta bde}\,\chi^{\beta W}
\nonumber\\
&&\qquad +4\,F_{ab}{}^{T}F_{cd}{}^{Y}F_{ce}{}^{V}\chi^{\alpha Z}\Gamma_{\alpha \beta abd}\,
\nabla_e\chi^{\beta W}
-8\,F_{ab}{}^{Y}F_{cd}{}^{T}F_{ce}{}^{V}\chi^{\alpha Z}\Gamma_{\alpha \beta abd}\,
\nabla_e\chi^{\beta W}
\nonumber\\
&&\qquad +6\,F_{ab}{}^{V}F_{cd}{}^{Y}F_{ce}{}^{W}\chi^{\alpha Z}\Gamma_{\alpha \beta abd}\,
\nabla_e\chi^{\beta T}
+5\,F_{ab}{}^{V}F_{cd}{}^{W}F_{ce}{}^{Y}\chi^{\alpha Z}\Gamma_{\alpha \beta abd}\,
\nabla_e\chi^{\beta T}
\nonumber\\
&&\qquad +6\,F_{ab}{}^{Y}F_{ac}{}^{T}F_{de}{}^{V}\chi^{\alpha Z}\Gamma_{\alpha \beta bcd}\,
\nabla_e\chi^{\beta W}
-2\,F_{ab}{}^{Y}F_{ac}{}^{T}F_{de}{}^{Z}\chi^{\alpha V}\Gamma_{\alpha \beta bcd}\,
\nabla_e\chi^{\beta W}
\nonumber\\
&&\qquad +4\,F_{ab}{}^{Y}F_{ac}{}^{V}F_{de}{}^{Z}\chi^{\alpha W}\Gamma_{\alpha \beta bcd}\,
\nabla_e\chi^{\beta T}
+4\,F_{ab}{}^{T}F_{cd}{}^{V}F_{ce}{}^{Y}\chi^{\alpha Z}\Gamma_{\alpha \beta abd}\,
\nabla_e\chi^{\beta W}
\nonumber\\
&&\qquad -4\,F_{ab}{}^{Y}F_{cd}{}^{V}F_{ce}{}^{W}\chi^{\alpha Z}\Gamma_{\alpha \beta abd}\,
\nabla_e\chi^{\beta T}
\nonumber\\
&&\qquad +\tfrac{1}{2}\,F_{ab}{}^{Y}F_{cd}{}^{T}F_{ef}{}^{V}\chi^{\alpha Z}\Gamma_{\alpha \beta abcde}\,
\nabla_f\chi^{\beta W}
+\tfrac{1}{2}\,F_{ab}{}^{Y}F_{cd}{}^{T}
   f_{ef}{}^{Z}\chi^{\alpha V}\Gamma_{\alpha \beta abcde}\,
\nabla_f\chi^{\beta W}\bigg]\,.
\nonumber
\end{eqnarray}

In these formulas capital Roman letters are Lie algebra indices, $f^{XYZ}$ are structure constants of the gauge group Lie 
algebra.

The way to get the formula (\ref{E:dvsdjvh111}) will be described below.

One can construct a SUSY-invariant deformation of SYM theory reduced to a point by the formula:
\begin{equation}\label{E:fadjghjd}
\dL=A\tr G
\end{equation} 
where the operator $A$ is given by 
\begin{equation}\label{E:int}
A=\theta_{1}\dots \theta_{16}.
\end{equation}

Here $\tr G$ is a gauge invariant expression (we can consider $G$ as an element of $U(YM)$). If
$G$ is  $\Spin (10)$-invariant
the deformation $\dL=A\tr G$ is super Poincar\'e invariant. 

Notice that supersymmetry transformations $\theta_{\alpha}$ commute with $\tr$ , hence $A\tr G=\tr AG.$

If $G$ is a gauge invariant local expression (an element of $U(TYM)$) the formula (\ref {E:fadjghjd})  specifies a supersymmetric deformation of non-reduced SYM theory.

Let us prove these statements. It is sufficient to check that  in non-reduced case $\theta_
{\alpha}\dL$ is a total derivative on shell, i.e. $\theta_{\alpha}\dL=\pr{x^i}H^i$ on the equations of motion of $\L_{SYM}$.  It will follow  that in the reduced case $\theta_
{\alpha}\dL$  vanishes on shell.

To prove this fact we notice that the anti-commutator $[\theta_{\alpha},\theta_{\beta}]$ is a total covariant derivative as 
follows from (\ref{E:hfdfjfhjd}). It follows from the same formula that $\theta_{\alpha}^2$ is a total covariant derivative. 
Calculating $\theta_{\alpha}A\tr G$ we are moving $\theta_{\alpha}$ using (\ref{E:hfdfjfhjd}) until we reach $\theta$ with 
the same index. Then we use a formula for $\theta_{\alpha}^2$: 
\begin{equation}\label{E:jgkngft}
\begin{split}
&\theta_{\alpha}\dL=\tr(\theta_{\alpha}\theta_{1}\cdots \theta_{16}G)=\\
&=\sum_{\gamma=1}^{\alpha-1}(-1)^{\gamma}\Gamma_{\alpha\gamma}^k\tr(\theta_{1}\cdots \theta_{\gamma-1}D_k
\theta_{\gamma+1}\cdots \theta_{16}G) +\\
&+\frac{1}{2}\sum_{\gamma=1}^{\alpha-1}(-1)^{\alpha}\Gamma_{\alpha\alpha}^k\tr(\theta_{1}\cdots\theta_{\alpha-1} D_k
\theta_{\alpha+1}\cdots \theta_{16}G).\\
\end{split}
\end{equation}
Expressions $\tr(D_k\theta_{\alpha}\cdots \theta_{16}G)=\pr{x_k}\tr(\theta_{\alpha}\cdots \theta_{16}G)$ are total 
derivatives.  Expressions $\tr(\theta_{1}\cdots \theta_{\gamma-1}D_k\theta_{\gamma+1}\cdots \theta_{16}G)$ are 
multiple supersymmetry transformations of total derivatives. Hence due to  equation (\ref{E:supergaugetr}) $A\tr G$ is 
also a total derivative
on the equations of motion for $\L_{SYM}.$

The reader will recognize in (\ref{E:int}) a 10-dimensional analog of $\theta$-integration in theories admitting 
superspace formulation with manifest supersymmetries.

The above considerations can 
be used to describe all infinitesimal deformations of YM theory reduced to a point. Namely we have the following 
theorem.
\begin{theorem}\label{T:reduced}
Every infinitesimal  super Poincar\'e-invariant deformation of $\L_{SYM}$ reduced to a point is a linear combination of $
\L_{16}$ and a deformation having a form $A\tr(G)$, where $G$ is an arbitrary $\Spin(10)$-invariant  combination of 
products of $A_i$  and $\chi^{\alpha}$.
\end{theorem}

To formulate the corresponding statement in the case of unreduced $\L_{SYM}$ we should generalize the above 
consideration a little bit. We notice that infinitesimal deformation of $\L_{SYM}$ descends  to  a 
deformation of reduced $\L_{SYM}$ (we formally replace  $\nabla_i $ by constant matrix $A_i$ and matrix function $\chi$ by a constant matrix). It is 
not true that all  Lagrangians in the reduced theory can be lifted a Lagrangian in ten-dimensional theory (in other words not every Lagrangian of reduced theory descends from a Lagrangian of non-reduced theory). For example an 
expression $\tr A_iA_i$  defines  a Lagrangian of reduced  theory, but

\begin{equation}\label{E;laplacian}
\Delta=\tr \nabla_i\nabla_i
\end{equation}
 does not make sense as a ten-dimensional Lagrangian.

Of course,
 if $G$ itself is a gauge-covariant local expression, the expression $A G$ is also local and, as we noticed already, specifies a supersymmetric deformation. However there are situations 
when $G$ is not of this kind but still $AG$ after adding commutator terms, total derivatives and performing field redefinition becomes  a gauge-invariant local expression; then this expression can 
be considered as a
Lagrangian of  a deformation of  ten-dimensional SYM Lagrangian  $\L_{SYM}.$ We are saying that in this case the supersymmetric deformation of reduced theory can be lifted to a deformation of non-reduced theory.

 Our homological computations \cite{M4} show that the number of linearly independent Poincar\'e invariant deformations  of  reduced SYM theory that
 do not have the form $A\tr(G)$ where $G\in U(TYM)$, but can be lifted to ten-dimensional theory is equal to two.

To construct the first one we take $G$ to be the "Laplacian" (\ref{E;laplacian}).

We have 
\begin{equation}
A\tr(\Delta)=2\tr((A\nabla_i)\nabla_i)+\cdots,
\end{equation}
 where the dots represent gauge-invariant  local terms. This follows from formula (\ref{E:susyform}). It remains to prove we that
 $(A\nabla_i)\nabla_i)$ is   equivalent to  a gauge-covariant local expression (recall that we allow field 
redefinition and adding commutator terms).
It follows from the remarks at the beginning of the section that instead of working with $(A\nabla_i)\nabla_i$ we can 
work with $(A\bD_i)\bD_i$ considered as an element of $U(YM)/[U(YM),U(YM)]$. Moreover, it is convenient to work in the algebra $U(L)$ generated by $\btheta_{\alpha}$ obeying 
$\Gamma^{\alpha\beta}_{i_1,\dots,i_5}[\bm{\theta}_{\alpha},\bm{\theta}_{\beta}]=0$ (see Section 2).

The commutators with $\btheta_{\alpha}$ act on $U(YM)$ as supersymmetries $\theta_{\alpha}.$ 

In the algebra $U(L)$ we can represent $A(\bD_i)$ as the  multiple commutator:
 
 \[A(\bD_i)=[\btheta_{1},\dots, [\btheta_{16},\bD_i]\dots].\]
 We have more then four $\btheta$'s  in a row applied to $\btheta_{\beta}$.
 We see that $A(\bD_i)$  is a commutator 
\begin{equation}\label{E:fdfpfgj}
A(\bD_i)=[\bD_k, \psi_{ki}]+[\bchi^{\alpha}, \psi_{\alpha i}]
\end{equation} in $U(YM)$, where $\psi_{ki}, \psi_{\alpha i}$ are  gauge-covariant local expressions.
We have the following line of identities  where we can neglect commutator terms:  \[A(\bD_i)\bD_i=[\bD_k, \psi_{ki}]\bD_i
+[\bchi^{\alpha}, \psi_{\alpha i}]\bD_i=-\psi_{ki}[\bD_k\bD_i]+[\bD_k,\psi_{ki}\bD_i]-\psi_{\alpha i}[\bD_i,\bchi^{\alpha}]+
[\chi^{\alpha},\psi_{\alpha i}D_i].\] We obtain that $\tr(A(\bD_i)\bD_i)=\tr(\psi_{ki}[\bD_k\bD_i])-\tr(\psi_{\alpha i}[\bD_i,
\bchi^{\alpha}])$.

One can check that supersymmetric deformation obtained from $\Delta=G_1$ is equivalent to (\ref{E:dvsdjvh111}). 

One can prove   that similar considerations can be applied  to \begin {equation}\label {g2} G_2=a\tr(F_{i_2i_3}F_{i_2i_3}D_{i_1}D_{i_1})+b\tr
(\Gamma_{\alpha\beta}^{i_2}[D_{i_1},\chi^{\alpha}]\chi^{\beta}D_{i_2}D_{i_1})+c\tr(\Gamma_{\alpha\beta}^{i_1i_2i_3}F_
{i_2i_3}\chi^{\alpha}\chi^{\beta}D_{i_1})\end{equation} for an appropriate choice of constants $a,b,c$. Corresponding deformation will 
be denoted by $\dL_{28}$.   

In the Appendix  F we describe a general way to obtain formulas for supersymmetric deformations.
In particular, we give another expression for $G_2$ that does not contain indeterminate constants.

\begin{theorem}\label{E:theorem10}
Every deformation of Lagrangian $\L_{SYM}$ that  can be obtained by means of lifting of super Poincar\'e invariant deformation of reduced SYM theory is a linear combination of $\dL_{16}$ given by the formula (\ref {E:dvsdjvh}) , $\dL_{20}=
\tr H_1$, $\dL_{28}=\tr H_2$ and a deformation of a form $\tr(AG)$ where $G$ is an arbitrary Poincar\'e-invariant 
combination of products of covariant derivatives of curvature $F_{ij}$ and spinors $\chi^{\alpha}$.
Here  $H_1$ can be obtained from $AG_1$  and $H_2$ can be obtained from $AG_2$ by adding commutator terms, total derivatives and terms coming from field redefinition. ( $G_1$ is defined by the formula (\ref {E;laplacian}) and $G_2$ by the formula (\ref {g2}))
\end{theorem}

There is a finer decomposition of the linear space of equivalence classes of  Lagrangians. 
Any Lagrangian $\L$ under consideration  has the form  $\L=\tr Y(\nabla,\chi)$, where $Y$ is some  non-commutative 
polynomial in $\nabla_i$  and $\chi^{\alpha}$. 
Let the non-commutative polynomial $Y$ be a linear combination of commutators. Then of course 
$\tr Y\equiv 0$, however if $Y,Y^{'}$ are commutators then
\begin{equation}\label{E:two}
\tr YY^{'}
\end{equation}
 could be nonzero. The grading $\deg_{[\ ]}$ of a Lagrangian of the form $\tr YY^{'}$ by definition is equal to two (to the 
number of commutators in the product under the trace in (\ref{E:two}).{\footnote {Lagrangians of this kind  make sense 
not only for the gauge group  $\UN$, but also for an arbitrary compact gauge group $G$ because they can be written 
intrinsically in terms of the  commutator and the  invariant inner product of the Lie algebra of  $G$.}} For example the 
basic Lagrangian $\L_{SYM}$ has degree $\deg_{[\ ]}$ equal to two. Likewise we can define Lagrangians of arbitrary 
degree $\deg_{[\ ]}$.  The equations of motion of YM theory are compatible with classification of Lagrangians by $\deg_
{[\ ]}$ in the sense that Lagrangians of different degree are not equivalent.

The following table is a result of classification of linearly independent on-shell supersymmetric Lagrangians of low 
degree. The numbers in the body of the  table represent dimensions of spaces of super Poincar\'e invariant Lagrangians 
of degrees $(\deg_{[\ ]}, \deg_{\alpha '}).$

 \begin{equation}\label{T:dkodd1}
\mbox{
\scriptsize{
$\begin{array}{|c|c|c|c|c|c|c|c|c|c}
 &1&2&3&4&5&6&7&8&k=deg_{\alpha'}\\\hline
2& & &1& &1&3&18&172&\dots\\\hline
3& & & & & & &13&281&\dots\\\hline
4& &1& & &1&2&20 &267&\dots\\\hline
5& & & & & & &1 &68 &\dots\\\hline
6& & & & & & &1 &17 &\dots\\\hline  
7& & & & & & & & &\dots\\\hline
p=\deg_{[\ ]}&\dots&\dots&\dots&\dots&\dots&\dots&\dots &\dots &\dots
\end{array}$
}
}
\end{equation}

The entry in the second column corresponds to the Lagrangian (\ref{E:dvsdjvh}), the entry in the third column 
corresponds to the Lagrangian (\ref{E:dvsdjvh111}).

\section{The Homological Approach to Infinitesimal Deformations}\label{S:homologicalapp}

In this section we shall describe a reduction of the problem of infinitesimal SUSY deformations of SYM to a homological 
problem. A general  way to 
give homological formulation  of a problem of classification of deformations will be described in Section \ref{S:BV} and in 
Appendix \ref{AppendixA}; the relation of this way to the approach of present section will be studied in Appendix \ref
{AppendixC}.

 First of all we consider  infinitesimal deformations of SYM reduced to a point. As we have seen, this theory can be 
expressed in terms of algebra $U(YM)$. We shall regard the deformations of this theory as deformations of algebra $U
(YM)$. In other words we think about deformation as of family of multiplications on linear space $U(YM)$ depending 
smoothly on parameter $\epsilon$. In the case of infinitesimal deformations we assume that $\epsilon^{2}=0$ (i.e. we 
neglect higher order terms with respect to $\epsilon     $). We say that deformation is supersymmetric if it is possible to 
deform the SUSY algebra action on $U(YM)$ in such a way that it consists  of derivations of the deformed multiplication.

\begin{theorem}\label{T:dfadfqq}
Every cohomology class $\lambda \in \rH^2(L,U(YM))=\rH^2(L,\Sym(YM))$ specifies an infinitesimal supersymmetric 
deformation of $U(YM)$. 

We consider here $U(YM)$ as a representation of Lie algebra $L$. Due to Poincar\'e-Birkhoff-Witt theorem this 
representation is isomorphic to $\Sym(YM)$.

\end{theorem}

We shall start with  general statement about deformations of associative algebra $A$. The multiplication in this algebra can 
be considered as a bilinear map $m: A\otimes A\to A$. An infinitesimal deformation $m+\delta m$ of this map 
specifies an associative multiplication if 
$$\delta m(a,b)c+\delta m (ab,c)=a\delta m(b,c)+\delta m(a,bc).$$
This condition means that $\delta m$ is a  two-dimensional Hochschild  cocycle with coefficients in $A$ (see Appendix 
\ref{AppendixA}). Identifying equivalent deformations we obtain that infinitesimal deformations of associative algebra are 
labeled by the elements of Hochschild cohomology $\rHH^2(A,A)$. ( Two deformations are equivalent if they are related 
by 
linear transformation of $A$.) 

Applying this statement to the algebra $U(YM)$
we obtain that the infinitesimal deformations of this algebra are labeled by the elements  of 
$\rHH^2 (U(YM), U(YM))$. 

Let us consider now the Hochschild cohomology $\rHH^2 (U(L),U(YM))$. (Notice that $U(YM)$ is an ideal in $U(L)$, 
hence it can be regarded as a $U(L)$-bimodule.)  We can consider the natural restriction map 
$\rHH^2 (U(L),U(YM))\to \rHH^2 (U(YM),U(YM))$; we shall check that the image of this map consists of
supersymmetric deformations.  Let us notice first of all that $L=L^1+YM$ and the derivations
$\gamma _a$ corresponding to the elements $a\in L^1$ act on $YM$ as supersymmetries; this action can be extended 
to $U(YM)$ and specifies 
an action on Hochschild cohomology, in particular, on the space of deformations $\rHH^2 (U(YM),U(YM))$. (The derivation 
$\gamma _a$ is defined by the formula $\gamma _a(x)=[a,x].$)
On $L$ one can consider $\gamma _a$  as an inner derivation, hence its action on the cohomology $\rHH^2 (U(L),U(YM))
$ is trivial. (This follows from well known results, see, for example, \cite {McL}.) This means that supersymmetry 
transformations act trivially on the image of $\rHH^2 (U(L),U(YM))$ in $\rHH^2 (U(YM),U(YM))$ (in the space of 
deformations).

To obtain the statement of the theorem it is sufficient to notice that the Hochschild cohomology of the enveloping algebra 
of Lie algebra can be expressed in terms of Lie algebra cohomology (see (\ref {E:dsfdsh})). 

Theorem \ref {T:dfadfqq} gives a homological description of 
supersymmetric deformations of the equations of motion.  We can use homological methods to answer the question: 
when the deformed EM come from a Lagrangian.  As we have seen in Section \ref{S:fdgdfgmjl} the space of infinitesimal 
Lagrangian deformations of SYM theory reduced to a point can identified with $U(YM)/[U(YM),U(YM)]= \rHH_0(U(YM),U
(YM))$.  Lagrangian deformation generates a deformation of EM, hence there exists a map $U(YM)/[U(YM),U(YM)]\to 
\rHH^2(U(YM),U(YM))=\rH^2(YM,U(YM))$.  It turns out (see \cite{M3} and \cite{M4} ) that the image of this map has a finite 
codimension in $\rH^2(YM,U(YM))$ and it is onto for $\Spin(10)$-invariant elements. This means that all Poincar\'e 
invariant
infinitesimal deformations of EM are Lagrangian deformations.

Let us consider now deformations of supersymmetric deformations of supersymmetric YM theory in ten-dimensional case 
(SYM theory).
The description of these deformations is similar to reduced case.
\begin{theorem}\label{T:theorem12}
Every element $\lambda\in \rH^2(L,U(TYM))$ specifies a supersymmetric deformation of SUSY YM.
\end{theorem}
The group $\Spin (10)$ acts on cohomology; Poincar\'e invariant deformations are identified with
$\Spin (10)$-invariant cohomology classes.

In the proof we interpret the deformations of SYM theory as deformations of the algebra $U(TYM)$ and 
identify infinitesimal deformations with elements of Hochschild homology $\rHH_0(U(TYM),U(TYM)).$ However, the proof is more complicated; it is based on results of
 Section \ref {S:BV} and Appendix \ref{AppendixC}. It is shown in Appendix \ref{AppendixC} that the elements of higher cohomology groups 
also correspond to supersymmetric deformations, however only elements of $\rH^2$
give non-trivial super-Poincar\'e invariant infinitesimal deformations of equations of motion.

The next section is devoted to the calculation of cohomology entering the formulation of the theorems \ref{T:dfadfqq} and \ref {T:theorem12}.  In present section we will describe the solution of simpler problem of calculation of corresponding Euler characteristics. The Euler characteristics are especially interesting, because the contribution of infinitesimal deformations given by the general formula  (\ref {E:int}) with $G\in YM$ and $G\in TYM$ cancels in Euler characteristic. 

To get more information about homology we calculate Euler characteristics of  homology groups considered as graded $\Spin (10)$-
modules (as representations of $\Spin (10)).$ Recall that $L$ is graded; this gives a grading on homology. However,  there is also another grading coming from Poincar\'e-Birkhoff-Witt identification of $\g$ -modules $U(\g)=\sum \Sym ^k\g$. As explained in  the Section \ref{S:Euler} the  Poincar\'e series of graded $\Spin (10)$-modules
$U(YM)$, $U(TYM), U(L)$ and Euler characteristics  $\rH^{\bullet}(L, U(YM))$, $\rH^{\bullet}(L,U(TYM))$ can be considered as functions of $g\in \Spin (10)$ and series with respect $t$ (we disregard the the second grading for a moment). 
It follows from (\ref {E:ourcharacters} )  that Euler characteristics can be expressed in terms of Poincare series:
\begin{equation}\label{E:our}
\begin{split}
&\chi_{\rH^{\bullet}(L, U(YM))}(g,t)=\frac{U(YM)(g,t)}{U(L)(g^{-1},t^{-1})}\\
&\chi_{\rH_{\bullet}(L, U(YM))}(g,t)=\frac{U(YM)(g,t)}{U(L)(g,t)}
\end{split}
\end{equation}

The factor-algebra $L/TYM$ can be identified with the super Lie algebra of supersymmetries (16 odd generators transforming according spinor representation $S$, 10 even generators transforming according vector representation $V$). Similarly, the factor algebra $L/YM$ has 16 odd generators with trivial anticommutation relations; they transform as spinors.

This remark (together with the formula for $U(L)(g,t)$ given in Corollary 8)  permits us to prove the following formulas
\[\begin{split}
&U(YM)(g,t)=U(L)(g,t)\Sym S^*(g,-t)=\\
&=\frac{1}{1-V(g^{-1})t^2-S(g^{-1})t^3+S(g)t^5+V(g)t^6-t^8}\end{split}\]
\[\begin{split}
&U(TYM)(g,t)=U(L)(g,t)\Sym S^*(g,-t)\Lambda V(g,-t)=\\
&=\frac{\Lambda V(g,-t^2)}{1-V(g^{-1})t^2-S(g^{-1})t^3+S(g)t^5+V(g)t^6-t^8}\end{split}\]

We can use these formulas together with (\ref {E:our}) to prove that
\begin{equation}\label{E:ourcharacter}
\begin{split}
&\chi_{\rH^{\bullet}(L, U(YM))}(g,t)=-\frac{1}{t^8}\Sym S(g,-t^{-1})\\
&\chi_{\rH_{\bullet}(L, U(YM))}(g,t)=\Sym S^*(g,-t)
\end{split}
\end{equation}



\begin{equation}\label{E:ourcharacterss}
\begin{split}
&\chi_{\rH^{\bullet}(L, U(TYM))}(g,t)=-\frac{1}{t^8}\Sym S(g,-t^{-1})\Lambda V(g,-t^{-2})\\
&\chi_{\rH_{\bullet}(L, U(TYM))}(g,t)=\Sym S^*(g,-t)\Lambda V(g,-t^2)
\end{split}
\end{equation}

To take into account the  second  grading we should use formulas 
\[\begin{split}
&\sum\chi_{\rH^{\bullet}(L, \Sym^jTYM)}(g,t)z^j=\frac{U(TYM)(g,t,z)}{U(L)(g^{-1},t^{-1})}\\
&\sum\chi_{\rH_{\bullet}(L, \Sym^jTYM)}(g,t)z^j=\frac{U(TYM)(g,t,z)}{U(L)(g,t)}
\end{split}\]
\[\begin{split}
&\sum\chi_{\rH^{\bullet}(L, \Sym^jYM)}(g,t)z^j=\frac{U(YM)(g,t,z)}{U(L)(g^{-1},t^{-1})}\\
&\sum\chi_{\rH_{\bullet}(L, \Sym^j YM)}(g,t)z^j=\frac{U(YM)(g,t,z)}{U(L)(g,t)}
\end{split}\]
that also follow from (\ref{E:ourcharacters}) (we have replaced  the numerator in (\ref {E:our}) by the Poincar\'e series with respect to  the double grading).

To  calculate $U(YM)(g,t,z)$ and $U(TYM)(g,t,z) $ we notice that in general knowing the Poincar\'e series of Lie algebra $\g$ we can calculate the Poincar\'e series of $U(\g)$ that takes into account the additional grading coming 
 from
the Poincar\'{e}-Birkhoff-Witt     isomorphism  $U(\g)\cong \Sym [\g]$. {\footnote { In the notations of Appendix \ref{S:Appendix I} we should calculate the character of $\sigma_t(\g)$. }} From the other side the knowledge of the Poincar\'e series  of $U(\g)$ defined without additional grading is sufficient to find  the Poincar\'e series $\g(t)$  of $\g$.

If we neglect the action of a group we obtain\[
\g(t)=\sum_{n\geq 1}\frac{\mu(n)}{n}\ln U(\g)(-(-t)^n)
\]
where $\mu$ stands for the M\"{o}bius function.

We derive from it that
\[U(\g)(t,z)=\exp\left(\sum_{n,l\geq 1}\frac{\mu(l)}{nl}z^n\ln U(\g)(-(-t)^{nl})\right)\]
There is a more general formula that incorporates the group action
\[U(\g)(g,t,z)=\exp\left(\sum_{n,l\geq 1}\frac{\mu(l)}{nl}z^n\ln U(\g)(g^{nl},-(-t)^{nl})\right)\]

It can be obtained from the formula (\ref {ad}) in Appendix \ref{S:Appendix I}.

\section{Calculation of the Cohomology}\label{S:calc}
The cohomology governing infinitesimal deformations of ten-dimensional SYM and its reductions to a point were calculated in \cite {M4}.    In present section we use the approach of \cite {M4} to justify the statements of  Theorems 9 and 10.
The calculation will be based on Corollaries \ref{P:tqydxc1},\ref{P:iqwwst2}
(Section \ref{S:wweqe}).
We mentioned in Section \ref{2.1} that the algebra $\Ss$ is related to the manifold of pure spinors $\mathcal{C}$ and to the 
corresponding compact manifold $\mathcal{Q}$. Namely $\Ss_k$ can be interpreted as a space of holomorphic sections of line 
bundle $\O(k)$ over $\mathcal{Q}$. In other words 
\begin{equation}\label{E:cohomology0}
\Ss_k=\rH^0(\mathcal{Q},\O(k))   \mbox { for }{}  k\geq 0.
\end{equation}
 One can prove that all other cohomology groups $\rH^i(\mathcal{Q},\O(k))$ of $\mathcal{Q}$ with coefficients in line bundles $\O(k)$ are 
zero except 

\begin{equation}\label{E:cohomology10}
\rH^{10}(\mathcal{Q},\O(k))=\Ss^*_{-k-8} \mbox { for } k\leq -8.
\end{equation}

 The proof is based on Borel-Weil-Bott theorem.{\footnote { Borel-Weil-Bott theory  deals with calculation of the 
cohomology of  $G/P$   with coefficients in  $G$-invariant holomorphic vector bundles over $G/P$.  Here $G/P$ is a
 compact homogeneous space,  $P$ is a complex subgroup of complex Lie group $G$. These bundles correspond to 
complex representations of the subgroup $P$; more precisely, the total space of vector bundle $\mathcal{E}$ 
corresponding to $P$-module $E$ ( to a representation of $P$ in the space $E$)  can be obtained from $E\times G$ by 
means of factorization with respect to the action of $P$. 
 
  Usually Borel-Weil -Bott theorem is applied in  the case when the representation  of $P$ is one-dimensional (in  the 
case of line bundles); it describes the cohomology as a representation of the group $G$.  However,  more general case  
also can be treated \cite {bott}.

We suppose that the group $G$ is connected and the homogeneous space $G/P$ is simply connected; then $G/P$ can 
be represented as $M/P\bigcap M$ where $M$ is  a compact Lie group and $G$ is a complexification of $M$. If  $E$ is a 
complex $P$-module then 
$$ \rH^{\bullet}(G/P,\mathcal{E})=\sum K\otimes
\rH^{\bullet}(\p,\vvv, {\rm Hom}(K,E))$$
where $K$ ranges over irreducible $M$-modules.
This formula  gives an expression of  cohomology with coefficients in vector bundle in terms of relative Lie algebra 
cohomology, $\p$ stands for real Lie algebra of $P$ and $\vvv$ stands for  Lie algebra of $P\bigcap M$.}} 

The relation between $\Ss$ and $\mathcal{Q}$ can be used to express cohomology of a graded  $L$-module $N=\bigoplus_{m
\geq m_0}N_m$ in terms of cohomology groups related to $\mathcal{Q}$. Recall that
Corollary \ref{P:tqydxc1} permits us to  reduce the calculation of the cohomology at hand to the calculation of the cohomology of the 
complex (of differential module)
$$N^{\bullet}_{c\  \bm{m}}=(N_{m}\otimes \Ss_0\rightarrow N_{m+1}\otimes \Ss_1 \rightarrow\dots )$$

 We can construct a differential vector bundle (a complex of holomorphic vector bundles $\N^{\bullet}$) over $\mathcal{Q}$  in 
such a way that one  obtains the above complex of modules considering holomorphic sections of vector bundles:
 \begin{equation}
\label{N}
\begin{split}
&\N^{\bullet}_{c\  \bm{l}}=(\cdots \rightarrow N_{l-1}\otimes \O(-1)\rightarrow N_{l}\otimes \O(0)\rightarrow N_{l+1}\otimes 
\O(1) \rightarrow\dots) =\\
&=(\cdots\to N_{l-1}(-1)\to N_l(0)\to N_{l+1}(1)\to \dots).
\end{split}
\end{equation}

We use here the notation $N(k)=N\otimes \O(k)$. Notice, that  the construction of the complex of vector bundles  
depends on the choice of index $l$, but this dependence is very simple:$\N^{\bullet}_{c\  \bm{l+1}}=\N^{\bullet}_{c\  \bm
{l}}(-1)$. The differential $d_e$ is a multiplication by
\begin{equation}\label{E:wuwhswtyrt}
e=\lambda^{\alpha}\bm{\theta}_{\alpha}.
\end{equation} 

Let us assume that the modules $N_i$ are  also
$\Spin (10)$-modules (more precisely, $N$ is
a module with respect of semidirect product of $L$ and $\Spin (10)$).
Then vector bundles in the complex (\ref {N}) are 
$\Spin (10)$-invariant; corresponding complex $N_P$
of $P$-modules has the form 
$$N_P=(\cdots \rightarrow N_{l-1}\otimes \mu_{-1}\rightarrow N_{l}\otimes \mu_0\rightarrow N_{l+1}\otimes \mu_1 
\rightarrow\dots) $$
(Recall that  $\mathcal{Q}= \SO (10)/\U(5)$ can be obtained also by means of taking quotient  of complex spinor group $\Spin
(10,\mathbb{C})$ with respect to the  subgroup $P$ defined as a stabilizer of a point $\lambda _0 \in \mathcal{Q}$; see Section  2.3. The complex of $P$-modules comes
from consideration of the complex of fibers over $\lambda _0$. )  

Let us consider hypercohomology of $\mathcal{Q}$ with the coefficients in the complex $\N^{\bullet}_{\bm{l}}=\N^{\bullet}_{c\ \bm
{l}}$. These hypercohomology can be expressed  in terms of the Dolbeault cohomology of $\N^{\bullet}_{\bm{l}}$. Namely 
we should consider the bicomplex $\Omega^{\bullet}(\N^{\bullet}_{\bm{l}})$ of smooth sections of the bundle of $(0,p)$-
forms with coefficients in  $\N^{q}_{\bm{l}}$. Two differentials are $\dbar$ and $d_e$. Hypercohomology $\mathbb{H}^{i}
(\mathcal{Q},\N_{\bm{l}} )$ can be identified with cohomology of the total differential $\dbar+d_e$ in $\Omega^{\bullet}(\N^{\bullet}
_{\bm{l}})$. 

As usual we can analyze cohomology of the total differential  by means of two spectral sequences whose $E_2$ terms 
are equal to  $\rH^i(\rH^j(\Omega(\N_{\bm{l}} ),\dbar),d_e)$ and $\rH^i(\rH^j(\Omega(\N_{\bm{l}} ),d_e),\dbar)$.
\begin{proposition}\label{C:osadfx}
There is a long exact sequence of cohomology
\begin{equation}\label{E:delta}
\begin{split}
& \dots \rightarrow \rH^i(N_{c\ \bm{l}})\rightarrow \mathbb{H}^{i}(\mathcal{Q},\N_{\bm{l}} )\rightarrow \rH^{i-10}(N_{h\ \bm{-8-l}})
\overset{\delta}{\rightarrow}\\
&\overset{\delta}{\rightarrow} \rH^{i+1}(N_{c\ \bm{l}})\rightarrow \dots
\end{split}
\end{equation}
\end{proposition}
\begin{proof}
It follows readily from equalities (\ref{E:cohomology0}, \ref{E:cohomology10}) that nontrivial rows in $E_2$ of the first 
spectral sequence are $(\rH^0(\Omega(\N^{\bullet}_{\bm{l}} ),\dbar)=N_{c\ \bm{l}}$ and $(\rH^{10}(\Omega(\N^{\bullet}_{\bm
{l}} ),\dbar)=N_{h\ \bm{-8-l}}$. ( We use the notations of Corollaries  \ref{P:tqydxc1} and \ref{P:iqwwst2}.) The operator $
\delta$ is the
differential in $E_2$. To complete the proof we notice that this is the only non-vanishing differential in the spectral 
sequence.
\end{proof}

We shall be interested in graded $L$ module
$N=YM$ ; the corresponding graded differential vector bundle (complex of vector bundles) is denoted by $\YM$. Notice, 
that this bundle  is $\Spin (10)$-invariant; it corresponds to the following representation of the group $P$:
\begin{equation}
\label{ym}
L^2+L^3\otimes \mu _1+L^4\otimes \mu _2+....
\end{equation}

 ( As we have noticed there is a freedom in the
construction of complex of vector bundle; the above formula corresponds to $l=2.$)  
 
Similarly starting with $L$ module $TYM$ one can define graded differential vector bundle $\TYM$; it corresponds to the 
representation 
\begin{equation}
\label{sym}
L^3\otimes \mu _1+L^4\otimes \mu _2+...
\end{equation}

 of the group $P$. 

More generally,  we can consider the module \[N=\bigoplus_{k\geq 0} N^k =\Sym^jYM=\bigoplus_{k\geq 0}\Sym^jYM
^k\]  equipped with adjoint action of $L$. Corresponding complexes of vector bundles  are denoted $\Sym ^j \YM$. 
Symmetric algebra $\Sym$ is understood in the graded sense. 

Similarly, we can define complexes of vector bundles $\Sym ^j \TYM$.

Let $\W^*$ be the  vector bundle on $\mathcal{Q}$ induced from the representation $W^*$ of $P$. 
It follows from (\ref{E:decompb}) that there is an embedding $W^*\subset L^2=YM^2\subset YM$.
  From this we conclude that  
there is an embedding $\W^*\rightarrow \YM^{\bullet}$, where we consider $\W^*$ as a graded vector bundle with one 
graded component $W^*$ in grading $2$ and zero differential (as one-term complex). 
\begin {proposition}\label{P:W}
The embedding $\W^*\rightarrow \YM^{\bullet}$ is a quasi-isomorphism.
\end{proposition}
We relegate  the proof to the Appendix \ref{AppendixE}.
\begin {corollary}
The embedding of  $\Sym^i (\W^*)$ into $\Sym^i (\YM)$ is a quasi-isomorphism.
\end {corollary}
Here  $\Sym^i (\W^*)$ is considered as
graded vector bundle with grading $2i$.
To deduce the corollary we use K\"unneth theorem.

 We can reformulate Proposition \ref {P:W} saying that the induced map of hypercohomology $\mathbb{H}^{\bullet}(\W^*)
\rightarrow \mathbb{H}^{\bullet}(\mathcal{Q}, \YM^{\bullet}_{\bm{0}}) $ is an isomorphism.   Similarly, the map $\mathbb{H}^
{\bullet}(\mathcal{Q},\Sym^i(\W^*)(l))\rightarrow \mathbb{H}^{\bullet}(\Sym^i\YM_{\bm{l}}^{\bullet}) $ is an isomorphism.

Using this statement  and (\ref {E:delta}) we obtain
\begin{corollary}\label{C:osadfkkkx}
There is a long exact sequence of cohomology
\begin{equation}\label{E:deltaaa}
\begin{split}
& \dots \rightarrow \rH^i(\Sym^jYM_{c\ \bm{l}})\rightarrow \rH^{i+l-2j}(\mathcal{Q},\Sym^j(\W^*)(2j-l) )\rightarrow \rH^{i-10}(\Sym^jYM_{h
\ \bm{-8-l}})\overset{\delta}{\rightarrow}\\
&\overset{\delta}{\rightarrow} \rH^{i+1}(\Sym^jYM_{c\ \bm{l}})\rightarrow \dots
\end{split}
\end{equation}
\end{corollary}

Using Corollary \ref{P:iqwwst2} we can identify the cohomology $\rH^{\bullet}(\Sym^jYM_{h})$ with homology $\rH_{\bullet}
(L,YM)$. Likewise $\rH^{\bullet}(\Sym^jYM_{c})$ is isomorphic to  $\rH^{\bullet}(L,\Sym(YM))$. This means that we can
formulate  (\ref {E:deltaaa}) as a long exact sequence
\begin{equation}
\label{E:deltaa2}
\begin{split}
& \dots \rightarrow \rH^{i,\bm{l}}(L,\Sym^jYM)\rightarrow \rH^{i+\bm{l}-2j}(\mathcal{Q},\Sym^j(\W^*)(2j-\bm{l}) )\rightarrow \\
&\rightarrow \rH_{2-i,\bm{l}-\bm{8}}(L,\Sym^jYM)\overset{\delta}{\rightarrow} \rH^{i+1,\bm{l}}(L,\Sym^jYM)\rightarrow \dots\\
\end{split}
\end{equation}

The hypercohomology   $\mathbb{H}^{\bullet}(\mathcal{Q},\Sym^i(\W^*)(l))
$ it is equal up to a shift in grading to the ordinary cohomology of the $\mathcal{Q}$ vector bundle  $\Sym^i(\W^*)(l)$. Such 
cohomology can be computed via Borel-Weil-Bott theory. 

\begin{proposition}\label{P:invariants}\ \\

$\rH^0(\mathcal{Q},\Sym^j(\W^*)(l))=[0,0,0,j,l-j], j,l-j \geq 0$

$\rH^4(\mathcal{Q},\Sym^j(\W^*)(l))=[j-3,0,0,l+2,0], l \geq-2,j\geq 3$

$\rH^{10}(\mathcal{Q},\Sym^j(\W^*)(l))=[j,0,0,-8-l,0], l \leq -8,j\geq 0$

Straightforward inspection of the cohomology groups shows that the following groups are generated by    $\Spin(10)$-
invariant elements : $\langle e\rangle=\rH^0(\mathcal{Q},\Sym^0(\W^*)(0))$, $\langle c\rangle=\rH^4(\mathcal{Q},\Sym^3(\W^*)(-2))$, $\langle 
e'\rangle=\rH^{10}(\mathcal{Q},\Sym^0(\W^*)(-8))$.
\end{proposition}

To analyze the super Poincar\'e invariant deformations we use $\Spin (10)$-invariant part
of exact sequence (\ref{E:deltaa2}). It is easy to check that $\Spin (10)$-invariant elements of hypercohomology listed 
above are mapped into zero in this long exact  sequence. This means that this long exact sequence splits into short 
exact sequences
  \begin{equation}
\label{E:deltasddd}
\begin{split}
&\mbox{ if } i=11, j=0,l=8 \mbox{  then }\\
&0 \rightarrow \rH^{i+l-2j-1}(\mathcal{Q},\Sym^j(\W^*)(2j-l) )^{\Spin(10)}\rightarrow \\
&\rightarrow \rH_{3-i,\bm{l}-\bm{8}}(L,\Sym^jYM)^{\Spin(10)}\overset{\delta}{\rightarrow} \rH^{i,\bm{l}}(L,\Sym^jYM)^{\Spin
(10)}\rightarrow  0\\
&\mbox{ otherwise }\\
&0 \rightarrow \rH_{3-i,\bm{l}-\bm{8}}(L,\Sym^jYM)^{\Spin(10)}\overset{\delta}{\rightarrow} \rH^{i,\bm{l}}(L,\Sym^jYM)^{\Spin
(10)}\rightarrow\\
&\rightarrow \rH^{i+l-2j}(\mathcal{Q},\Sym^j(\W^*)(2j-l) )^{\Spin(10)}\rightarrow 0
\end{split}
\end{equation} 

We see that  that $\Spin (10)$-invariant elements  of hypercohomology $e,c,e'$ contribute to (co)homology $\rH^{0,\bm
{0}} (L, \mathbb{C})$, $\rH^{2,\bm{8}} (L, \Sym^3 YM)^{\Spin(10)}$ and $\rH_{0,\bm{0}} (L, \mathbb{C})$. The only non-
trivial contribution corresponds to $c$ and gives the infinitesimal deformation $\dL_{16}$ (\ref {E:dvsdjvh}).

{\bf Proof of Theorem \ref{T:reduced}}

We shall give a proof of this theorem assuming that
all infinitesimal supersymmetric deformations are given by
Theorem \ref {T:dfadfqq}. 
The key moment in the proof is the use of short exact sequence (\ref{E:deltasddd}).

The operator $\delta$ in exact sequence  (\ref{E:deltaa2}) defines a map \[\delta:\rH_{1}(L,\Sym(YM))\rightarrow \rH^{2}(L,\Sym(YM))\] whose kernel and cokernel are controlled by the exact sequence. We conclude that the space $\delta(\rH_{1}(L,\Sym(YM))^{\Spin(10)})$ has codimension one in $\rH^{2}(L,\Sym(YM))^{\Spin(10)}$. We shall prove that the space of super Poincar\'e invariant deformations of equations of motion given by  the
formula (\ref {E:fadjghjd})  has the same codimension in $\rH^{2}(L,\Sym(YM))^{\Spin(10)}$ as $\delta(\rH_{1}(L,\Sym(YM))^{\Spin(10)})$; this gives a proof of the
theorem \ref{T:reduced}. (The formula (\ref {E:fadjghjd}) specifies a supersymmetric deformation of Lagrangian. However,  a deformation of Lagrangian function produces a deformation of equations of motions; this manifests in a map \[var:\rH_{0}(YM,U(YM))\rightarrow \rH^{2}(YM,U(YM)).\]  See Section   \ref{S:homologicalapp}
for more detail.)

Supersymmetry transformations $\btheta_{\alpha}$  act by derivations on Lie algebra $YM$. From this we conclude $\btheta_{\alpha}$ induce operators that act on objects constructed naturally (functorially)  from $YM$. In particular they act on 
\begin{equation}
\rH^i(YM,U(YM))\overset{P}{\cong}\rH_{3-i}(YM, U(YM))
\end{equation}
Here $P$ denotes the Poincar\'e isomorphism (see Appendix A).
The composition $\btheta_1\cdots \btheta_{16}$ defines an operator in homology.  We shall use the notation $A_k$ for this operator acting on $k$-dimensional homology:
\[A_k:\rH_k(YM,U(YM))\overset{A}{\rightarrow }\rH_k(YM,U(YM))\] 
 In Section \ref {S:fdgdfgmjl} we have interpreted the linear space $\rH_0(YM,U(YM))\cong \rH_0(YM,\Sym(YM))$ as a linear space of infinitesimal deformations of action functions in the reduced theory. Obviously the operator $A_0$ coincides with $A$ defined in (\ref{E:fadjghjd}). 
 
 The maps $A_k$ have an alternative description. Let $N$ be an $L$-module. It is also an $YM$ -module. Since homology is a covariant functor with respect to the Lie algebra argument there is a map $(i_{*})_k:\rH_k(YM,N)\rightarrow \rH_k(L,N)$. Likewise there is a map in opposite direction on cohomology $(i^*)_k:\rH^k(L,N)\rightarrow \rH^k(YM,N)$.  These observations enable us to define composition maps 
\[T_k:\rH_k(YM, U(YM))\overset{(i_{*})_k}{\rightarrow} \rH_k(L, U(YM))\overset{\delta}{\rightarrow }\rH^{3-k}(L, U(YM))\overset{(i^{*})_{3-k}}{\rightarrow}\rH^{3-k}(YM, U(YM))\overset {P}{\rightarrow} \rH_k(YM, U(YM))\]
Notice, that   the map $i^*_2$ acts from
$\rH^2(L,U(YM))$ into $\rH^2(YM,U(YM))$;  we have shown in Section \ref{S:homologicalapp} that the elements in the image of this map correspond to supersymmetric deformations.  The same arguments can be applied to the map $i^*_k$; they lead to the conclusion that $$T_k:\rH_k(YM, U(YM))\to \rH_k(YM, U(YM))^{\susy}$$ (in other words, the image of $T_k$ consists of supersymmetric elements). The map $A_k$  obviously has the same feature, therefore it is natural to conjecture that
the maps $A_k$ and $T_k$ coincide. To prove this conjecture we notice that the
 operators $A_k$ and $T_k$ can be defined for arbitrary $L$- module $N$ as operators
 $$\rH_k(YM, N)\to \rH_k(YM,N)^{\susy}$$.  Using free resolutions one can reduce the proof to the consideration of the module $N\cong U(L)$ where $L$ acts on  $U(L)$  by left multiplication
(see \cite{M4} for details).

In general it is not easy to describe maps $i_{*}$ and $i^{*}$. It is easier to analyze their restrictions to $\Spin (10)$-invariant elements.  Let us consider maps $i_{*1}:\rH_1(YM,U(YM))^{\Spin(10)}
\rightarrow \rH_1(L,U(YM))^{\Spin(10)}$ and $i^{*}_2:\rH^2(L,U(YM))^{\so(10)}\rightarrow \rH^2(YM,U(YM))^{\so
(10)\ltimes \susy}\subset \rH^2(YM,U(YM))^{\Spin(10)} $.  One can prove the following
\begin{lemma}\label{L:dewr23}
The maps $i_{*1},i^{*}_2$ are surjective.
\end{lemma}
If we take this Lemma for granted we conclude that $A_1:\rH_1(YM,U(YM))^{\so(10)}\rightarrow \rH_1(YM,U(YM))^
{\so(10)\ltimes \susy}$ has one-dimensional co-kernel (of the same dimension as the co-kernel of the map $\delta$).

The rather technical  proof of the lemma (see \cite{M4}) is based on analysis of Serre-Hochschild spectral sequences associated with 
extension $YM\subset L$:
\[\begin{split}
&\rH_{i}(YM,\Sym YM)\otimes\Sym^jS^*\Rightarrow \rH_{i+j}(L,\Sym YM)\\
&\rH^{i}(YM,\Sym YM)\otimes\Sym^jS\Rightarrow \rH^{i+j}(L,\Sym YM).
\end{split}\]

Notice that the surjectivity of $i^*_2$ has clear physical meaning: it can be interpreted as  a statement that all super Poincar\'e invariant  deformations in the sense of Section \ref{S:homologicalapp} are described by Theorem \ref{E:theorem10}.

The above considerations gave us the information about the codimension of the image
of the operator $A_1$. To prove Theorem \ref {T:reduced} we need information about the codimension of the image of $A_0$. This information can be obtained from the results about operator $A_1$ by means of Connes differential
$$B: \rH_k(YM, U(YM))\to \rH_{k+1} (YM, U(YM))$$
(see Appendix \ref{AppendixA}).
Using the fact that supersymmetries commute with the Connes differential we obtain that
$$ A_{k+1}B= BA_k$$, in particular, $A_1B=BA_0$.
 
We need the following 
\begin{lemma}\label{L:BBB}
The map $B$ defines a surjective map $\rH_0(YM,U(YM))^{\so(10)}\rightarrow \rH_1(YM,U(YM))^{\so(10)}$ with 
one-dimensional kernel generated by constants.
\end{lemma}
\begin{proof}
The proof (see \cite{M3} and \cite{M4}) is based on  a general theorem (see \cite{Loday}) which asserts that the 
cohomology of $B$ in  $\rH_{i}(\g,U(\g))$ for positively graded $\g$ is trivial and generated by constants $\mathbb{C}
\subset \rH_{0}(\g,U(\g))$. The rest follows from the information about homology  of $YM$ with coefficients in $U(YM)$ (see \ref{E:cohcomp}).
\end{proof}

The proof of the statement that co-dimension of $\Imm(A_0)$ in the space of $\susy$-invariant elements in 
$\rH_{0}(YM,U(YM))$ is equal to one easily follows from this lemma. We know that the image of map $A_1$ has co-
dimension one in the space of $\susy$-invariants. The operator $B$ preserves  $\so(10)\ltimes \susy$-invariant 
subspaces. If we write $\rH_{0}(YM,U(YM))=\mathbb{C}+\underline{\rH}_{0}(YM,U(YM))$, the operator $B$ admits the 
inverse: $B^{-1}:\rH_{1}(YM,U(YM))^{\Spin(10)}\rightarrow \underline{\rH}_{0}(YM,U(YM))^{\Spin(10)}$. The identity 
$A_1B=BA_0$ implies that $A_0$ is equal to $B^{-1}A_1B$, when restricted on $\underline{\rH}_{0}(YM,U(YM))^{\Spin
(10)}$. The claim follows from the corresponding statement for $A_1$.

 The reader should consult for missing details the references \cite{M3} and \cite{M4}.

We have analyzed  the case of reduced SYM theory. Very similar considerations can be applied to the unreduced case.

First of all we  should formulate the analog of
Proposition \ref {P:W}. Let us notice that it follows from (35) that $W\otimes \mu _{-1}\subset L^3,$ hence $W\subset 
L^3\otimes \mu _1.$   Using (\ref {sym}) we conclude that  
there is an embedding $ {\W}\rightarrow \TYM^{\bullet}$, where we consider $ {\W}$ as a graded vector bundle with one 
graded component  in grading $3$ that corresponds  to $P$-module $ W$ and has zero differential. 
\begin{proposition}
\label{P:WW}
The embedding of $
 {\W}$ into $\TYM^{\bullet}$  is a quasi-isomorphism.
\end{proposition}
The proof will be given in Appendix \ref{AppendixE}.

Using this proposition we can write down an exact sequence analogous to (\ref{E:deltaa2}).
\begin{corollary}\label{C:fullexaxt}
There is a long exact sequence connecting $\rH^k(L,U(TYM))$, $\rH_k(L,U(TYM))$ and hypercohomology:
\begin{equation}\label{E:ywers}
\begin{split}
&\dots \rightarrow \rH_{3-i,a-8}(L,\Sym^j(TYM))\overset{\delta}{\rightarrow} \rH^{i,a}(L,\Sym^j(TYM))\rightarrow\\
& \rightarrow \rH^{i+a-3j}(\mathcal{Q} ,\Lambda^j(\W)(3j-a))\overset{\iota}{\rightarrow}  \rH_{2-i,a-8}(L,\Sym^j(TYM))\rightarrow\dots
\end{split}
\end{equation}

\end{corollary}

Again using Borel-Weil-Bott theorem we can calculate  the cohomology of $\mathcal{Q}$ with coefficients in vector bundles that 
enter this sequence.
\begin{proposition} 
\begin{equation}\label{E:hdytsf}
\begin{split}
&i\geq 0\\
&\rH^0(\mathcal{Q} ,\O(i))=[0,0,0,0,i],\quad \rH^{10}(\mathcal{Q} ,\O(-8-i))=[0,0,0,i,0],\\
&\rH^0(\mathcal{Q} ,\W(i+1))=[1,0,0,0,i],\quad \rH^{10}(\mathcal{Q} ,\W(-8-i))=[0,0,0,i,1],\\
&\rH^0(\mathcal{Q} ,\Lambda^2(\W)(2+i))=[0,1,0,0,i],\quad \rH^{10}(\mathcal{Q} ,\Lambda^2(\W)(-8-i))=[0,0,1,i,0], \\
& \rH^9(\mathcal{Q} ,\Lambda^2(\W)(-6))=[0,0,0,0,0]\\
&\rH^0(\mathcal{Q} ,\Lambda^3(\W)(3+i))=[0,0,1,0,i],\quad \rH^{10}(\mathcal{Q} ,\Lambda^3(\W)(-7-i))=[0,1,0,i,0], \\
& \rH^1(\mathcal{Q} ,\Lambda^3(\W)(1))=[0,0,0,0,0]\\
&\rH^0(\mathcal{Q} ,\Lambda^4(\W)(3+i))=[0,0,0,1,i],\quad \rH^{10}(\mathcal{Q} ,\Lambda^4(\W)(-6-i))=[1,0,0,i,0], \\
&\rH^0(\mathcal{Q} ,\Lambda^5(\W)(3+i))=[0,0,0,0,i],\quad \rH^{10}(\mathcal{Q} ,\Lambda^5(\W)(-5-i))=[0,0,0,i,0], 
\end{split}
\end{equation}
\end {proposition}
To analyze super Poincar\'e invariant deformations of unreduced theory we should study $\Spin (10)$-invariant part of 
long exact sequence (\ref {E:ywers}). As in reduced case $\Spin(10)$-part of the exact sequence splits into short exact 
sequences. More precisely if the indices $(i,j,\bm{a})$  belong to the set $\{(3,0,8),(4,2,12),(6,5,20)\}$ then we have the 
splitting

\begin{equation}\label{E:ywersplit}
\begin{split}
&0\rightarrow \rH^{i+a-3j-1}(\mathcal{Q} ,\Lambda^j(\W)(3j-a))^{\Spin(10)}\rightarrow \rH_{3-i,a-8}(L,\Sym^j(TYM))^{\Spin(10)}\overset
{\delta}{\rightarrow} \\
&\overset{\delta}{\rightarrow} \rH^{i,a}(L,\Sym^j(TYM))^{\Spin(10)}\rightarrow 0\\
&\mbox{ If }(i,j,\bm{a}) \in \{(0,0,0)(2,3,8)(3,5,12)\}\\
&0 \rightarrow \rH_{3-i,a-8}(L,\Sym^j(TYM))^{\Spin(10)}\overset{\delta}{\rightarrow} \rH^{i,a}(L,\Sym^j(TYM))^{\Spin(10)}
\rightarrow\\
& \rightarrow \rH^{i+a-3j}(\mathcal{Q} ,\Lambda^j(\W)(3j-a))^{\Spin(10)}\rightarrow 0\\
&\mbox{and for all other }  (i,j,\bm{a})\\
&\rH_{3-i,a-8}(L,\Sym^j(TYM))^{\Spin(10)} \overset{\delta}{\cong} \rH^{i,a}(L,\Sym^j(TYM))^{\Spin(10)}
\end{split}
\end{equation}

The $\Spin (10)$-invariant part of hypercohomology is six-dimensional, but only three-dimensional   part of it , as the 
reader can see in (\ref{E:ywersplit}),   gives a contribution to the
cohomology $\rH^0, \rH^2$ and $\rH^3$. The contribution to $\rH^0$ is not interesting; the contribution to $\rH^2$ gives the 
deformation $\dL _{16}$  and  the contribution to $\rH^3$ is trivial at the level of infinitesimal deformations of equations of 
motion (but it gives a  non-trivial deformation of L$\ity$ action of  supersymmetry, hence the construction of Section \ref
{E:Formal} can give a non-trivial formal deformation).

The analogs of operators $A_k$ and $T_k$ can be defined in the situation at hand; again $A_k=T_k$.

The most technical part of the proof is hidden  in the verification of the analog of Lemma \ref{L:dewr23}. 
\begin{lemma}
The co-kernels of the maps $i_{*1}:H_1(YM, U(TYM))^{\Spin(10)}\rightarrow \rH_1(L,U(TYM))^{\Spin(10)}$ and $i^
{*}_2:H^2(L,U(TYM))^{\Spin(10)}\rightarrow \rH^2(YM,U(TYM))^{\Spin(10)}$ have dimensions two and zero 
respectively.
\end{lemma}

The rest of the proof follows along the lines of the proof in reduced case. In particular one should use the analog of  Lemma \ref{L:BBB}.

\section{The BV formalism: a geometric approach}\label{S:BV}

Our considerations will be based on the Batalin-Vilkovisky (BV) formalism. In this formalism a classical system is 
represented by an action functional $S$ defined on an odd symplectic manifold $M$ and obeying the classical Master 
equation 
\begin{equation}\label{E:fggdsfjh}
\{S,S\}=0.
\end{equation} where $\{\cdot, \cdot \}$ stands for the odd Poisson bracket. Using an odd symplectic form $\omega=dx^i
\omega_{ij}dx^j$
we assign to every even functional $F$ an odd vector field $\xi=\xi_{F}$ defined by the formula 
\begin {equation}\label {ham}
\xi^{i}\omega_{ij}=\frac{\partial F}{\partial z^j}
\end{equation}
(Hamiltonian vector field corresponding to the functional $F$). The form $\omega$ is invariant with respect to $\xi_{F}$. In particular we 
may consider an odd vector field $Q=\xi_{S}$; this field obeys $[Q,Q]=0$. Here $[\cdot,\cdot ]$ stands for 
supercommutator. The solutions to the equations of motion (EM) are identified with zero locus of $Q$. 

In an equivalent formulation of BV we start with an odd vector field $Q$ obeying $[Q,Q]=0$. We require the existence of
$Q$-invariant odd non-degenerate closed two- form $\omega$ (odd symplectic form). Then we can restore the action functional from $Q^{i}\omega_{ij}=\frac
{\partial S}{\partial z^j}$. 

Sometimes is convenient to drop the condition of non-degeneracy of the form $\omega$. In this case one cannot define the Poisson bracket, but the definition of Hamiltonian vector field still makes sense.

We say that in BV-formalism a classical system  is defined by means of an odd vector field $Q$ obeying $[Q,Q]=0$. In geometric 
language we are saying that a classical system is a $Q$-manifold. Fixing a vector field $Q$ we specify equations of 
motion of our system,  but we do not require that EM come from an action functional. If there exists a $Q$-invariant odd 
symplectic form we can say that our system comes from action functional $S$ obeying classical Master equation $\{S,S\}
=0$.  In this case we say that we are dealing with a Lagrangian system. In geometric language we can identify it with an 
odd symplectic $Q$-manifold. 

Infinitesimal deformation of a classical system corresponds to a vector field $\xi$ obeying $[Q,\xi]=0$ (then $[Q+\xi,Q+\xi]
=0$ in the first order with respect to $\xi$). An infinitesimal deformation $\xi$ is trivial if $\xi=[Q,\eta]$ because such a 
deformation  corresponds to a change of variables (field redefinition) $x^i\rightarrow x^i+\eta^i$. The operator $\xi\rightarrow [Q,\xi]$ defines a cohomological differential on $Vect(M)$ which we by abuse of notations denote by the same letter $Q$.
The deformations 
of a classical system  corresponding to vector field $Q$ are labeled by $Q$-cohomology  of the space of vector fields $Vect
(M)$.

The algebra of smooth functions $C^{\infty}(M)$ on $M$ can be considered as super commutative differential graded algebra with 
differential $Q^i\pr{x^i}$, which we as in case of the space of vector fields denote by $Q$. The cohomology $\rH(C^{\infty}(M),Q)$ can be identified with classical observables. In other 
words a classical observable is defined as a function $O$ obeying $Q(O)=0$. Two classical observables $O_1,O_2$ are 
identified if $O_1-O_2=Q(O')$ for some $O'\in C^{\infty}(M)$. Classical observables label infinitesimal deformations of solutions to the classical master equation which are the same as deformations of classical Lagrangian system in BV formalism. This follows from the remark that  the equation $\{S+\sigma , S+\sigma\}=0$ where $S$ is the solution of master equation and $\sigma$ is infinitesimally small is equivalent to the equation $Q(\sigma) =0$.

In the space $Sol$ of solutions to EM (in the zero locus 
of $Q$) we should identify solutions $x$ with  $x+\delta x$ where $\delta x^i=Q^i(x+\delta)-Q^i(x)$, where $\delta$ is 
infinitesimally small. The space obtained by means of this identification is denoted by $Sol/\sim$.
\footnote{More geometrically 
we can say that on the zero locus $Sol$ of $Q$ there exists a foliation $\F_Q$. The tangent 
vectors to $Sol$ can be identified with the kernel $\Ker \frac{\delta Q^i}{\delta x^j}$ of $Q$; the leaves of $\F_Q$ are 
tangent to the image  $\Im \frac{\delta Q^i}{\delta z^j}$. We identify two solutions belonging to the same leaf , hence 
$Sol/\sim$ can be considered as the space of leaves of the foliation. Notice that in most cases the foliation on the space  of solutions is singular (the dimension of the kernel varies).
}

A classical system has many equivalent descriptions in BV-formalism. The simplest way to see this is to notice that a 
system with coordinates $(y^1,\dots,y^n,\xi_1,\dots,\xi_n),$  symplectic form $dy^id\xi_i$ and action functional $a_{ij}
y^iy^j$ is physically trivial. Here $\xi_i$ and $y^i$ have opposite parities and the matrix  $a_{ij}$ is nondegenerate.

Consider two Q-manifolds $(M,Q)$ and $(M',Q')$. A map $f:M\rightarrow M'$ is called a $Q$-map if it agrees with action 
of $Q$'s (i.e. $Q f^*=f^*Q$ where $f^*$ is the homomorphism $C^{\infty}(M')\to C^{\infty}(M)$ induced by the map $f$). Such 
a map induces a map of observables (a homomorphism of cohomology groups $\rH(C^{\infty}(M'),Q')\rightarrow \rH(C^{\infty}(M),Q)$). If $f
$ defines  an isomorphism between  spaces of observables we say that $f$ is a quasi-isomorphism.Under some 
additional requirements  this isomorphism implies isomorphism of spaces of solutions $Sol/\sim.$ Quasi-isomorphism 
should be considered as isomorphism of classical physical systems. However for Lagrangian systems one should modify 
the definition of physical equivalence, requiring that quasi-isomorphism is compatible with symplectic structure in some 
sense.

Let us consider the Taylor series decomposition \[Q^a(x)=\sum_{b_1,\dots,b_n}Q^a_{b_1,\dots,b_n}x_1^{b_1}\dots x_n^
{b_n}\] of the coefficients of the vector field $Q=\sum Q^a\frac{\partial}{\partial x^a}$ in the neighborhood of the critical 
point. Here $x^i$ are local coordinates in the patch, the critical point is located at $x=0$. The coefficient $Q^a_
{b_1,\dots,b_n}$ of this expansion specifies an algebraic $n$-ary operation $\psi_n(s_1,\dots,s_n)$ on $\Pi T_0$ (on the 
tangent space with reversed parity  at $x=0$). A set of some  quadratic relations on $\psi_n$ is a corollary of  Master equation   $[Q,Q]=0$. A collection  $\{\psi_n\}_{n=1}^{\infty}$ that satisfies these quadratic relations  
specify a structure of L$_{\ity}$ algebra on $T_0$ (see Appendix \ref{AppendixA} for more detail). One can say that  L$_
{\ity}$ algebra is a formal $Q$-manifold.\footnote{We define a formal manifold saying that functions on it are power  series with respect to $n
$ commuting and $m$ anticommuting variables. Notice that this definition is not standard; other terms for the same notion are "germ of a manifold" and "infinitesimal manifold". From the other side  Sullivan suggested to use the term "formal manifold" in completely different setting.}

In the case when the only nonzero coefficients are $Q_b^a$ and $Q_{b_1,b_2}^a$ the corresponding L$\ity$ algebra can 
be identified with a differential graded Lie algebra.  The tensor $Q_b^a$ corresponds to the differential and $Q_{b_1,b_2}
^a$ to the bracket.

An L$\ity$ homomorphism of L$\ity$ algebras is defined as a $Q$-equivariant map between the corresponding formal $Q$-manifolds. We can use 
this notion to define an L$\ity$ action of a Lie algebra on a $Q$-manifold $M$. Conventional action of a Lie algebra is a 
homomorphism of this Lie algebra into Lie algebra $Vect (M)$ of vector fields  on $M$. 
 L$\ity$ action  of a Lie algebra $\g$ on $(M,Q)$ is an  L$\ity$ homomorphism of $\g$ to the differential 
graded algebra $(Vect (M),Q)$. This definition is rather inexplicit to say the least. A more direct definition shall be introduced presently.

The action of a Lie algebra is specified by vector fields $q_{\alpha}$, corresponding to generators  $e_{\alpha}
$ of $\g$. The generators obey relations $[e_{\alpha},e_{\beta}]=f_{\alpha\beta}^{\gamma}e_{\gamma}$, where $f_
{\alpha\beta}^{\gamma}$ are the structure constants of $\g$ in the basis $e_{\alpha}$.  A weak action of $\g$ 
requires that this relation is valid up to $Q$-exact terms:
\begin{equation}\label{E:fjsjs}
[q_{\alpha},q_{\beta}]=f_{\alpha\beta}^{\gamma}q_{\gamma}+[Q,q_{\alpha\beta}]
\end{equation}
Even for a weak action  we have a genuine Lie algebra action on observables and on $Sol/\sim$.

To define an L$\ity$ action of Lie algebra $\g$ we need not only
 $q_{\alpha},q_{\alpha\beta}$ , but also their higher analogs $q_{\alpha_1\dots\alpha_i}$ obeying the relations similar to 
(\ref{E:fjsjs}). This can be formalised as follows. One can  consider $q_{\alpha_1\dots\alpha_i}$ as components of linear 
maps  
\begin{equation}\label{E;fshdg}
q^i:\Sym^i(\Pi \g)\rightarrow \Pi Vect (M)
\end{equation}
They  can be assembled  into a vector field $q$ on $\Pi \g\times M$. A choice of a basis in $\g$  defines coordinates on  
$\Pi \g$. In such coordinates the $i$-th Taylor coefficient coincides with the map $q^i$.  The coordinates $c^{\alpha}$ on $\Pi \g$ can be identified with ghost variables of the Lie algebra $\g$. One can    consider $q$ 
as a vector field on $M$ depending on ghost variables. 

Let us introduce a  super-commutative  differential algebra $C^{\bullet}(\g)$ as the algebra of polynomial functions in
ghost variables $c^{\alpha}$ with the differential 

\begin{equation}\label{E:dg}
d_{\g}=\frac{1}{2}f^{\gamma}_{\alpha \beta}c^{\alpha}c^{\beta} \frac{\partial}{\partial c^{\gamma}}.
\end{equation}
Odd ghosts correspond to even generators, even ghosts correspond to odd generators. The Lie group cohomology with trivial coefficients  is 
defined as the cohomology of $d_{\g}$.
 
The collection (\ref{E;fshdg}) defines a L$\ity$ action if the ghost dependent vector field  $q$ satisfies
\begin{equation}\label{E:fadafdh}
d_{\g}q+[Q,q]+\frac{1}{2}[q,q]=0.
\end{equation}

Notice, that instead of $q$ we can consider ghost dependent vector field $\tilde q= Q+q$; in terms of this field (\ref 
{E:fadafdh}) takes the form \begin{equation}
\label{tq}
d_{\g}\tilde q+\frac{1}{2}[{\tilde q},{\tilde q}]=0
\end{equation}
The notion of L$\ity$ action is a particular case of the notion of L$\ity$ module. Recall that a $\g$-module where $\g$ is a 
Lie algebra can be defined as 
as a homomorphism of $\g$ in the Lie algebra of linear operators acting on vector space $N$. (In other words a $\g$-module  is the same as linear representation of $\g$.) If $N$ is a complex
the space of linear operators on $N$ is a differential 
Lie algebra. A structure of L$\ity$ $\g$ module on $N$ is an L$\ity$ homomorphism of $\g$ into this differential Lie 
algebra. This structure can be described as a polynomial function $q$ of ghosts $c^{\alpha}$ taking values in the space 
of linear operators on $N$  and obeying relation:
\begin{equation}
\label{LN}
d_{\g}q+[Q,q]+\frac{1}{2}[q,q]=0.
\end{equation}
As usual we denote by $Q$ the differential in $N$ (cf. (\ref {E:fadafdh})).

We can define cohomology $ \rH_{\g}^{\bullet} (N)=\rH^{\bullet}(\g, N)$  of the Lie algebra $\g$ with coefficients in L$\ity$ $
\g$-module $N$ to be the  cohomology of the differential

\begin{equation}
\label{CN}
d_c= d_{\g}+q+Q=\frac{1}{2}f^{\gamma}_{\alpha \beta}c^{\alpha}c^{\beta} \frac{\partial}{\partial c^{\gamma}}+\sum _k \frac
{1}{k!}q_{\alpha _1,...,\alpha _k}  c^{\alpha _1}\cdots c^{\alpha _k}+Q.
\end{equation}
acting on the space of $N$-valued functions of ghosts (i.e. on the tensor product  $C^{\bullet}(\g)\otimes N$).
It follows immediately from (\ref {LN}) that $d_c$ is a differential. Conversely, if the expression (\ref {CN}) is a differential then
$q$ specifies an L$\ity$ action.

To define homology of the Lie algebra $\g$ with coefficients in L$\ity$ module $N$ we use the differential $d_h$ acting 
on $N$-valued polynomial  functions
of ghost variables $c_{\alpha}$ (on the tensor product $\Sym \Pi \g\otimes N$). This differential can be obtained from 
$d_c$ by means of formal Fourier transform in ghost variables, i.e. a  substitution of the derivation with respect to $c_{\alpha} $ instead of multiplication by $c^{\alpha}$
and of the multiplication by $c_{\alpha}$ instead of derivation with respect to $c^{\alpha}$:
\begin{equation}
\label{HN}
d_h=\frac{1}{2}f^{\gamma}_{\alpha \beta}c_{\gamma}\frac {\partial }{\partial c_{\alpha}} \frac{\partial}{\partial c_{\beta}}+
\sum _k \frac{1}{k!}q_{\alpha _1,...,\alpha _k} \frac {\partial }{\partial c_{\alpha_1}} \cdots \frac {\partial }{\partial c_{\alpha 
_k}}+Q
\end{equation}
 
In light of this discussion  the definition of a Hamiltonian L$\ity$ action  on odd symplectic manifold $M$ is obvious.  In the formula (\ref
{E:fadafdh}) we  replace the vector field $q$ by a function and the commutator  by the Poisson bracket.
 A Hamiltonian L$\ity$ symmetry of classical BV action functional $S$ can be specified by a function of ghosts and fields 
(i.e. by an element   $\sigma \in C^{\bullet}(\g)\otimes C^{\infty}( M)$).
 This element should obey the equation
 \begin{equation}
\label{litys}
d_{\g}\sigma+\{S,\sigma\}+\frac{1}{2}\{\sigma,\sigma\}=0.
\end {equation}
Introducing a function $\hat {S}=\sigma+ S$ we can rewrite (\ref {litys}) in the form
\begin{equation}
\label{lityss}
d_{\g}{\hat S}+\frac{1}{2}\{{\hat S},{\hat S}\}=0.
\end{equation}
Mathematically it is very natural to let both ingredients of a BV package $(Q,\omega)$ to be ghost dependent. Such ghost dependent pair $(\hat{Q},\hat{\omega})$ satisfies a  block of axions and give rise to  constructions analogous to ghost independent setup:
\begin{equation}
\label{hhq}
\hat Q^i\hat \omega _{ij}=\frac{\partial \hat S}{\partial x^i}.
\end{equation}
In this formula ghost variables are parameters.
The condition $Q\omega=0$ extends to the condition $d_c\hat \omega=0$ where in $q_{\alpha _1,...,\alpha _k}$ and $Q$ in  (\ref{CN}) act on $\hat{\omega}$ by Lie derivatives; the operator $d_{\g}$ acts only on ghost parameters. In addition we  assume that ghost coefficients of $\hat \omega$ are de Rham closed and and ghost independent term is nondegenerate. We shall call ghost dependent two-forms, that satisfy last two conditions {\it closed} and {\it nondegenerate} respectively. 

In many interesting situations an action of a Lie algebra on shell (on the space $Sol/\thicksim$) can be lifted to an L$\ity$ 
action off shell. Conversely any L$\ity$  action of Lie algebra (or, more generally, any weak action) on a Q-manifold (off-
shell action) generates ordinary Lie algebra action on shell.

{\footnote {More generally, we can consider L$\ity$ $\g$-module $N$. We say that the structure of L$\ity$
module on $N$ is Hamiltonian if $N$ can be equipped with $\g$-invariant inner product. (We say that the inner product is 
$\g$-invariant if  $q$ specifying L$\ity$ $\g$-action takes values in the Lie algebra of linear operators on $N$ 
that preserve the inner product.) If we have an  odd symplectic $Q$- manifold $M$ we can take as $N$
the L$\ity$ algebra constructed as the Taylor decomposition of $Q$ in Darboux coordinates
in the neighborhood of a point belonging to the zero locus of $Q$. This algebra is equipped with odd inner product 
coming from the odd symplectic form.  A Hamiltonian  L$\ity$ action on $M$ generates a Hamiltonian L$\ity$ action on 
$N$.}}


Let us come back to the general theory of deformations in BV-formalism.
Recall that infinitesimal deformations of solution to the classical Master equation (\ref{E:fggdsfjh}) (of classical Lagrangian system in BV formalism) are labeled by 
observables (by cohomology $\rH(C^{\infty}(M),Q)$ of $Q$ in the space $C^{\infty}(M)$). Every deformation of action $S$ induces a 
deformation of $Q$ and we have a homomorphism of corresponding cohomology groups
$\rH(C^{\infty}(M),Q)\to \rH(Vect (M), Q)$.

Let us analyze  classification of infinitesimal deformations of a classical system compatible  with symmetries of the system. We 
assume that the system is described by an odd vector field $Q$ obeying $[Q,Q]=0$ on a supermanifold $M$ and  the 
Lie algebra of symmetries $\g$ acts in   L$\ity$ fashion on $M$. Then the complex $(Vect (M), Q)$ as L$\ity$ $\g$-module. We would like to deform simultaneously the vector field $Q$ and the L$\ity$ action specified by the ghost dependent vector field $q$. {\it We shall show that the infinitesimal deformations are 
classified  by elements of cohomology $\rH^{\bullet}(\g, (Vect (M), Q))$ of  Lie algebra $\g$ with coefficients
in differential L$\ity$ $\g$-module $(Vect (M), Q)$.} To prove this statement we notice that  $Q$ and $q$ are combined 
in the ghost dependent
vector field $\tilde q$, hence the deformation we are interested in can be considered as the deformation of $\tilde q$ that 
preserves the relation (\ref {tq}). In other words this deformation should obey
\begin{equation}
\label{dtq}
d_{\g}\delta {\tilde q}+[{\tilde q},\delta {\tilde q}]=0.
\end{equation}
This condition means that  $\delta {\tilde q}$ is a $d_c$-cocycle  (\ref 
{CN})).
It is easy to see that cohomologous cocycles specify equivalent deformations and $\rH^{\bullet}(\g, (Vect (M), Q))$ is in one-to-one correspondence with the infinitesimal deformations of our objects.

It is important to emphasize that commutation relations of the new symmetry  generators are deformed (even if we have 
started with genuine action of $\g$ we can obtain a weak action after the deformation).
 Nevertheless   on shell, i.e. after restriction to $Sol/\sim$,  commutation relations do not change. 

As we have said the cohomology $\rH^{\bullet}(\g, (Vect (M), Q))$ describes L$\ity$ deformations of $Q$ and L$\ity$ 
action. Notice, that two different L$\ity$ actions can induce the same Lie algebra action on shell.  It is easy to see that 
only the ghost number one components enter in the expressions for generators of Lie algebra symmetries on shell.

It is important to emphasize  that analyzing deformations  in
BV formulation we can chose any of physically equivalent
classical systems ( the cohomology we should calculate is invariant with respect to quasi-isomorphism).

We have analyzed the deformations of $Q$ (of equations of motion) preserving the symmetry.
Very similar consideration can be applied in Lagrangian formalism. In this case we start with the  functional $\hat S=
\sigma +S$ combining the classical BV functional $S$ and Hamiltonian L$\ity$ symmetry. We should deform this 
functional
preserving the relation (\ref {lityss}). We keep symplectic form undeformed.  We see that
that the infinitesimal deformation obeys
$$d_{\g}\delta {\hat S}+\{\hat S, \delta \hat S\}=0.$$ Interpreting this equation as a cocycle condition we obtain the 
following statement.

\begin{proposition} Let us consider a BV action functional $S$ on an odd symplectic manifold $M$ together with
Hamiltonian L$\ity$ action of Lie algebra $\g$ (i.e.
with a  functional $\sigma$ of fields and ghosts
such that  $\hat S=\sigma +S$ obeys (\ref {lityss})). Then the algebra $C^{\infty}( M)$
can be considered as a differential L$\ity$ $\g$-module (the differential is defined as a Poisson bracket with $S$). 
Infinitesimal deformations of
BV-action functional and Hamiltonian L$\ity$ action  are governed by the Lie algebra cohomology of $\g$ with coefficients 
in this module.
\end{proposition}

\begin{proposition} \label {2} 
Let us consider a  supermanifold $M$ together  with an odd vector field $Q$ obeying $[Q,Q]=0$ and with L$\ity$ action of a Lie algebra $\g$  on $(M,Q)$ (in other words we have a classical system with symmetry Lie algebra $\g$). Let us assume that the manifold $M$ is equipped with a non-degenerate ghost-dependent  closed two-form $\hat \omega$ that is compatible with the L$\ity$ action(i.e. $d_c\hat{\omega}=0$).  
 Then the algebra $C^{\infty}( M)$ is an L$\ity$ $\g$-module. 
The infinitesimal deformations of
$Q$ together with Hamiltonian L$\ity$ action preserving $\hat\omega$ are governed by the Lie algebra cohomology of $\g$ with coefficients 
in this module.
\end{proposition}
Last two propositions are related because all nondegenerate  ghost-dependent closed odd two-forms are diffeomorphic, if we allow ghost-dependent  diffeomorphisms. Such diffeomorphism can be used to eliminate ghost dependence of the form.

Our experience with pure spinor formalism dictates a need for extension of above theorem to manifolds with a degenerate two-form.  As in the strictly symplectic  variant of this definition every hamiltonian vector field $\xi$ ( a vector field  preserving the form $\omega$)  defines a Hamiltonian $F$ by the formula (\ref{ham}) (of course, the Hamiltonian is defined up to an additive constant).This construction of the Hamiltonian defines a linear map from the algebra of hamiltonian vector fields $Ham$ to the space of functions $C^{\infty}(M)/\mathbb{C}$:
\[H:Ham\rightarrow C^{\infty}(M)/\mathbb{C}.\]
For degenerate form $\omega$ the map $H$ is not invertible and we do not expect  to have a Poisson structure on the algebra of functions. 
The commutator with $Q$ defines a differential  on $Ham$. Infinitesimal deformations of $Q$ are in bijections with cohomology classes of the complex $(Ham,Q)$.  A closed two-form $\omega$ is   homologically non-degenerate
if it defines a nondegenerate pairing on cohomology $\rH(T_x, Q)$ of the tangent complex $(T_x, Q)$ at each  critical point $x$ of $Q$).   In this case we expect  that $H$ is a quasiisomorphism (i.e. it induces an isomorphism $\rH(Ham,Q)\rightarrow \rH(C^{\infty}(M)/\mathbb{C},Q)$, which enables us to transfer the Lie algebra structure onto $\rH(C^{\infty}(M)/\mathbb{C},Q)$). The reason to believe in this statement is a validity of a  local statement at a critical point $x$ of $Q$. More precisely, the statement is correct if we replace $M$ by a germ of $M$ at a critical point . As it has already been pointed out this framework is equivalent  to  the framework of $L\ity$-algebras with odd inner product). Hence in homologically non-degenerate case deformations are classified by  $\rH(C^{\infty}(M))$
(at least for  formal $Q$-manifolds).

These statements can be generalized to a $\g$-equivariant setup .
We say that a two-form  $\hat \omega$ homologically non-degenerate if its ghost-free term is homologically non-degenerate.
The proposition \ref{2} is still valid if the closed two-form  $\hat \omega$  on a formal $Q$-manifold is  homologically non-degenerate. Identifying formal $Q$-manifolds with $L\ity$-algebras we can use proposition \ref{2} to classify deformations of $L\ity$ algebras with homologically non-degenerate odd inner product (i.e. with inner product that is non-degenerate on homology).

Let us describe some formulations of 10D SUSY YM in BV formalism. For simplicity we shall restrict ourselves to the theory 
reduced to a point.

In component formalism besides fields $A_i,\chi^{\alpha}$, antifields $A^*_i,\chi^*_{\alpha}$ we have ghosts $c$ and  
anti-fields for ghosts $c^*$ (all of them are $n\times n$ matrices).  The BV action functional has the form 
\begin{equation}\label{E:bvsym}
\L_{BVSYM}=\tr\bigg(\frac{1}{4}F_{ij}F_{ij}+\frac{1}{2}\Gamma^i_{\alpha\beta}\chi^{\alpha}\nabla_i\chi^{\beta}+\nabla_ic 
A^*_i+\chi^{\alpha}[c,\chi^*_{\alpha}]+\frac{1}{2}[c,c]c^*\bigg)
\end{equation}
The corresponding vector field $Q$ is given 
\begin{equation}\label{E:yunbr}
\begin{split}
&Q(A_i)=-\nabla_ic\\
&Q(\psi^{\alpha})=[c,\psi^{\alpha}]\\
&Q(c)=\frac{1}{2}[c,c]\\
&Q(c^*)=\sum_{i=1}^{10}\nabla_iA^{*i}+
\sum_{\alpha}[\psi^{\alpha},\psi_{\alpha}^{*}]+[c,c^*]\\
&Q(A^{*m})=-\sum_{i=1}^{10}\nabla_iF_{im}
+\frac{1}{2}\sum_{\alpha 
\beta}\Gamma_{\alpha \beta}^m[\psi^{\alpha},\psi^{\beta}]-[c,A^{*m}]\\
&Q(\psi_{\alpha}^{*})=-\sum_{i=1}^{10}\sum_{\beta}\Gamma_{\alpha 
\beta}^i\nabla_i\psi^{\beta}
-[c,\psi_{\alpha}^{*}]\\
\end{split}
\end{equation}


Another possibility is to work in the formalism of pure spinors.

Let $\rS=\mathbb{C}^{16}$ be a 16-dimensional complex vector space with coordinates $\lambda^1,\dots,\lambda^{16}$.
Denote by $\mathcal{C}$ a cone of pure spinors in $\rS$ defined by equation 
\begin{equation}\label{E:dgdsjg}
\Gamma^i_{\alpha\beta}\lambda^{\alpha}\lambda^{\beta}=0
\end{equation}
and by  $\Ss=\mathbb{C}[\lambda^1,\dots,\lambda^{16}]/\Gamma^i_{\alpha\beta}\lambda^{\alpha}\lambda^{\beta}$ the 
space of polynomial functions on $\mathcal{C}$. 

The fields in this formulation are elements $A(\lambda,\theta)\in \Ss\otimes \Lambda [\theta^1,\dots,\theta^{16} ]\otimes 
\Mat_N$; they can be considered as $\Mat_N$-valued polynomial super-functions on $C\mathcal{Q}\times \Pi S$.  We define 
differential $d$ acting on these fields by the formula $d=\lambda ^{\alpha}\frac{\partial}{\partial \theta ^{\alpha}}$. Using 
the terminology of Section \ref{2.1} we can identify the  space of fields with tensor product of reduced Berkovits algebra  
$B_0$ and $\Mat_N$.

The vector field $Q$ on the space of fields  is given by the formula
\begin{equation}\label{E:vectorfield}
\delta_{Q}A=dA+\frac{1}{2}\{A,A\}.
\end{equation}
This vector field specifies a classical system quasi-isomorphic to
 the classical system corresponding to the action functional
( \ref {E:bvsym}) (see \cite {MSch}).
               
An odd two- form on the space of fields is given by the formula:
\begin{equation}
\label {tr}
\omega(\delta A_1,\delta A_2)=\tr(\delta A_1\delta A_2)
\end{equation}
The trace $\tr$  is nontrivial only on $\Ss_3\otimes \Lambda^5[\theta^1,\dots,\theta^{16}]$. Denote $\Gamma$ be the 
only $\Spin(10)$ invariant element in $\Ss_3\otimes \Lambda^5[\theta^1,\dots,\theta^{16}]$. Let $p$ be the only $\Spin
(10)$-invariant projection on the span $<\Gamma>$. Then $p(a)=\tr (a)\Gamma$. This definition fixes $\tr$ up to a 
constant. (The trace at hand was introduced in \cite{Berkovits}, where more explicit formula was given.)
Notice that $\omega$ is a degenerate closed two-form. However, it is homologically nondegenerate; this allows us to study infinitesimal Lagrangian deformations using a generalization of proposition \ref {2}.

In  BV-formalism equations of motion      can be obtained from the action functional $$S(A)=\tr(AdA+\frac{2}{3}A^3),$$
obeying the classical Master equation $\{S,S\}=0$ (recall that we factorize the space of fields with respect to $\Ker 
\omega$ and $S$ descends to the quotient). The vector field  $Q$ specified by the formula (\ref{E:vectorfield}) 
corresponds to this action functional.

The BV formulation of unreduced SYM theory in
terms of pure spinors is similar. The basic field
$A(x,\lambda, \theta)$ where $x$ is a ten-dimensional vector is matrix-valued. The differential $d$ is defined by the 
formula $d=\lambda ^{\alpha}(\frac{\partial}{\partial \theta ^{\alpha}}+\Gamma ^i _{\alpha \beta} \theta ^{\beta}\frac
{\partial}{\partial x^i}).$ In the terminology of Section \ref{2.1} the space of fields is a tensor product of Berkovits algebra $B$ 
and $\Mat_N$. The expressions for action functional and odd symplectic form remain the same, but $\tr$ includes 
integration over ten-dimensional space. 

 {\footnote {To establish the relation to the  superspace formalism we recall that in $(10|16)$ dimensional superspace $
(x^n,\theta^{\alpha})$ SYM equations together with constraints  can be represented in the form
\begin{equation}\label{E:dafdsf}
F_{\alpha\beta}=0
\end{equation}
where $F_{\alpha\beta}=\{\nabla_{\alpha},\nabla_{\beta}\}-\Gamma^i_{\alpha\beta}\nabla_{i}$, $\nabla_{\alpha}=D_
{\alpha}+A_{\alpha}$, $D_{\alpha}=\frac{\partial}{\partial \theta^{\alpha}}+\Gamma^i_{\alpha\beta}\theta^{\beta}\frac
{\partial}{\partial x^i}$.
It follows from these equations that  the covariant derivatives $\nabla (\lambda)=\lambda ^{\alpha}\nabla _{\alpha}$ obey 
$[\nabla (\lambda),\nabla (\lambda)]=0$ if $\lambda$ is a pure spinor. This allows us interpret Yang-Mills fields as degree 
one components of $A(x,\theta,\lambda)$.
Degree  zero components of $A(x,\theta,\lambda)$ correspond to ghosts. Degree two components correspond to 
antifields, degree three to antifields for ghosts. Components of higher degree belong to the kernel of $\omega$ and can be 
disregarded
(see \cite {Berkovits} and \cite {MSch} for detail).}}

Notice that that in component version of
BV formalism  the standard supersymmetry algebra acts on shell, but off shell we have weak action of this algebra 
(commutation relations are satisfied up to $Q$-trivial terms). In pure spinor formalism we have genuine action of 
supersymmetry algebra, but  the form $\omega$ is not invariant with respect to supersymmetry transformations. 
(However, the corrections to this form are $Q$-trivial.)

We shall show that  the weak action of supersymmetry algebra can be extended to L$\ity$ action (Appendix C). 
Moreover, in Appendix D we shall prove that  for appropriate choice of this action it will be compatible with odd symplectic 
structure.

Let us apply our general considerations to 10D SUSY YM
reduced to a point. In Section \ref{S:fdgdfgmjl} we described SUSY deformations
of this action functional in component formalism. Now we shall
rewrite these deformations in BV formalism. Moreover,
we shall be able to write down also the deformed supersymmetry.

Let us start with BV description of  the theory based on
the Lagrangian (\ref{E:bvsym}). A vector field $\xi$ on the underlying space is completely characterized by the values  on the generators of the algebra of functions. We shall refer to these values as to components. As the generators can 
be naturally combined in matrices, the components of the vector fields are also matrix-valued.   The vector fields  of 
sypersymmetries $\theta_{\alpha}$ in the matrix space description have the following components (we omit matrix 
indices ):  \[\theta_{\alpha}A^i=\Gamma^i_{\alpha\, \beta}\chi^{\beta} \] \[\theta_{\alpha}\chi^{\beta}=\Gamma_{\alpha}^
{\beta\, ij}[A_i,A_j].\]
The component  description of  the vector fields $D_i$ and $G_{\alpha}$ is \[D_iA_j=[A_i,A_j],\quad  D_i\chi^{\alpha}=
[A_i,\chi^{\alpha}]\] and 
\[G_{\alpha}A_j=[\chi_{\alpha},A_j],\quad  G^{\alpha}\chi^{\beta}=[\chi^{\alpha},\chi^{\beta}].\]
In this setup  we have the identities
\begin{equation}
\begin{split}
& [\theta_{\alpha}, \theta_{\beta}]-\Gamma_{\alpha\beta}D_i=[Q,\eta_{\alpha\beta}],\\
&[\theta_{\alpha},D_i]-\Gamma_{\alpha\beta i}G_{\chi^{\beta}}=[Q, \eta_{\alpha i}]
\end{split}
\end{equation}

Here 
\[\eta_{\alpha\beta}\chi^{\gamma}=2P_{\alpha\beta}^{\gamma\delta}\chi_{\delta}^*,
\]
\[\eta_{\alpha i}A_j=C_{\alpha i}^{\beta j}\chi^*_{\beta},\quad \eta_{\alpha i}\chi^{\beta}=-C_{\alpha i}^{\beta j}A^*_j\]

The tensors $P_{\alpha\beta}^{\gamma\delta}$ and $C_{\alpha i}^{\beta j}$ are proportional to $\Gamma_{\alpha\, \beta}
^{i_1,\dots,i_5}\Gamma^{\gamma\, \delta}_{i_1,\dots,i_5}$ and to $\Gamma_{\alpha i}^{\beta j}$ respectively.
We have described in Section \ref{S:fdgdfgmjl}  an infinite family of SUSY deformations (\ref{E:fadjghjd}). 
It is easy to write down the  terms $q$ and $q_{\alpha}$ in the corresponding cocycle.  It is obvious that $q$ is a 
Hamiltonian vector field corresponding to the functional $\dL$ given by the formula (\ref{E:fadjghjd}).  To find  the 
functional $\sigma _{\alpha} $ generating the Hamiltonian vector field $q_{\alpha}c^{\alpha}$   we should calculate $
\theta_{\alpha}\dL$ and use (23).The calculation of $\theta_{\alpha}\dL$ repeats the proof of
the supersymmetry of the deformation (9) and leads to the following result:  
 \begin{equation}
\begin{split}
&\theta_{\alpha}\dL=\tr(\theta_{\alpha}\theta_{1}\dots \theta_{16}G)=\\
&\tr(\sum_{\gamma=1}^{\alpha-1}(-1)^{\gamma-1}\Gamma_{\alpha\gamma}^k\theta_{1}\dots \theta_{\gamma-1}\TQ(\eta_
{\alpha\gamma})\theta_{\gamma+1}\dots \theta_{16}G)+\\
&+\tr(\sum_{\gamma=1}^{\alpha-1}(-1)^{\gamma-1}\Gamma_{\alpha\gamma}^k\theta_{1}\dots \theta_{\gamma-1}D_k
\theta_{\gamma+1}\dots \theta_{16}G)+\\
&+\frac{1}{2}\tr((-1)^{\alpha-1}\Gamma_{\alpha\alpha}^k\theta_{1}\dots \theta_{\alpha-1}\TQ(\eta_{\alpha\gamma})\theta_
{\alpha+1}\dots \theta_{16}G)+\\
&+\frac{1}{2}\tr((-1)^{\alpha-1}\Gamma_{\alpha\alpha}^k\theta_{1}\dots \theta_{\alpha-1}D_k\theta_{\alpha+1}\dots 
\theta_{16}G)=\\
&=Q\sum_{\gamma=1}^{\alpha-1}\Gamma_{\alpha\gamma}^k\tr(\theta_{1}\dots \theta_{\gamma-1}\eta_{\alpha\gamma}
\theta_{\gamma+1}\dots \theta_{16}G)+\\
&+Q\frac{1}{2}\Gamma_{\alpha\alpha}^k\tr(\theta_{1}\dots \theta_{\alpha-1}\eta_{\alpha\alpha}\theta_{\alpha+1}\dots 
\theta_{16}G)+\\
&+Q\sum_{\gamma=1}^{\alpha-1}\sum_{\gamma'=1}^{\gamma-1}(-1)^{\gamma+\gamma'}\Gamma_{\alpha\gamma}^k\tr
(\theta_{1}\dots\theta_{\gamma'-1} \eta_{\alpha k}\theta_{\gamma'+1}\dots \hat\theta_{\gamma}\dots \theta_{16}G)+\\
&+Q\frac{1}{2}\sum_{\gamma=1}^{\alpha-1}(-1)^{\alpha+\gamma}\Gamma_{\alpha\alpha}^k\tr(\theta_{1}\dots \theta_
{\gamma-1} \eta_{\alpha k}\theta_{\gamma+1}\dots \hat\theta_{\alpha}\dots \theta_{16}G)+\\
&+\sum_{\gamma=1}^{\alpha-1}\pr{x^k}\tr((-1)^{\gamma-1}\Gamma_{\alpha\gamma}^k\theta_{1}\dots \hat\theta_
{\gamma}\dots \theta_{16}G)+\\
&+\frac{1}{2}\pr{x^k}\tr((-1)^{\alpha-1}\Gamma_{\alpha\alpha}^k\theta_{1}\dots \hat\theta_{\alpha}\dots \theta_{16}G).
\end{split}
\end{equation}
As usual ``$\hat{\ }$'' marks the symbol that should be omitted in the formula.
If subscript in $\theta_{\gamma}$ is out of range $[1,16]$ then $\theta_{\gamma}$  must be omitted.
From this computation we conclude that for Hamiltonian of the vector field $q_{\alpha}$ as a function of $G$ is
\begin{equation}
\begin{split}
&\sum_{\gamma=1}^{\alpha-1}\Gamma_{\alpha\gamma}^k\tr(\theta_{1}\dots \theta_{\gamma-1}\eta_{\alpha\gamma}
\theta_{\gamma+1}\dots \theta_{16}G)+\\
&+\frac{1}{2}\Gamma_{\alpha\alpha}^k\tr(\theta_{1}\dots \theta_{\alpha-1}\eta_{\alpha\alpha}\theta_{\alpha+1}\dots 
\theta_{16}G)+\\
&+\sum_{\gamma=1}^{\alpha-1}\sum_{\gamma'=1}^{\gamma-1}(-1)^{\gamma+\gamma'}\Gamma_{\alpha\gamma}^k\tr
(\theta_{1}\dots \theta_{\gamma'-1}\eta_{\alpha k}\theta_{\gamma'+1}\dots \hat \theta_{\gamma}\dots \theta_{16}G)+\\
&+\frac{1}{2}\sum_{\gamma=1}^{\alpha-1}(-1)^{\alpha}\Gamma_{\alpha\alpha}^k\tr(\theta_{1}\dots \theta_
{\gamma-1}\eta_{\alpha k}\theta_{\gamma+1}\dots  \hat \theta_{\alpha}\dots\theta_{16}G)
\end{split}
\end{equation}

\section{Formal  SUSY  deformations}\label{E:Formal}

We have analyzed infinitesimal SUSY deformations of reduced and unreduced SUSY  YM theory.  One can prove that all 
of these deformations can  be extended to formal deformations (i.e. there exist  SUSY deformations
represented as formal series with respect to parameter $\epsilon$ and giving an arbitrary 
infinitesimal deformation in the first order with respect to $\epsilon$).  We shall sketch the proof of this fact in present 
section.

 We have seen in Section \ref{S:BV} that there is a large class of infinitesimal supersymmetric deformations that have a 
form $\theta_1\dots \theta_{16}G$. We shall start with the proof that all these infinitesimal deformations can be extended 
to formal deformations.

We shall consider  more general situation when we have any action functional in BV formalism that is invariant with 
respect to L$\ity $ action of SUSY. As follows from Appendix \ref{AppendixD} our considerations can be applied to ten-
dimensional SYM theory.

The SUSY Lie algebra has $m$ even commuting generators $X_1,\dots,X_m$ and $n$ odd generators
$\tau _1,\dots,\tau _n$ obeying relations
$$[\tau _{\alpha},\tau _{\beta}]=\Gamma^i_{\alpha \beta}X_i.$$

In the definition of L$\ity$-action of $\g$ we use the algebra $C^{\bullet}(\g)$ of functions of corresponding ghosts. In our 
case this algebra is the algebra  
\begin{equation}\label{E:Kalgebra}
K=\mathbb{C}[[t^1,\dots,t^n]]\otimes \Lambda[\xi^1,\dots,\xi^m].
\end{equation} The odd variables $\xi^1,\dots,\xi^m$ are the ghosts for even generators (space-time translations), the 
even variables are the ghosts for odd generators. The algebra $K$ is equipped with the differential \[d=\Gamma_{\alpha
\beta}^it^{\alpha}t^{\beta}\frac{\partial}{\partial \xi_i},\]  where $\Gamma_{\alpha\beta}^i$ are the structure constants of the 
supersymmetry algebra.

The  L$\ity$ action can be described by an element $\hat{S}\in A=K\otimes C^{\infty}(M)$, where $M$ is the space of 
fields (in other words
$\hat {S}$ is a function of ghost variables $t^i,\xi ^{\alpha}$ and fields).

The  equation (\ref{lityss}) for $\hat{S}$ takes the form 

\begin{equation}\label{E:linftyeq1}
d\hat{S}+\frac{1}{2}\{\hat{S},\hat{S}\}=0.
\end{equation}
A solution to this equation gives us a solution $S$ to the
BV Master equation (obtained if we assume that ghost variables are equal to zero) and L$\ity$ action of supersymmetries  
preserving $S$.
We would like to construct a formal deformation of such a solution, i.e. we would like to construct a formal power series $
\hat {S}(\epsilon)$ with
respect to $\epsilon$ obeying the equation (\ref {E:linftyeq1}) and giving the original solution for 
$\epsilon=0.$ We shall start with a construction of the
solution of the equation for infinitesimal deformation
\begin{equation}\label{di}
dH+\{\hat {S},H\}=0.
\end{equation}
If we know the solution of the equation (\ref{di})
for every $\hat {S}$ we  can find the deformation
solving the equation
\begin{equation}
\label{f}
\frac{d\hat{S}(\epsilon)}{d\epsilon}=H(\hat{S}(\epsilon)).
\end{equation}
To solve the equation (\ref{di}) we construct a family of functions $F^k$ defined by inductive
formula
\begin{equation}\label{E:dsfjejjd}
{F}^{k+1}=\frac{1}{t_{k+1}}\left(d_{k+1}{F}^k+\{\hat{S}^{k},{F}^k\}\right).
\end{equation}
where  $d_k$ is defined as  $\sum_{\alpha\beta\leq k}\Gamma_{\alpha\beta}^it^{\alpha}t^{\beta}\frac{\partial}{\partial \xi_i}
$. 
We assume that $F^k$ and $\hat {S}^k$ do not depend on
$t^{k+1},\dots,t^n$ and $\hat {S}^k$ coincides with
$\hat {S}$ if $t^{k+1}=\cdots =t^n=0.$
We impose also an initial condition  $F^0$ obeying $\{\hat {S}_i, F^0\}=0$ where $\hat{S}_i=\frac{\partial \hat{S}}{\partial 
\xi_i}|_{0}.$
We shall see that $F^n$ is a solution of the equation (\ref {di}); this allows us to take $H=F^n.$
To  prove this fact we should give geometric interpretation of (\ref {E:dsfjejjd}). First of all we notice that the solutions of 
(\ref {di}) are cocycles of the differential $d_S=d+\{\hat {S},\cdot\}$ acting on the algebra  $A$ of functions of ghosts 
$t^1,\dots,t^n, \xi ^1,\dots, \xi ^m$ and fields.
We consider differential ideal $I_k$ of this algebra  defined as set of functions that vanish
if $t^1=  \cdots =t^k=0$ (in other words, $I_k$ is generated by $t^1,\dots,t^k$) and the quotient $A_k$ of the algebra $A$ 
with respect to this ideal. The differential algebra $A_k$ can be interpreted as the algebra of functions depending on 
ghosts $t^1,\dots,t^k, \xi ^1, \dots, \xi ^m$ and fields. The inductive formula 
(\ref {E:dsfjejjd}) gives a map of $A_k$ into  $A_{k+1}$
that descends to cohomology.
To construct this map we notice that the embedding $I_{k+1}\subset I_k$ generates a short exact sequence
$$0\to I_k/I_{k+1}\to A_{k+1}\to A_k\to 0.$$
The ideal $I_k/I_{k+1}$ of the algebra $ A_{k+1}$
is generated by one element $t^{k+1}$. This means we can rewrite the exact sequence in the form
$$0\to A_{k+1}\to A_{k+1}\to A_k\to 0,$$
where the map $ A_{k+1}\to A_{k+1}$ is a multiplication by $t^{k+1}.$ The boundary map in
the corresponding exact cohomology sequence  gives (\ref {E:dsfjejjd}). The condition imposed on $F^0$ means that 
$F^0$ is a cocycle in $A_0.$

For every admissible $F^0$ we have constructed
$H(S)$ as a solution of (\ref {di}); we have used this solution to construct formal deformation by means of  (\ref {f}).

This fairly simple description of supersymmetric deformations has one obvious shortcoming. The Poincar\'e invariance is 
hopelessly lost in the the formula (\ref{E:dsfjejjd}) even if we start with Poincar\'e invariant ${F}^0$. This can be fixed if 
we work in the euclidean signature. The algebra $A$ contains a subalgebra  $A_{\SO(m)}$ of $\SO(m)$-invariant 
elements. The vector field $H({F}^0)$ , restricted on $A_{\SO(m)}$ can be replaced by \[H^{\SO(m)}=\frac{1}{\vol(\SO
(m))}\int_{\SO(m)}H^gdg\] - the  average of the $g$-rotated element $H$ over $\SO(m)$. It can be proved by  other 
means that $H^{\SO(m)}_{{F}^0}$ is nonzero if ${F}^0$ is Poincar\'e-invariant. The above prescription can be formulated 
also in more algebraic form where 
Euclidean signature is unnecessary. We  decompose $A$ into direct sum of irreducible representations of $\SO(m)$ and 
leave only
$\SO(m)$ invariant part of $H$.

Let us make a connection with Section \ref{S:fdgdfgmjl}.

We start with  identifications. The odd symplectic manifold $M$  coincides with the space of fields in the maximally super-
symmetric Yang-Mills theory in  Batalin-Vilkovisky formalism (we can consider both reduced case when $n=16, m=0$ 
and unreduced case when $n=16, m=10$). It can be shown  that the supersymmetry action can be extended to an L$\ity
$ action, whose generating function satisfies equation (\ref{E:linftyeq1}); see Appendix \ref{AppendixD}.

 Let us start with a Poincar\'e invariant  $F^0=G$ as described in Section \ref{S:fdgdfgmjl}.  The l'H\^{o}pital's rule applied 
to $H={F}^n$ shows that its   leading term coincides with (\ref{E:fadjghjd}).  This means that infinitesimal  SUSY 
deformations of the form $\tr \theta _1\cdots\theta _{16}G$
can be extended to formal deformations. In reduced case this logic can be applied to arbitrary Poincar\'e invariant $G$, in 
unreduced case we should consider local gauge covariant $G$ to obtain SUSY deformation. 

There exists only one infinitesimal deformation $\delta \mathcal{L}_{16}$ that does not have the form $\tr \theta _1\cdots
\theta _{16}G$ (Theorem \ref{T:reduced} and Theorem \ref{E:theorem10}). {\footnote { It is better to say that every 
infinitesimal deformation can be represented as linear combination of $\delta \mathcal{L}_{16}$ and $\tr \theta _1\cdots
\theta _{16}G$. }} One can prove that this deformation also can be extended to formal deformation together with
L$\ity$ action of SUSY algebra (\ref{E:hfdfjfhjd})  {\footnote {Notice that superstring theory gives a formal SUSY 
deformation of SYM theory that corresponds to infinitesimal deformation represented as linear combination of $\delta 
\mathcal{L}_{16}$  (with non-zero coefficient) and $\tr \theta _1\cdots\theta _{16}G$. If we were able to prove that that 
SUSY action extends to L$\ity$ action we could use this deformation to extend all infinitesimal deformations in reduced 
case.}}Constructing formal deformations of this formal deformation we obtain that in the reduced case all infinitesimal 
deformations can be extended to formal ones.

We have noticed in Section \ref{S:fdgdfgmjl} that for $G$ of the form $\Delta=
\nabla _i\nabla _i$ the expression $\tr \theta _1\cdots\theta _{16}G$  generates a SUSY infinitesimal deformation of  
unreduced YM action functional.  One can prove that this deformation also can be extended to formal deformation, 
however, the above construction of formal deformation does not work in this case. The proof is based on the remark that 
infinitesimal deformation $A\tr \Delta$  can be applied to a formal deformation we constructed and it remains local.
\vskip.2in

\appendix
{\bf \LARGE  Appendices}

\vskip.2in
\section{L$\ity$ and A$\ity$ algebras}\label{AppendixA}

Let us consider a supermanifold equipped with an odd vector field $Q$ obeying $[Q,Q]=0$ (a $Q$-manifold).
Let us introduce a coordinate system in a neighborhood of a point of $Q$-manifold belonging to zero locus $Q$. Then 
the vector field $Q$ considered as a derivation of the algebra of formal power series can be specified by its action on the 
coordinate functions $z^A$:
\begin{equation}
Q(z^A)=\sum_{n}\sum \pm \mu^A_{B_1,\dots,B_n}z^{B_1}\dots z^{B_n}
\end{equation} 
 We can use tensors $\mu_n=\mu^A_{B_1,\dots,B_n}$  to define a series of operations. The operation $\mu_n$ has $n$ 
arguments;  it can be considered as a linear map $V^{\otimes n}\rightarrow V$ (here $V$ stands for the tangent space at 
$x=0$). However, it is convenient to change parity of $V$ and consider $\mu_n$ as a symmetric map $(\Pi V)^{\otimes 
n}\rightarrow \Pi V$.  It is convenient to add some signs in the definition of $\mu_n$. With appropriate choice of signs we 
obtain that operations $\mu_n$ obey some quadratic relations; by
definition the operators $\mu _n$ obeying these relations specify a structure of  L$\ity$ algebra  on $W=\Pi V$. We see 
that a point of zero locus of the  field $Q$ specifies an L$\ity$ algebra;  geometrically one can say that L$\ity$ algebra is 
a formal $Q$-manifold. (A formal manifold  is a space whose algebra of functions  can be identified with the algebra of 
formal power series. If the algebra is equipped with odd derivation $Q$, such that $\{Q,Q\}=0$ we have a structure of 
formal $Q$ manifold.) The considerations of our paper are formal. This means that we can interpret all functions of fields 
at hand as formal power series. Therefore  instead of working with $Q$-manifolds we can work with L$\ity$ algebras.

On a $Q$-manifold with odd symplectic structure we can choose the coordinates $z^1,\dots, z^n$ as Darboux 
coordinates,i.e. we can assume that the coefficients of symplectic form do not depend on $z$. Then the  L$\ity$ algebra 
is equipped with invariant odd inner product. 

Hence we can say that L$\ity$ algebra specifies  a classical system and L$\ity$ algebra with  invariant odd inner product 
specifies a Lagrangian classical system.

It is often important to consider $\mathbb{Z}$-graded L$\ity$-algebras (in BV-formalism this corresponds to the case 
when the fields are classified according to ghost number). {\it We assume  in this case that the derivation $Q$ raises the 
grading (the ghost number) by one.}

An  L$\ity$ algebra where all operations $\mu_n$  with $n\geq 3$ vanish can be identified with differential graded Lie 
algebra (the operation $\mu_1$ is the differential, $\mu_2$ is the bracket). An L$\ity$ algebra corresponding to Lie 
algebra with zero differential is $\mathbb{Z}$-graded.

For L$\ity$ algebra $\g=(W,\mu _n)$ one can define a notion of cohomology generalizing the standard notion of 
cohomology of Lie algebra. For example, in the case of trivial coefficients we can consider cohomology of $Q$ acting as 
a derivation of the algebra $\widehat {\Sym}(W^*)$ of formal functions on $W$ (of the algebra of formal series). In the 
case when the L$\ity$ algebra corresponds to  differential Lie algebra $\g$ this cohomology coincides with Lie algebra 
cohomology $\rH(\g,\mathbb{C})$ (cohomology with trivial coefficients). Considering cohomology of $Q$ acting on the 
space of vector fields (space of derivations of the algebra of functions) we get a notion generalizing  the notion of 
cohomology $\rH(\g,\g)$  ( cohomology with coefficients in adjoint representation).
{\footnote { Usually  the definition of Lie algebra
cohomology is based on the consideration of polynomial functions of ghosts; using formal series we obtain a completion 
of cohomology.}}

Notice, that to every L$\ity$ algebra  $\g=(W, \mu _n)$ we can assign a supercommutative differential
algebra $(\widehat {\Sym}(W^*),Q)$ that is in some sense dual
to the original L$\ity$-algebra.  If only a finite number of operations $\mu _n$ does not vanish
the operator $Q$  transforms a polynomial function into a polynomial function, hence we can consider also a free 
supercommutative differential algebra $(\Sym(W),Q)$ where $\Sym(W)$ stands for the algebra of polynomials on $W$.
We shall use the notations $(\Sym(W^*),Q)=C^{\bullet} (\g), (\widehat {\Sym}(W^*),Q)={\hat C}^{\bullet}(\g)$ and the 
notations  $\rH (\g,\mathbb{C}),\hat {\rH} (\g,\mathbb{C})$ for corresponding cohomology. {\footnote {In the case of Lie 
algebra the functor $C^{\bullet}$ coincides with Cartan-Eilenberg construction
of differential algebra giving Lie algebra cohomology.}}
Similarly for the cohomology in the space of derivations we use the notations $\rH(\g,\g), \hat {\rH}(\g, \g)$.

In the case when L$\ity$ algebra is $\mathbb{Z}$-graded the  cohomology   $\rH(\g,\mathbb{C})$  and $\rH(\g,\g)$ are  also 
$\mathbb{Z}$-graded.

 One can consider intrinsic cohomology of an L$\ity$ algebra. They are defined   as $\Ker\mu_1/\Im\mu_1$. One says 
that an  L$\ity$ homomorphism, which is the same as  $Q$-map in the language of $Q$-manifolds{\footnote {Recall, that 
a map of $Q$-manifolds is a $Q$-map if it is compatible with $Q$.}}, is a quasi-isomorphism if it induces an isomorphism of 
intrinsic cohomology. Notice, that in the case of $\mathbb{Z}$-graded L$\ity$ algebras L$\ity$ homomorphism should 
respect $\mathbb{Z}$ grading.

Every $\mathbb{Z}$-graded L$\ity$ algebra is quasi-isomorphic to L$\ity$ with $\mu_1=0$. (In other words every L$\ity$ 
algebra has a minimal model). Moreover,  every $\mathbb{Z}$-graded L$\ity$ algebra algebra is quasi-isomorphic to 
direct product of minimal  L$\ity$ algebra and a trivial one. (We say that  L$\ity$ algebra  is trivial if it can be regarded as 
differential abelian Lie algebra with zero cohomology.)

The role of zero locus of $Q$ is played by the space of solutions of Maurer-Cartan (MC) equation:
$$\sum _n \frac{1}{n!}\mu _n (a,...,a)=0.$$
 To obtain a space of solutions $Sol/\thicksim$ we should factorize space of solutions $Sol$ of MC in appropriate way  or 
work with a minimal model of $A$.

Our main interest lies in gauge theories. We consider these theories for all groups U$(n)$ at the same time. To analyze 
these theories it is more convenient to work with A$\ity$ instead of  L$\ity$ algebras. 

An A$\ity$ algebra can be defined as a formal non-commutative $Q$-manifold. In other words we consider an algebra of  
power series  of several variables which do not satisfy any relations (some of them are even, some are odd). An A$\ity$ 
algebra is defined as an odd derivation $Q$ of this algebra which satisfies $[Q,Q]=0$. 

 More precisely we consider a $\mathbb{Z}_2$-graded vector space $W$ with coordinates $\bz^{A}$. The algebra of 
formal noncommutative power series $\mathbb{C}\langle \langle \bz^A \rangle \rangle$ is a completion  $\hat{T}(W^*)$ of 
the tensor algebra $T(W^*)$  (of the algebra of formal noncommutative polynomials). The derivation is specified by the 
action on $\bz^A$:
\begin{equation}\label{E:quusuf}
Q(\bz^A)=\sum_{n}\sum \pm \mu^A_{B_1,\dots,B_n}\bz^{B_1}\dots \bz^{B_n}
\end{equation} 
We can use $\mu^A_{B_1,\dots,B_n}$ to specify a series of operations $\mu_n$ on the space $\Pi W$ as in L$\ity$ case. 
(In the case when $W$ is $\mathbb{Z}$-graded instead parity reversal $\Pi$ we should consider the shift of the grading 
by $1$.) If
$Q$ defines an A$\ity$ algebra then the condition $[Q,Q]=0$
leads to quadratic relations between operations; these relations can be used to give an alternative definition of A$\ity$ 
algebra.
In this definition an A$\ity$ algebra is a $\mathbb{Z}_2$-graded  or $\mathbb{Z}$- graded linear space , equipped with a 
series of maps $\mu_n:A^{\otimes n}\to A, n\ge 1$ of degree $2-n$ that  satisfy quadratic relations:
\begin{equation}\label{E:gfdjruv}
\begin{split}
&\sum_{i+j=n+1}\sum_{0\le l\le i}\epsilon(l,j)\times\\
&\mu_i(a_0,\dots,a_{l-1},\mu_j(a_l,\dots,a_{l+j-1}),a_{l+j},\dots,a_n)=0
\end{split}
\end{equation}
where $a_m\in A$, and 

$\epsilon(l,j)=(-1)^{j\sum_{0\le s\le l-1}deg(a_s)+l(j-1)+j(i-1)}$.

In particular, $\mu_1^2=0$.

Notice that in the case when only finite number of operations $\mu _n$ do not vanish  (the RHS of (\ref {E:quusuf}) is a 
polynomial) we can work with polynomial functions instead of power series. We obtain in this case a differential on the 
tensor algebra  $(T(\Pi W^*),Q)$.
The transition from A$\ity$ algebra $A=(W,\mu_n)$ to a differential algebra $\cobar A=(T(\Pi W^*),Q)$ is known as  co-
bar construction.  If we consider instead of tensor algebra its completion 
(the algebra of formal power series) we obtain the differential algebra $({\hat T}(\Pi W^*),Q)$  as a completed co-bar 
construction ${\widehat {\cobar}} A$.  

The cohomology of differential algebra $(T(\Pi W^*),Q)$=$\cobar A$ are called Hochschild cohomology of $A$ with 
coefficients in trivial module $\mathbb{C}$; they are denoted  by $\rHH (A,\mathbb{C})$.  Using the completed co-bar 
construction  we can give another definition of
Hochschild cohomology of A$\ity$ algebra  
as the cohomology of the differential algebra
$({\hat T}(\Pi W^*),Q)$=$\widehat {\cobar} A$; this cohomology can be defined also in the case when we have infinite 
number of operations. It will be denoted by  $\widehat {\rHH} (A,\mathbb{C})$.  Under some mild conditions (for example,
if the differential is equal to zero) 
one can prove that $\widehat {\rHH} (A,\mathbb{C})$ is a completion of ${\rHH} (A,\mathbb{C})$; in the case when  ${\rHH} 
(A,\mathbb{C})$ is finite-dimensional this means that the definitions coincide. We shall always assume  that $\widehat 
{\rHH} (A,\mathbb{C})$ is a completion of ${\rHH} (A,\mathbb{C})$.

The theory of A$\ity$ algebras is very similar to the theory of L$\ity$ algebras. In particular $\mu_1$ is a differential: $
\mu_1^2=0$. It can be used to define intrinsic cohomology of A$\ity$ algebra. If $\mu_n=0$ for $n\geq3$ then operations 
$\mu_1,\mu_2$ define a structure of differential associative algebra on $W$.

The role of equations of motion is played by so called MC equation
\begin{equation}
\sum_{n\geq 1}\mu_n(a,\dots,a)=0
\end{equation}
Again to get a space of solutions $Sol/\sim$  we should factorize solutions of MC equation in appropriate way or to work 
in a framework of minimal models, i.e. we  should use the  A$\ity$  algebra that is quasi-isomorphic to the original algebra 
and has $\mu_1=0$. (Every $\mathbb{Z}$-graded A$\ity$ algebra has a minimal model.)

We say that $1$ is the unit element of A$\ity$ algebra if $\mu _2(1,a)=\mu _2(a,1)=a$ (i.e. $1$ is the unit for binary 
operation) and all other operations with $1$ as one of arguments give zero. For every A$\ity$ algebra $A$ we  construct
a new A$\ity$ algebra $\tilde A$ adjoining a unit element. {\footnote { Notice, that in our definition
of Hochschild cohomology we should work with
non-unital algebras; otherwise the result for the cohomology with coefficients in $\mathbb{C}$ would be trivial. In more standard approach one defines 
Hochschild cohomology of unital algebra using the augmentation ideal.}}

Having an A$\ity$ algebra $A$ we can construct a series $L_N(A)$ of L$\ity$ algebras. If $N=1$ it  is easy to describe 
the corresponding L$\ity$ algebra  in geometric language. There is a map from noncommutative formal functions on $\Pi 
A$ to ordinary (super)commutative formal functions on the same space. Algebraically it corresponds to imposing (super)
commutativity relations among generators. 
Derivation $Q$ is compatible with such modification. It results in $L_1(A)$.  By definition $L_N(A)=L_1(A\otimes \Mat_N)
$.

If $A$ is an ordinary associative algebra, then $L_1(A)$ is in fact a Lie algebra- it has the same space and the   
operation  is equal to the commutator  $[a,b]=ab-ba$.

The use of A$\ity$ algebras in the YM theory is based on the remark that one can construct an A$\ity$ algebra $\mathcal
{A}$ with inner product such that for every $N$ the algebra $L_N(\tilde {\mathcal{A}})$ specifies YM theory with matrices 
of size $N\times N$ in $BV$ formalism. (Recall, that we construct  $\tilde {\mathcal{A}}$ adjoining unit element to $
\mathcal{A}$.)   The construction of the A$\ity$ algebra $\tilde {\mathcal{A}}$ is very simple: in  the formula for $Q$ 
in BV-formalism of YM theory in component formalism we replace matrices with free variables. The operator $Q$ 
obtained in this way specifies also differential algebras $\cobar \tilde {\mathcal{A}}$ and
${\widehat {\cobar}}\tilde {\mathcal{A}}$. To construct the A$\ity$ algebra $\mathcal{A}$ in the case of reduced YM theory 
we notice that the elements of the basis of $ \tilde {\mathcal{A}}$ correspond to the fields of the theory; the element 
corresponding to the ghost field $c$ is the unit; remaining elements of the basis generate the algebra $\mathcal{A}$.  In 
the case
of reduced theory the differential algebra
$\cobar \mathcal{A}$ can be obtained  from
$\cobar \tilde {\mathcal{A}}$ by means of factorization with respect to the ghost field $c$;
we denote this algebra by $BV_0$ and the original algebra $\mathcal{A}$ will be denoted by $bv_0$.  The construction 
in unreduced case is similar. In this case the ghost field (as all other fields) is a function on ten-dimensional space; to 
obtain $\cobar \mathcal{A}$ (that will be denoted later by $BV$)  we factorize $\cobar \tilde {\mathcal{A}}$ with respect to 
the ideal generated by the constant ghost field $c$. We shall use the notation $bv$ for the algebra $\mathcal{A}$ in 
unreduced case.

Instead  of working with component fields we can use pure spinors. Then instead of the  algebra $bv_0$ we should work 
with reduced Berkovits algebra $B_0$ that is quasi-isomorphic to $bv_0$; the algebra $BV_0$ is quasi-isomorphic to $U
(YM)$.  In unreduced case we work 
with Berkovits algebra $B$ that is quasi-isomorphic to $bv$ and with the algebra $U(TYM)$ quasi-isomorphic to $BV$ 
(see Section \ref{S:BV}, \cite {MSch} and \cite {MSch2} for more detail).

Notice, that the quasi-isomorphisms we  have described are useful for calculation of homology.
For example, as we have seen in Section \ref{S:BV} the space of fields in pure spinor formalism can be equipped with odd 
symplectic form (\ref {tr}) that
vanishes if the sum of ghost numbers of arguments is $>3$; the space of fields should be factorized with respect to the 
kernel of this form.
It  follows that homology and cohomology of
$U(YM)$ with coefficients in  any module vanish 
in dimensions $>3$. From the other side the form
(\ref {tr}) can be used to establish Poincar\'e duality
in  the cohomology of $U(YM)$.

It is easy to reduce classification of deformations of A$\ity$ algebra $A$ to a homological problem (see \cite{Penkava}). 
Namely it is clear that an infinitesimal deformation of $Q$ obeying $[Q,Q]=0$ is an odd derivation $q$ obeying $[Q,q]
=0$. The operator $Q$ specifies a differential on the space of all derivations by the formula 
\begin{equation}\label{E:defds}
\TQ q=[Q,q]
\end{equation} We see that infinitesimal deformations correspond to cocycles of this differential. It is easy to see that two 
infinitesimal deformations belonging to the same cohomology  class are equivalent (if $q=[Q,v]$ where $v$ is a 
derivation then we can eliminate $q$ by a change of variables $exp(v)$).
We see that the classes of infinitesimal deformations can be identified with homology  $\rH(Vect(\mathbb{V}),d)$ of the 
space of vector fields. 
(Vector fields on $\mathbb{V}$ are  even and odd derivations of $\mathbb{Z}_2$-graded algebra of formal power series.) 
If the number of operations is finite we can restrict ourselves to polynomial vector fields (in other words, we can replace 
$Vect(\mathbb{V})$ with $\cobar A\otimes A$).

The above construction is another particular case of Hochschild cohomology ( the cohomology with coefficients in 
coefficients in $\mathbb{C}$ was defined in terms of cobar construction. ) We denote it by $\widehat {\rHH}(A,A)$ (if we 
are working with formal power series) or by $\rHH(A,A)$  (if we are working with polynomials). Notice that these 
cohomologies  have a structure of  (super) Lie algebra induced  by commutator of vector fields. 

We shall give a definition of Hochschild cohomology  of differential graded associative algebra $(A,d_A)$ \[ A=\bigoplus_{i
\geq 0} A_i\] with  coefficients in a differential bimodule $(M,d_M)$\[M=\bigoplus_{i} M_i\]  in terms of
Hochschild cochains (multilinear functionals on $A$ with values in $M$).

We use the standard notation for the degree $\bar{a}=i$  of a homogeneous element $a\in A_i$. 

We first associate with the pair $(A,M)$ a bicomplex $(C^{n,m},D_I,D_{II}), n\geq 0$, $D_I:C^{n,m}\rightarrow C^{n+1,m}
$,$D_{II}:C^{n,m}\rightarrow C^{n,m+1}$ as follows: 
\begin{equation}\label{E:prod1}   
C^{n.m}(A,M)=\prod_{i_1,\dots,i_n } \Hom(A_{i_1}\otimes \cdots\otimes A_{i_n},M_{m+i_1+\cdots i_n} )
\end{equation}
 and for $c\in C^{n,m} $

\begin{equation}
\begin{split}
&D_{I}c=a_0c(a_1,\dots,a_n)+\sum_{i=0}^{n-1}(-1)^{i+1}c(a_0,\dots, a_ia_{i+1},\dots,a_n)+(-1)^{m\bar{a}_n+n}c
(a_0,\dots,a_{n-1})a_n\\
&D_{II}c=\sum_{i=1}^n(-1)^{1+\bar{a}_1+\cdots a_{i-1}}c(a_1,\dots, d_A(a_i),\dots,a_n)+(-1)^kd_Mc(a_1,\dots,a_n)\\
\end{split}
\end{equation}
Clearly \[D_{I}^2=0, D_{II}^2=0, D_ID_{II}+D_{II}D_I=0\]
We define the space of  Hochschild $i$-th cochains as 
\begin{equation}\label{E:prod2}
\widehat{C}^i(A,M)=\prod_{n+m=i} C^{n,m}(A,M)`.
\end{equation}
 Then $\widehat{C}^{\bullet}(A,M)$ is the complex $(\prod C^{i}(A,M),D)$ with $D=D_I+D_{II}$.
The operator  $D$ can also  be considered as a differential on the direct sum $C(A,M)=\bigoplus_i C^i(A,M)$ with direct 
products in (\ref{E:prod1},\ref{E:prod2}) replaced by the direct sums (on the space of non-commutative polynomials on $
\Pi A$ with values in $M$).
Similarly  $\widehat{C}(A, M)$ gets interpreted as  the space of formal power series on $A$ with values in $M$. We 
define the Hochschild  cohomology $\rHH(A, M)$ and  $\widehat {\rHH}(A, M)$ as the cohomology of this differential. Again 
under certain mild conditions that will be assumed in  our consideration the second group is a completion of the first one; 
the  groups coincide if $\rHH(A, M)$ is finite-dimensional.

Notice that $C(A,M)$ can be identified with the tensor product $\cobar A\otimes M$ with a differential defined by the 
formula
\begin{equation}
\label{hh0}
D(c\otimes m)=(d_{\cobar}+d_M)c\otimes m+[e,c\otimes m]
\end{equation}
where $e$ is the tensor of the identity map $\id\in \End(A)\cong \Pi A^*\otimes A\subset \cobar(A)\otimes A$

A similar statement is true for ${\hat C}(A, M).$

Notice that we can define the total grading of Hochschild cohomology  $\rHH^i(A,M)$ where  $i$ stands for the total 
grading defined in terms of $A$, $M$ and the ghost  number (the number of arguments).

In the case when $M$ is the algebra $A$ considered as a bimodule the elements of $\rHH^2(A,A)$ label infinitesimal 
deformations of associative algebra $A$ and the elements of $\rHH^{\bullet}(A,A)$ label infinitesimal deformations of $A$ into A$\ity$  algebra.  Derivations of $A$ specify elements of
$\rHH^1(A,A)$ (more precisely, a derivation can be considered as one-dimensional Hochschild cocycle; inner derivations are homologous to zero). 

We can define Hochschild homology $\rHH_{\bullet}$ considering Hochschild chains (elements of $A\otimes...\otimes A 
\otimes M$
). If $A$ and $M$  are finite-dimensional (or graded with finite-dimensional components) we
can define homology by means of dualization of cohomology 
$$\rHH_i(A,M)=\rHH^i(A, M^*)^*.$$
Let us assume that  the differential bimodule $M$ is equipped with bilinear inner product  of degree $n$ {\footnote { This 
means that  the inner product does not vanish only if the sum of degrees of arguments is equal to $n$. For example, the 
odd bilinear form in pure spinor formalism of SYM can be considered as inner product of degree 3.}} that descends to 
non-degenerate inner product on homology. This product generates a 
quasi-isomorphism $M\to M^*$ and therefore an
isomorphism between $\rHH_i(A,M)$ and $\rHH^{n-i}(A,M)$ (Poincar\'e isomorphism).
Let us suppose now that A$\ity$ has Lie algebra of symmetries $\g$ and we are interested in deformations of this algebra 
preserving  the symmetries.

This problem appears if we consider YM theory for all groups $\U(n)$ at the same time and we would like to deform the 
equations of motion preserving the symmetries of the original theory (however we do not require that  the deformed 
equations  come from an action functional).

When we are talking about symmetries of A$\ity$ algebra  $A$ we have in mind derivations of the algebra $\widehat 
{\cobar} A=  (\hat T(W^*),Q)$ (vector fields on a formal non-commutative manifold) that commute  with $Q$; see equation 
(\ref{E:quusuf}). We say that symmetries $q_1,\dots,q_k$ form Lie algebra $\g$ if they satisfy commutation relations of $
\g$ up to $Q$-exact terms. These symmetries determine  a homomorphism of Lie algebra $\g$ into Lie algebra $\widehat 
{\rHH}(A,A)$.  We shall say that
this homomorphism specifies weak action of $\g$ on $A$.

In the case when A$\ity$ algebra is $\mathbb{Z}$-graded we  can impose the condition that the symmetry is compatible 
with the grading.

Another way to define symmetries of A$\ity$ algebra is to identify them with L$\ity$ actions of Lie algebra $\g$ on this 
algebra, i.e. with L$\ity$ homomorphisms of $\g$ into differential Lie algebra of derivations $Vect$ of the algebra $
\widehat {\cobar} A$ (the differential acts on $Vect$ as (super)commutator with $Q$).  More explicitly L$\ity$ action is 
defined as a linear map
\begin{equation}
\label{li}
q:\Sym \Pi \g \to \Pi Vect
\end{equation}
or as an element of odd degree 
\begin{equation}
\label{lit}
q\in C^{\bullet}(\g)\otimes Vect
\end{equation}
obeying
\begin{equation}
\label{lity}
d_{\g}q+[Q,q]+\frac{1}{2}[q,q]=0.
\end{equation}
where $d_{\g}$ is a differential entering the definition of Lie algebra cohomology.
We can write $q$ in the form \[q=\sum \frac{1}{r!}q_{\alpha _1,\dots,\alpha _r} c^{\alpha _1}\cdots c^{\alpha _r}\] where 
$c^{\alpha}$ are ghosts of the Lie algebra; here $d_{\g}=\frac{1}{2}f^{\alpha}_{\beta\gamma}c^{\beta}c^{\gamma}\pr{c^
{\alpha}}$ where $f^{\gamma}_{\alpha \beta}$ denote structure constants of $\g.$

One can represent  (\ref{lity})  as an infinite
sequence of equations for the coefficients; the first of these equations has the form
$$[q_{\alpha},q_{\beta}]=f^{\gamma}_{\alpha \beta}q_{\gamma}+[Q,q_{\alpha\beta}].$$
  We see that
$q_{\alpha}$ satisfy commutation relations of $\g$ up to $Q$-exact terms (as we have said this means that  they specify  
a weak action of $\g$ on $A$ and a homomorphism $\g\to\widehat {\rHH}(A,A)$ ).

{\it In the remaining part of this section  we use the notation $\rHH$ instead of $\widehat {\rHH}$. }

Let us consider now an A$\ity$ algebra  $A$ equipped with L$\ity$ action of Lie algebra $\g$.
 To describe infinitesimal deformations of $A$ preserving the Lie algebra of symmetries we should find solutions of 
equations (\ref{lity})
 and $[Q,Q]=0$ where $Q$ is replaced by $ Q+\delta Q$ and $q$ by $q+\delta q$. After appropriate identifications these 
solutions can be described by  elements of cohomology group  that will be denoted by $\rHH_{\g}(A,A)$. To define this 
group we introduce ghosts $c^{\alpha}$. In other words we multiply $Vect(\V)$  by $\Lambda(\Pi \g^*)$  and define the 
differential by the formula
\begin{equation}\label{E:fdjerppq1}
d=\TQ+\frac{1}{2}f^{\alpha}_{\beta\gamma}c^{\beta}c^{\gamma}\pr{c^{\alpha}}+q_{\alpha}c^{\alpha}+\dots
\end{equation}
The dots denote the terms having higher order with respect to $c^{\alpha}$. They should be included to satisfy $d^2=0$ if 
$q_{\alpha}$ obey commutations of $\g$ up to $Q$-exact term. They
can be expressed in terms of $q_{\alpha _1,\dots, \alpha _r}$:
\begin{equation}\label{E:eqcohomologydig}
d=\TQ+\frac{1}{2}f^{\alpha}_{\beta\gamma}c^{\beta}c^{\gamma}\pr{c^{\alpha}}+\sum_{r\geq 1} \frac{1}{r!}c^{\alpha 
_1}\cdots c^{\alpha _r}q_{\alpha _1,\dots, \alpha _r}
\end{equation}
 In the terminology introduced in Section \ref{S:BV}  $\rHH_{\g} (A,A)$ is the Lie algebra cohomology of $\g$ with coefficients in the 
L$\ity$ differential $\g$ -module $(Vect(\V),\TQ)$:
 \begin{equation}
\label{lh}
\rHH_{\g} (A,A)=H(\g, (Vect(\V),\TQ)).
\end{equation}
 
From other side in the case of trivial $\g$ we obtain Hochschild cohomology. Therefore we shall use the term Lie-
Hochschild cohomology 
for the group (\ref {lh}).

Every deformation of A$\ity$ algebra $A$ induces a deformation of the algebra $\tilde A$ and of the corresponding L$\ity
$ algebra $L_N(\tilde A)$ ; if A$\ity$ algebra has Lie algebra of symmetries $\g$ then the same is true for this  L$\ity$ 
algebra. Deformations of A$\ity$ algebra preserving the symmetry algebra $\g$ induce symmetry preserving 
deformations of the  L$\ity$ algebra\footnote{At the level of cohomology groups it means that we have a series of 
homomorphisms $\rHH_{\g}(A,A)\rightarrow \rH_{\g}(L_N(\tilde A),L_n (\tilde A)) $}. This remark permits us to say that the 
calculations of symmetry preserving  deformations of A$\ity$ algebra $A$ corresponding to YM theory induces symmetry 
preserving deformations of EM for YM theories with gauge group $\UN$ for all $N$.

The calculation of cohomology groups $\rHH_{\susy}(YM,YM)$ permits us to describe SUSY-invariant deformations of EM. 
However we would like also to characterize  Lagrangian deformations of EM. This problem also can be formulated in 
terms of homology. Namely we should consider A$\ity$ algebras with invariant inner product and their deformations. We 
say that A$\ity$ algebra $A$ is equipped with odd invariant nondegenerate inner product $\langle.,.\rangle$ if $\langle 
a_0,\mu_n(a_1,\dots,a_n)\rangle=(-1)^{n+1}\langle a_n,\mu_n(a_0,\dots,a_{n-1})\rangle$. It is obvious that the 
corresponding L$\ity$ algebras $L_N(A)$ are equipped with odd invariant inner product. Therefore the corresponding 
vector field $Q$ comes from a solution of a Master equation $\{S,S\}=0$ (i.e. we have Lagrangian equations of motion). 
We shall  check that the deformations of A$\ity$ algebra preserving invariant inner product are labeled by cyclic 
cohomology of the algebra \cite{Penkava}.

As we have seen the deformations of A$\ity$ algebra are labeled by Hochschild cohomology cocycles of differential $\TQ
$ (see formula (\ref{E:defds})) acting on the space of derivations  $Vect(\V)$.

A derivation $\rho$ is uniquely defined by its values on generators of the basis of vector space $W^*$( on generators of 
algebra $\hat T(W^*)$). Let us introduce notations $\rho(z^i)=\rho^i(z^1,\dots,z^n)$. The condition that $\rho$ specifies a 
cocycle of $d$  means that it specifies a Hochschild cocycle with coefficients in $A$. The condition that $\rho$ preserves 
the invariant inner product is equivalent the  cyclicity condition on $\rho_{i_0,i_1\dots i_n}$, where $\rho_{i_0}
(z^1,\dots,z^{n})=\sum \rho_{i_0,i_1\dots i_k} z^{i_1}\dots z^{i_k}$. (We lower the  the upper index in $\rho$ using the 
invariant inner product.) The cyclicity condition has the form 
\begin{equation}\label{F:cycl}
\rho_{i_0,i_1\dots i_k}=(-1)^{k+1}\rho_{i_k,i_0\dots i_{k-1}}
\end{equation}
We say that $\rho_{i_0,i_1\dots, i_k}$ obeying formula (\ref{F:cycl}) is a cyclic cochain. To define cyclic cohomology we 
use Hochschild differential on the space of cyclic cochains.

{\footnote {One can say that the vector field $\rho$ preserving inner product is  a Hamiltonian vector field. The  cyclic 
cochain $ \rho_{i_0,i_1\dots i_k}$  can be considered as its Hamiltonian. The differential (\ref{E:defds}) acts on the space 
of Hamiltonian vector fields. The cohomology of corresponding differential acting on the space of Hamiltonians  is  called 
cyclic cohomology.}}

If we consider deformations of A$\ity$ algebra with inner product and Lie algebra $\g$ of symmetries and we are 
interested in deformations of $A$ to an algebra that also has invariant inner product and the same algebra of symmetries 
we should consider cyclic cohomology $\rHC_{\g}(A)$. The definition of this cohomology can be obtained if we modify the 
definition of $\rHC(A)$ in the same way as we modified the definition of $\rHH(A,A)$ to $\rHH_{\g}(A,A)$.

It is obvious that there exist a homomorphism from $\rHC(A)$ to $\rHH(A,A)$ and from $\rHC_{\g}(A)$ to $\rHH_{\g}(A,A)$ 
(every deformation preserving inner product is a deformation).{\footnote {Notice that we have assumed that $A$ is 
equipped with non-degenerate inner product. The definition of cyclic cohomology does not require the choice of inner 
product; in general there exists a homomorphism $\rHC(A)\to \rHH(A,A^*)$. The homomorphism $\rHC(A)\to \rHH(A,A)$ can 
be obtained as a composition of this homomorphism with a homomorphism $\rHH(A,A^*)\to \rHH(A,A)$ induced by a map 
$A^*\to A$.}}
 Our main goal is to calculate the image of $\rHC_{\g}(A)$ in $\rHH_{\g}(A,A)$ for the case of A$\ity$ algebra of YM theory, 
i.e. we would like to describe all supersymmetric deformations of YM that come from a Lagrangian.
 
 Cyclic cohomology are related to Hochschild cohomology by Connes exact sequence:
$$...\to \rHC^n(A)\to \rHH^n(A,A^*)\to \rHC^{n-1}(A) \to \rHC^{n+1}(A) \to...$$
Similar sequence exists for Lie-cyclic cohomology.

To define the cyclic homology $\rHC_{\bullet}(A)$ we work with cyclic chains (elements of $A\otimes...\otimes A$ 
factorized with  respect to the action of cyclic group). The natural map of Hochschild chains with coefficients in $A$ to 
cyclic chains commutes with the differential and therefore specifies a homomorphism $\rHH_k (A,A)\overset {l}{\to} \rHC_k(A)$. This homomorphism enters the homological version of Connes exact sequence
$$...\to \rHC_{n-1}(A)\overset{b}\to \rHH_n(A,A)\overset{I}{\to} \rHC_{n}(A) \overset{S}\to \rHC_{n-2}(A) \to...$$
We define the differential $B:\rHH_n(A,A)\to \rHH_{n+1}(A,A)$ as a composition $b\circ I$.
  
An interesting refinement of  Connes exact sequence exists in the case when $A$  is the universal enveloping of a Lie algebra $\g$ over $\mathbb{C}$. In this case cyclic homology get an additional index:  $\rHC_{k,j}(A)$. Such groups fit into the long exact sequence \cite{Loday}:

$$...\to \rHC_{n-1,i}(U(\g))\overset{b_{n-1,i}}\to \rHH_n(U(\g),\Sym^i(\g))\overset{I_{n,i}}\to \rHC_{n,i+1}(U(\g)) \overset{S_{n,i+1}}\to \rHC_{n-2,i}(U(\g)) \to...$$
 The differential \[B_i:\rHH_n(U(\g),\Sym^i(\g))\to \rHH_{n+1}(U(\g),\Sym^{i-1}(\g))\] is defined as a composition $b_{n,i+1}\circ I_{n,i}$. Finally if the Lie algebra $\g$ is graded then all homological constructs acquire an additional bold  index: $\rHH_{n\bm{l}}(U(\g),\Sym^i(\g))$, $\rHC_{n,i,\bm{l}}(U(\g))$. This index is preserved by the differential in the above sequence.
 
It is worthwhile to mention that all natural constructions that exist in cyclic homology can be extended to Lie-cyclic homology. 
 
It is important to emphasize that  homology and cohomology theories we considered in this section   are invariant  with 
respect to quasi-isomorphism (under certain conditions that are
fulfilled in our situation). {\footnote {The most general results and precise formulation of this statement can be found in 
\cite{Kellerhochschild} for Hoschschild cohomology and in \cite{Kellercyclic} for cyclic cohomology.}
 
According to \cite{Kellerhochschild} a quasi-isomorphism of two algebras $A\rightarrow B$ induces an isomorphism in 
Hochschild cohomology $\rHH^{\bullet}(A,A)\cong \rHH^{\bullet}(B,B). $ As we have mentioned Hochschild cohomology  
$\rHH^{\bullet}(A,A)$ is equipped with a structure of super Lie algebra, the isomorphism is compatible with this structure. 

This  theorem guarantees that quasi-isomorphism $A\rightarrow B$  allows  us to translate a  weak $\g$  action from $A$ 
to $B$.  

We have defined L$\ity$ action as an L$\ity$ homomorphism of Lie algebra $\g$ into differential Lie algebra of derivations 
$Vect (A)$.
It follows from the results of \cite{Kellerhochschild} that  a quasiisomorphism $\phi:A\rightarrow B$ induces a quasi-
isomorphism $\tilde{\phi}: Vect (A)\to Vect (B)$ compatible with L$\ity$  structure. 
{\footnote { In fact the structure of $Vect (A)$ is richer: it is a B$\ity$ algebra  (see \cite{Kellerhochschild} for details), but 
we shall use only L$\ity$ (Lie)  structure. One of the results of \cite{Kellerhochschild} asserts that $\tilde {\phi}$ is 
compatible with B$\ity$  structure. As a corollary it induces a quasi-isomorphism of L$\ity$ structures.}} We obtain that L$
\ity$ action on $A$ can be transferred to an L$\ity$ action on quasi-isomorphic algebra $B$.

The calculation of cohomology groups  we are interested in is a difficult problem. To solve this problem we apply the 
notion of duality of associative and A$\ity$ algebras.

\section{Duality}\label{AppendixB}

We define a pairing of two differential graded augmented \footnote{A differential graded algebra $A$ is called augmented 
if it is equipped with a $d$-invariant homomorphism $\epsilon:A\rightarrow \mathbb{C}$ of degree zero. We assume that 
the algebras at hand are $\mathbb{Z}$-graded and graded components are finite-dimensional.} algebras $A$ and $B$ as 
a degree one element $e \in A\otimes B$ that satisfies Maurer-Cartan equation 
\begin{equation}\label{E:MCEE}
(d_A+d_B)e+e^2=0
\end{equation}
Here we understand $A\otimes B$ as a completed tensor product.
\begin{example}
Let $x_1,\dots,x_n$ be the generating set of the quadratic algebra $A$. The set $\xi^1,\dots,\xi^n$ generates the dual  
quadratic algebra $A^!$ (see preliminaries). The element $e=x_i\otimes \xi^i$ has degree one, provided $x_i$ and $\xi^i$ 
have degrees two and minus one.  The element $e$ satisfies $e^2=0$ - a particular case of (\ref{E:MCEE}) for algebras 
with zero differential and therefore specifies a pairing between $A$ and $A^!$. 
\end{example}

Notice that the grading we are using here differs from the grading in the Section \ref{S:Preliminaries}.

{\bf Remark.} Many details of the theory depend on the completion of the tensor product, mentioned in the definition of $e
$. We, however, chose to completely ignore this issue because the known systematic way to deal with it requires 
introduction of a somewhat artificial language of co-algebras. {\footnote {One can define the notion of duality
between algebra  and co-algebra. This notion has better properties than the duality between algebras.}}

We call a non-negatively (non-positively) graded differential algebra $A=\bigoplus_i A_i$ connected, if $A_0\cong 
\mathbb{C}$. Such algebra is automatically augmented $\epsilon :A\rightarrow A_0$.
We call a  non-negatively graded connected   algebra $A$ simply-connected if  $A_1\cong 0$.

Let us consider a differential graded algebra $\cobar A=(T(\Pi A^*),d)$ where $A$ is an associative algebra and $d$ is 
the Hochschild differential. In other words we consider the co-bar construction for the algebra $A$. 
\begin{proposition}
The pairing $e$ defines the map \[\rho:\cobar(A)\rightarrow B\] of differential graded algebras. 
\end{proposition} 
\begin{proof}
The algebra $\cobar(A)$ is  generated by elements of $\Pi A^*$. The value of the map $\rho$ on $l\in \Pi A^*$ is equal to 
\[\rho(l)=\langle l,f_i\rangle g^i\] where $e=f_i\otimes g^i\in A\otimes B$. The compatiblity of $\rho$ with the differential 
follows automatically from (\ref{E:MCEE}). (Notice, that for graded spaces we always consider the dual as graded dual, 
i.e. as a direct sum of dual spaces to the graded components.)

\end{proof}

Similarly the element $e$ defines a map $\cobar(B)\rightarrow A$. 

\begin{definition}\label{D:dual}
The differential algebras $A$ and $B$ are dual if there exists a pairing $(A,B,e)$ such that the maps $\cobar(A)
\rightarrow B$ and  $\cobar(B)\rightarrow A$ are quasi-isomorphisms.
\end{definition}
{\footnote {Very similar notion of duality was  suggested independently by Kontsevich \cite {Ko}.}} 

Notice that duality is invariant with respect to quasi-isomorphism.

 If $A$ is quadratic then $A$ is dual to $A^!$ iff $A$ is a Koszul algebra.

If a differential graded algebra $A$ has a dual algebra, then $A$ is dual to $\cobar A$.
If  $A$ is a connected and simply-connected differential graded algebra, i.e. $A=\bigoplus_{i\geq 0} A_i$ and $A_0=
\mathbb{C}$ and $A_1=0$, then $A$ and $\cobar A$ are dual.

If differential graded algebras $A$ and $B$ are dual it is clear that Hochschild cohomology $\rHH(A,\mathbb{C})$ of $A$ 
with trivial coefficients coincide with intrinsic cohomology of $B$. This is because $B$ is quasi-isomorphic to $\cobar(A)$. 
One can say also that
\begin{equation}
\label{ab}
\rHH(A,A)=\rHH(B,B), 
\end{equation}
This is clear because these cohomology can be
calculated in terms of complex $A\otimes B$,
that is quasi-isomorphic  both to $A\otimes \cobar A$ and $\cobar B\otimes B$. Notice, that the isomorphism (\ref {ab}) does not preserve the grading.

This statement can be generalized to Hochschild cohomology of $A$ with coefficients in any bimodule $M$. Namely, we 
should introduce in $B\otimes M$  a differential by the  formula
\begin{equation}\label{E:dwwdiff}
d(b\otimes m)=(d_B+ d_M)b\otimes m +[e,b\otimes m]
\end{equation}
\begin{proposition}\label{E:rfjrvjdii}
Let $A$ be a connected and simply-connected differential graded algebra, i.e. $A=\bigoplus_{i\geq 0} A_i$ and $A_0=
\mathbb{C}$ and $A_1=0$. Then the Hochschild cohomology $\rHH(A,M)$ coincide with the cohomology of $B\otimes M$ 
with respect to  differential (\ref{E:dwwdiff}
).
\end{proposition}
To prove this statement we notice that the quasi-isomorphism $\cobar A\to B$ induces a homomorphism $C(A,M)=\cobar 
A\otimes M\to B\otimes M$; it follows from (\ref{hh0}) that this homomorphism commutes with the differentials
and therefore induces a  homomorphism on homology.  The induced homomorphism is an isomorphism; this can be 
derived from the fact that the map  $\cobar A\to B$ is a quasi-isomorphism. (The derivation is based on the
techniques of spectral sequences; the condition on algebra $A$ guarantees the convergence of spectral sequence.)

The  above proposition can be applied to the case when $A$ is a Koszul quadratic algebra and
$B=A^!$ is the dual quadratic algebra. We obtain the following statement generalizing the proposition  \ref{P:tqydxc1iww}  in Section \ref{S:wweqe}. (There is also a  statement for Hochschild homology generalizing  proposition \ref{P:iqwwstiwww2}. )

\begin{proposition}\label{P:tqydxc}  \cite{Qalg}
We assume that $A$ is Koszul.
The  Hochschild cohomology $\rHH^{\bullet}(A,N)$ is equal to the  cohomology of the complex  $N_c\overset{\ddef}=N
\otimes A^!$. The differential is the graded commutator with $e$.
 The complex $N_c$ splits according to degree: 
\begin{equation}\label{E:fdgisev}
\begin{split}
&N^{\bullet}_{c\  \bm{m}}=N_{m}\otimes A^!_0\rightarrow N_{m+1}\otimes A^!_1 \rightarrow\dots 
\end{split}
\end{equation}
The complex $N^{\bullet}_{c\  \bm{m}}$ is defined for positive and negative $\bm{m}$, we assume that $N_m=0$ if 
$m<m_0$.
Then $\rHH^{k,\bm{m}}(A,N)=H^k(N^{\bullet}_{c\ \bm{m}})$.
\end{proposition}
\begin{proposition}\label{P:iqwwst} \cite{Qalg}
We assume that $A$ is Koszul.
Homology $\rHH_{\bullet}(A,N)$ are equal to the cohomology of the complex $N_h\overset{\ddef}=N\otimes A^{!*}$. The 
space $A^{!*}=\bigoplus_{n\geq 0} A^{!*}_n$ is an $A^!$-bimodule dual to  $A^!$. The differential is a commutator with  
$e$ given by the formula (\ref{E:diff}). 
The complex $N_h$  splits :
\begin{equation}\label{E:homol22}
\begin{split}
&N^{\bullet}_{h\ \bm{m}}=N_{m_0}\otimes A^{!*}_{m-m_0}\overset{d}\rightarrow \dots N_{0}\otimes A^{!*}_{m}\overset{d}
\rightarrow  \dots N_{m-1}\otimes A^{!*}_{1} \overset{d}\rightarrow N_{m}\otimes A^{!*}_{0}
\end{split}
\end{equation}

Then $\rHH_{k,\bm{m}}(A,N)=H^{m-k}(N^{\bullet}_{h\ \bm{m}})$.

\end{proposition}
Proposition \ref{P:tqydxc1iww}  follows from Propositions \ref{P:tqydxc} ,\ref{P:iqwwst} if we set 
$A=U(\g),A^!=\C$  and use the fact that Lie algebra  cohomology  of $\g$
with coefficients in a $\g$-module coincide with Hochschild cohomology  of $U(\g)$ with coefficients in $U(\g)$-bimodule.
{\footnote {Let $N$ be a $U(\mathfrak{g})$-bimodule. Define a new structure of $\mathfrak{g}$-module on $N$ by the 
formula $l\otimes n\rightarrow ln-nl, l\in\mathfrak{g}, n\in N $
There is an isomorphism \cite{CE}\begin{equation}\label{E:dsfdsh}\rHH^i(U(\mathfrak{g}),N)\rightarrow \rH^i(\mathfrak{g},N),\end 
{equation} 
defined by the formula:
\begin{equation}\gamma(l_1,\dots, l_n) \rightarrow \tilde \gamma=\frac{1}{n!}\sum_{\sigma\in S_n}\pm \gamma(l_{\sigma
(1)},\dots,l_{\sigma(n)}),l_i\in \mathfrak{g}\end{equation}.
There is a similar isomorphism for homology. }}

\begin {proposition}
\label {al}
If   differential graded algebra $A$ is dual to $B$ and quasi-isomorphic to the envelope $U(\g)$ of Lie algebra $\g$ then 
$B$ is quasi-isomorphic  to the
super-commutative differential algebra  $C^{\bullet}(\g).$
\end{proposition} 
This statement  follows from the fact that the cohomology  of $C^{\bullet}(\g)$ (=Lie algebra cohomology of $\g$) 
coincides with Hochschild cohomology of $U(\g)$ with trivial coefficients.

It turns out that it is possible to calculate cyclic and Hochschild cohomology of $A$ in terms of suitable homological 
constructions for a dual algebra $B$.

Let $A$ and $B$ be dual differential graded algebras. Let us assume that $A$ and $B$ satisfy assumptions of 
Proposition \ref{E:rfjrvjdii}. 

\begin{proposition}
Under  above assumptions there is a canonical  isomorphism \[\rHC_{-1-n}(A)\cong \rHC^n(B),\]
where $\rHC^n$($\rHC_n$) stands for $i$-th cyclic cohomology(resp. homology) of an algebra.
\end{proposition}

\begin{proposition}
Under the above assumptions there is an isomorphism \[\rHH^n(A,A^*)=\rHH_{-n}(B,B),\]
where  $\rHH^n$($\rHH_n$) stands for $n$-th Hochschild cohomology (resp. homology).
\end{proposition}

For the case when $A$ and $B$ are quadratic algebras these two propositions were proven in
\cite{FT}. The proof in general  case is similar.
It can be based on results of \cite{Burg} or \cite{Loday}.

Let us illustrate some of  above theorems on concrete examples.
The  algebra $\Ss$ is dual to $U(L)$. This means,
that 
$$\rHH^{\bullet}(\Ss,\Ss)=\rHH^{\bullet}(U(L),U(L))=H^{\bullet}(L,U(L)).$$
The reduced Berkovits algebra $B_0$ is dual to
$U(YM)$, hence
$$\rHH^{\bullet}(B_0,B_0)=\rHH^{\bullet} (U(YM),U(YM))=\rH^{\bullet}(YM,U(YM))$$. 

We need the following information about this cohomology (\cite{M4}):  
\begin{equation}
\begin{split}\label{E:cohcomp}
& \rH^{0}(YM, U(YM))\cong \mathbb{C} \\
&\rH^{1}(YM,U(YM))\cong \mathbb{C} +S^*+\Lambda^2(V)+ V+S^*\\
\end{split}
\end{equation}
 Notice, that the answer for $\rH^1(YM, U(YM))=\rHH^1(U(YM),U(YM))$ has clear physical interpretation: symmetries of SYM theory (translations, Lorenz transformations and supersymmetries) specify
 derivations of the algebra $U(YM)$.  

Representing $U(YM)$ as $\Sym (YM)$ we obtain additional grading on cohomology:
\begin{equation}
\begin{split}\label{E:cohcompu}
\rH^{1}(YM,\Sym ^1(YM))\cong \mathbb{C} +S^*+\Lambda^2(V)\\
\rH^{1}(YM,\Sym ^0 (YM))\cong V+S^*\\
 \rH^{1}(YM,\Sym^i(YM))=0,i\geq 2 \\
\end{split}
\end{equation}

It follows from the remarks in Appendix A that
$\rH^k(YM, U(YM))=\rHH^k(U(YM),U(YM))=0$ for $k>3$.  

As we mentioned in Section \ref{S:BV} the odd symplectic structure in pure spinor formalism is specified by  degenerate  closed two-form $\omega$. This form determines an odd inner product of degree 3 on $B_0$
that  generates Poincar\'e isomorphism
$\rH^i(YM,U(YM))\overset{P}{\cong}\rH_{3-i}(YM, U(YM)).$

\section{On the relation of the Lie algebra and the BV approaches to the deformation problem}\label{AppendixC}

Our main goal is to calculate SUSY deformations of 10D YM theory and its reduction to a point.  In Section \ref
{S:homologicalapp} we have reduced this question to a homological problem. Another reduction of this kind comes from 
BV formalism (Section \ref{S:BV} and Appendix \ref{AppendixA}). Here we shall relate these two approaches. For 
simplicity we shall talk mostly about the
reduced case; we shall describe briefly the modifications that are necessary in the unreduced case.

We shall use the fact that under certain conditions all objects we are interested in are invariant with respect to quasi-
isomorphisms.

We can study the symmetries of Yang-Mills theory using the A$\ity$ algebra $\mathcal{A}$ constructed in Appendix \ref
{AppendixA} or any other algebra that is  quasi-isomorphic to $\mathcal{A}$. In  BV formalism a Lie algebra  action 
should be replaced by weak action or by  L$\ity$ action. It will be important for us to work with L$\ity$ action, because 
this action is used in the construction of formal deformations (Section \ref{E:Formal}).  We consider the case of YM 
theory dimensionally reduced to a point ; in this case  we use the notation $\mathcal{A}=bv_0$ and
the algebra $\cobar bv_0=BV_0$ is quasi-isomorphic to $U(YM)$ (to the envelope of
Lie algebra $YM$); see \cite{MSch2}, Theorem 1.
{\footnote {The algebra $BV_0$ has as generators the generators of
$U(YM)$, corresponding antifields  and  $c^*$ (the antifield for ghost)  ; sending  antifields and $c^*$ to zero we obtain a 
homomorphism of differential algebras. (Recall that  the differential on $U(YM)$ is trivial.)  It has been  proven in \cite 
{MSch2} that this homomorphism is a 
quasi-isomorphism.}}

The algebra $bv_0$ is dual to the algebra
$BV_0$.  This means that $bv_0$ is quasi-isomorphic to $C^{\bullet} (YM)$ (to the differential
commutative algebra
that computes Lie algebra cohomology with trivial coefficients; see Appendix \ref{AppendixB}, (\ref {al})).

 One can construct an L$\ity$ action of the reduced supersymmetry algebra $\g=\Pi\mathbb{C}^{16}$ on the algebras 
$bv_0,C^{\bullet} (YM), U(YM)$.  It is sufficient to construct such an action on one of these algebras.

Let us describe the  action on $C^{\bullet}(YM)$.

We shall use the Lie algebra $L$ defined in Section 2. By construction $L$  as a linear space is isomorphic to the direct sum $S
+YM$, where $S=L^1$ is spanned by $\btheta_{1},\dots,\btheta_{16}$. Thus \[C^{\bullet}(L)\cong \mathbb{C}[[t^1,\dots,t^
{16}]]\otimes C^{\bullet}(YM),\] where $t^{\alpha}$ are   even  variables dual to $\btheta_{\alpha}$. The differential $d_L$ 
in $C^{\bullet}(L)$ is the sum \[d_L=d_{YM}+q,\]
 where $q$ is equal to $t^{\alpha}t^{\beta}q_{\alpha\beta}+t^{\gamma}q_{\gamma}$. The operators  $q_{\alpha\beta},q_
{\gamma}$ are derivations of $C^{\bullet}(YM)$. 
 We can interpret  $q$  as map of  $\Sym (\Pi \g)$=$\Sym (\mathbb{C}^{16})$ into the space of derivations of  $C^
{\bullet}(YM)$. It is easy to check, that this map obeys (\ref{lity}); hence it specifies L$\ity$ action of  $\g=\Pi\mathbb{C}^
{16}$ on $C^{\bullet}(YM)$.
 
 Another way to describe  this L$\ity$ action of $\g=\Pi\mathbb{C}^{16}$ is to
 construct the corrections that arise because 
  the dimensionally reduced supersymmetries $q_{\gamma}$ defined by the formula $q_{\gamma}=[\btheta, x] $ anti-
commute only on-shell. The operators $q_{\alpha\beta}$ can be interpreted as L$\ity$ corrections to the  action of the Lie 
algebra $\g=\Pi\mathbb{C}^{16}$. In this construction no higher order operators $q_{\alpha_1,\dots,\alpha_n}$ ($n\geq 
3$) are present.

We should say  a word of caution. The action of $\Pi \mathbb{C}^{16}$ on $bv_0$ constructed this way could be 
incompatible with the inner product. A refined version of this action , free from the  shortcoming, is constructed in  
Appendix \ref{AppendixD}.

Similar arguments permit us to construct an L$\ity$ action  of SUSY Lie algebra in unreduced case. In this case the 
algebra $\mathcal{A}$ is denoted by $bv$, $\cobar bv=BV$ is quasi-isomorphic to $U(TYM)$ and $bv$ is quasi-
isomorphic to $C^{\bullet}(TYM)$. To construct  an L$\ity$ action of SUSY Lie algebra
on  $C^{\bullet}(TYM)$ we notice that as a vector space $L$ can be represented as a direct sum of vector subspaces
$L^1+L^2$ and $TYM$.  This means that
$$C^{\bullet}(L)\cong \mathbb{C}[[t^1,\dots,t^{16}]]\otimes \Lambda [\xi ^1,\dots,\xi ^{10}]\otimes C^{\bullet}(TYM),$$
where $t^{\alpha}, \xi ^i$ can be interpreted as even and odd ghosts of the Lie algebra $\susy$.
Again we can construct the L$\ity$ $\susy$ action on $C^{\bullet}(TYM)$ using the differential $d_L$ acting on $C^
{\bullet}(L).$
Namely, we choose a basis $\langle e_{\gamma},  \rangle \gamma\geq 1$ in $TYM$,  a basis $\langle \theta_{\alpha} ,
\rangle \alpha=1,\dots,16$ in $S$ and $\langle v_i \rangle, i=1,\dots,10$ in $V$. Together they form a basis in $\langle e_
{\gamma}, \theta_{\alpha}, v_i\rangle$  of $TYM+S+V\cong L$. The structure constants of the Lie algebra in this basis 
are $[e_{\gamma},e_{\gamma'}]=\sum_{\delta\geq 0} c^{\delta}_{\gamma\gamma'}e_{\delta}$, $[\theta_{\alpha},e_
{\gamma}]=\sum_{k}f_{\alpha \gamma}^{\delta} e_{\delta}$, $[v_i,e_{\gamma}]=\sum_{k}g_{ ij}^k e_{\gamma}$, $[\theta_
{\alpha},\theta_{\beta}]=\Gamma_{\alpha\beta}^iv_i$, $[\theta_{\alpha},v_i]=\sum_{\delta}h_{\alpha i}^{\delta}e_{\delta}$.  
The algebra $\Sym{\Pi L^*}\cong \mathbb{C}[[t^1,\dots,t^{16}]]\otimes \Lambda [\xi ^1,\dots,\xi ^{10}]\otimes C^{\bullet}
(TYM)$ has generators $\epsilon^{\gamma}$, $t^{\alpha}$, $\xi_{i}$ dual and opposite parity to $e_{\gamma},\theta_
{\alpha}, v_i$. The differential can be written as
\[\begin{split} &c_{\gamma\gamma'}^{\delta}\epsilon^{\gamma}\epsilon^{\gamma'}\frac{\sd}{\sd \epsilon^{\delta}}+\\
&+f_{\alpha \gamma}^{\gamma'}  t^{\alpha}\epsilon^{\gamma}\frac{\sd}{\sd \epsilon^{\gamma'} }+\\
&+g_{i\gamma}^{\gamma'}\xi^{i}\epsilon^{\gamma}\frac{\sd}{\sd \epsilon^{\gamma'}}+\\
&+r_{i\alpha}^{\gamma}\xi^{i}t^{\alpha}\frac{\sd}{\sd\epsilon^{\gamma} }+\Gamma_{\alpha\alpha'}^it^{\alpha}t^{\alpha'}
\frac{\sd}{\sd \xi_{i}}\end{split}\]

We identify $c_{\gamma\gamma'}^{\delta}\epsilon^{\gamma}\epsilon^{\gamma'}\frac{\sd}{\sd \epsilon^{\delta}}$ with the 
differential in $C^{\bullet}(TYM)$. All other terms define the desired L$\ity$action.

To classify infinitesimal SUSY deformations
of reduced YM theory it is sufficient to calculate
Lie- Hochschild cohomology
$\rHH_{\g}(C^{\bullet}(YM),C^{\bullet}(YM))$. 
First of all we notice that one can use duality between 
$C^{\bullet}(YM)$ and $U(YM)$  to calculate
 Hochschild cohomology
\begin{equation}
\label{hh}
\rHH(C^{\bullet}(YM),C^{\bullet}(YM))=\rHH(U(YM),U(YM))=\rH(YM,U(YM))=\rH(YM, \Sym (YM)).
 \end{equation}
 Here we are  using (\ref {ab}) and the relation between Hochschild cohomology of enveloping algebra $U(YM)$
 and Lie algebra cohomology of $YM$ as well as Poincar\'e -Birkhoff-Witt theorem. Analyzing the
 proof of (\ref {hh}) we obtain quasi-isomorphism 
 between the complex $\cobar C^{\bullet}(YM)\otimes C^{\bullet}(YM)$ that we are using in calculation of $\rHH(C^{\bullet}
(YM),C^{\bullet}(YM))$ and the complex $\Sym (\Pi YM)^*\otimes \Sym YM$ with cohomology $\rH(YM, \Sym (YM))$.
 
 To calculate the Lie-Hochschild cohomology
 $\rHH_{\g}(C^{\bullet}(YM),C^{\bullet}(YM))$
 we should consider  a complex 
 $C^{\bullet}(\g)\otimes \cobar C^{\bullet}(YM)\otimes C^{\bullet}(YM)$
 with the differential defined in (\ref  {E:eqcohomologydig} ). This complex is quasi-isomorphic to
 $C^{\bullet}(\g)\otimes \Sym (\Pi YM)^*\otimes \Sym YM$ with appropriate differential.
 
 Now we can notice that  $C^{\bullet}(\g)\otimes \Sym (\Pi YM)^*=\mathbb{C}[[t^1,\dots,t^{16}]]\otimes \Sym(\Pi YM^*)=
\Sym(\Pi L^*)$.  Hence we can reassemble $C^{\bullet}(\g)\otimes \Sym (\Pi YM)^*\otimes \Sym YM$ into $\Sym(\Pi L^*)
\otimes \Sym( YM)$, which is isomorphic to $C^{\bullet}(L,U(YM))$. 

We see that  Lie- Hochschild cohomology
$\rHH_{\g}(C^{\bullet}(YM),C^{\bullet}(YM))$ with $\g= \Pi \mathbb{C}^{16}$ classifying  infinitesimal SUSY deformations in 
the reduced case is isomorphic to $\rH^{\bullet}(L,U(YM))$.

Lie-Hochschild cohomology
$\rHH_{\g}(C^{\bullet}(TYM),C^{\bullet}(TYM))$
where $\g=\susy$ govern SUSY deformations of
unreduced SYM. Similar considerations permit us to prove that these cohomology are isomorphic to  $\rH^{\bullet}(L,U
(TYM))$.

The above statements agree with the theorems of Section \ref{S:homologicalapp} where it is proven that two-dimensional 
cohomology of $L$ with coefficients
in $U(YM)$ and in $U(TYM)$ correspond to
SUSY deformations in reduced and unreduced
cases. We see  that BV approach
leads to  wider class of SUSY deformations.
{\footnote { One can modify the arguments of Section \ref{S:homologicalapp} to cover the additional deformations arising in BV formalism. The 
modification is based on consideration of A$\ity$ deformations of associative algebras $U(YM)$ and $U(TYM).$}}
However, one can prove that all super Poincar\'e invariant deformations in the reduced case are covered by the 
constructions of Sections \ref{S:homologicalapp}. The proof is based on the remark that the groups $\rH^i(L,U(YM))$ do 
not contain $\Spin (10)$-invariant elements for $i>2.$ (This remark can be  derived from the considerations of Section \ref{S:calc}.) 
Notice, however, that  the group 
$\rH^i(L,U(TYM))$ contains  $\Spin (10)$-invariant element for $i=3$. This element is responsible for the super Poincar\'e 
invariant deformation of L$\ity$ action of supersymmetry, but it generates a trivial infinitesimal variation of action 
functional.
However, corresponding formal deformation
constructed in Section \ref{E:Formal} can be non-trivial.

\section{The  L$\ity$ action of the supersymmetry algebra in the BV formulation}\label{AppendixD}

In Appendix \ref{AppendixC} we have shown that one  can construct an L$\ity$ action of SUSY algebra
on $bv$. In this section we shall give another
proof of the existence  of this action; we shall show that  this proof permits us to construct an L$\ity$
action that is compatible with invariant inner product on $bv$. 
We  use the formalism of pure spinors  in our considerations.

The pure spinor construction will be preceded by a somewhat general discussion of L$\ity$-invariant traces. 

Suppose that the tensor product $A\otimes \Sym(\Pi \g)$ is furnished with a differential $d$ which can be written as $d_A
+d_{\g}+q$, where $d_A$ is a differential in $A$, $d_{\g}$ is the Lie algebra differential (\ref{E:dg} ) in $\Sym(\Pi \g)$ and 
$q=\sum_{n\geq 1}\frac{1}{n!} c^{\alpha_1}\cdots c^{\alpha_n} q_{\alpha_1,\dots,\alpha_n}$ is  the generating function of 
derivations $q_{\alpha_1,\dots,\alpha_n}$ that satisfies the analog of (\ref{lity} ). We say that $A$ is equipped with $\g$-
equivariant trace if there is a linear map \[p:A\otimes \Sym(\Pi \g)\rightarrow \Sym(\Pi \g)\] which satisfies $p([a,a'])=0$ and 
\[p(Q_A+d_{\g}+q)a=d_{\g}p a\] for every $a\in A\otimes \Sym(\Pi \g)$. 
In the case when we have an ordinary  action of a Lie algebra $\g$ on a differential graded Lie algebra $A$ and a trace 
functional $p$ is $\g$-invariant, i.e. $p(la)=0$ for any $l\in \g$ and $a\in A$ then $p$ is trivially a $\g$-equivariant 
functional.

This construction provides us with an inner product $\langle a, b \rangle = p_{\g}(ab)$ on $A$ with values in $\Sym(\Pi\g)$ (ghost valued inner product). If the  trace $p$ is odd the corresponding inner product is also odd.

In pure spinor formalism the algebra $S\otimes C^{\infty}(\mathbb{R}^{10|16})$ is equipped with the differential $D$ 
given by the formula (\ref{E:differential}) and the $D$-closed odd linear functional
\begin{equation}
p:S\otimes C^{\infty}(\mathbb{R}^{10|16})\rightarrow \mathbb{C}
\end{equation} 
It is defined on elements that decay sufficiently fast at the space-time infinity. The functional $p$ splits into a tensor 
product of translation-invariant volume form $\vol$  on $\mathbb{R}^{10}$ and a functional $p_{\red}:S\otimes C^{\infty}
(\mathbb{R}^{0|16})\rightarrow \mathbb{C}$.

The super-symmetries generators are \[\theta_{\alpha}=\frac{\partial}{\partial\psi^{\alpha}}+\Gamma^i_{\alpha\beta}\psi^
{\beta}\frac{\partial}{\partial x^i}.\]

The functional $p_{\red}$ is $\Spin(10)$-invariant. Also it  can be characterised as the only nontrivial $\Spin(10)$-
invariant functional on 
$S\otimes C^{\infty}(\mathbb{R}^{0|16})$.  This follows from simple representation theory for $\Spin(10)$. 
This fact enables us to construct an "explicit" formula for $p_{\red}$. The projection \[\mathbb{C}[\lambda^1,\dots,
\lambda^{16}]\overset{k}\rightarrow S\] commutes with the action of $\Spin(10)$. A simple corollary of representation 
theory  is that $k$ has a unique linear $\Spin(10)$-equivariant splitting $k^{-1}$. Using this splitting we can identify  
elements of $S$ with $\Gamma$-traceless elements in $\mathbb{C}[\lambda^1,\dots,\lambda^{16}]$. Let us define a $
\Spin(10)$-invariant  differential operator on $\mathbb{C}[\lambda^1,\dots,\lambda^{16}]\otimes \Lambda[\psi^1,\dots,
\psi^{16}]$ by the formula

\begin{equation}\label{E:pdiff}
P=\Gamma^{\alpha\beta}_m
\frac{\partial}{\partial \lambda^{\alpha}}\frac{\partial}{\partial \psi^{\beta}}
\Gamma^{\gamma\delta}_n
\frac{\partial}{\partial \lambda^{\gamma}}\frac{\partial}{\partial \psi^{\delta}}
\Gamma^{\epsilon\varepsilon}_k
\frac{\partial}{\partial \lambda^{\epsilon}}\frac{\partial}{\partial \psi^{\varepsilon}}
\Gamma^{\mu\nu}_{mnk}
\frac{\partial}{\partial \psi^{\mu}}\frac{\partial}{\partial \psi^{\nu}}
\end{equation}

We define  $\pp_{\varnothing}(a)$ as $Pa|_{\lambda,\psi=0}$.

One of the properties of $\pp_{\varnothing}$ is that it is $D$-closed \[\pp_{\varnothing}(Da)=0.\]
It is not, however, invariant with respect to the action of supersymmetries. It satisfies a weaker condition \[\pp_
{\varnothing}(\theta_{\alpha}a)=\pp_{\alpha}(Da)\]

The generating function technique that was used for formulation of L$\ity$ action in the BV formulation can be used here.  
We have even ghosts $ t^{1},\dots,t^{16}$ and odd ghosts $\xi^1,\dots,\xi^{10}$.   We define the total L$\ity$ action 
operator $D_{\infty}$ as the sum
\begin{equation}
D_{\infty}=D+t^{\alpha}\theta_{\alpha}+\xi^i\frac{\partial}{\partial x^i}+d_{\susy}
\end{equation}
The condition $D_{\infty}^2=0$ is equivalent to the standard package of properties of the pure spinor BV differential and  
supersymmetries.

We shall construct a generating function of functionals $\pp$ that satisfies equation
\begin{equation}
\pp(a)=d_{\susy}\pp(a)+\mbox{ exact terms }
\end{equation}
We define  a formal perturbation  $\pp$ of $P$ as a composition
\begin{equation}
\pp=\sum_{k\geq 0}\frac{1}{k!}P\bigg(t^{\alpha}\frac{\partial}{\partial \lambda^{\alpha}}\bigg)^k
\end{equation}
Then \[\pp_{\varnothing}(a) \mbox{ is equal to } \pp(a)|_{t,\psi=0}\]
The L$\ity$-invariant trace  functional can be defined as \[p_{\susy}(a)=\int \pp(a)\vol\]

The (ghost dependent) odd inner product corresponding to this trace is also L$\ity$ invariant. It specifies a homologically non-degenerate  two-form on a formal $Q$-manifold; this two-form obeys the condition of proposition \ref {2},  hence we  can apply the conclusion of this proposition to classify Lagrangian deformations of SYM theory.

\section{Calculation of the hypercohomology}\label{AppendixE}

To justify the  calculations of Section \ref{S:calc} we should check that the embedding $\W^*\rightarrow \YM^{\bullet}$ 
and the embedding $\Sym^i (\W^*)\to\Sym^i (\YM)$ are quasi-isomorpisms. In other words, we should prove that these 
homomorphisms induce isomorphisms of hypercohomology. We should 
prove also similar results for embedding $\W\to \TYM$  and embedding $\Sym^i (\W)\to\Sym^i (\YM)$.

We shall start with some general considerations.
As we have noticed in Section \ref{S:calc} , there are two spectral sequences that can be used in calculation of 
hypercohomology of the complex of vector bundles $\N^{\bullet}_{\bm{l}}$. Here we shall use the second one (with 
$E_2=H^i(H^j(\Omega(\N_{\bm{l}} ),d_e),\dbar)$.

First of all we shall consider the  modules $N$ where $N=L, YM$ or $TYM$, corresponding  differential vector bundles $
\N=\cal L, \YM, \TYM$ and differential $P$-modules $N_P$  obtained as
fibers of these bundles over the point $\lambda _0\in \mathcal{Q}$. The differential on the module $N_P$ is obtained as restriction 
of the  differential $d_e$ 
on vector bundle $\N$ and will be denoted by the same symbol.

Let us start with calculation of the cohomology of the module $U(L)_P$.
\begin{proposition}\label{P:equiva}
$\rH^i(U(L)_P,d_{e})=\rHH^i(S,\mathbb{C}_{\lambda_0})$
\end{proposition}

Here one-dimensional $S$-bimodule $\mathbb{C}_{\lambda_0}$ is obtained by specialization at $\lambda_0\in \mathcal{C} $ 
with coordinates $\lambda_0^{\alpha}$. 
In more details the left and right  actions of polynomial $f(\lambda)$ on generators $a\in \mathbb{C}_{\lambda_0}$ is 
given by the formula $f(\lambda)\times a=f(\lambda_0)a$

\begin{proof}
This is a direct application of Proposition \ref{P:tqydxc} where $N=\mathbb{C}_{\lambda_0}$ and $A=\Ss$,$A^!=U(L)$.
\end{proof}

To calculate  RHS in Proposition \ref{P:equiva}  we use  the following statement that can be considered as  a weak form 
of Hochschild-Kostant-Rosenberg theorem (see \cite {Connes}):

\begin{proposition}
\label {P:HKR}
Suppose $A$ is a ring of algebraic functions on affine algebraic variety. Let $\mathbb{C}_x$ denote a one-dimensional 
bimodule, corresponding to a smooth point $x$. Then $\rHH^i(A,\mathbb{C}_x)=\Lambda^i(T_x)$, where $T_x$ is the 
tangent space at $x$.
\end{proposition}
\begin{corollary}\label{C:hgfjs}
$\rHH^i(S,\mathbb{C}_{\lambda_0})=\rH^i(U(L)_P,d_e)=\rH^i(\Sym(L)_P,d_e)=\Lambda^i(T_{\lambda_0})$, where $T_
{\lambda_0}$ is the tangent space to $\mathcal{C} $ at the point ${\lambda_0}\neq 0$.
It follows from this that
$\rH^i(\Sym^i(L)_P,d_e)=\Lambda^i(T_{\lambda_0})$ and $\rH^j(\Sym^i(L)_P,d_e)=0,i\neq j$. In particular, for $i=1$ we 
obtain $\rH^1(L_P)=T_{\lambda _0}, \rH^j(L_P)=0$ if $j>1$.
\end{corollary}
The corollary follows from Proposition \ref{P:HKR} because   $\mathcal{C}$ is a smooth homogeneous space away from $
\lambda= 0$.

Recall that Lie algebra $L$ as a vector space is equal to $L^1+YM$. The action of the differential $d_e$ is $P$-covariant. 
This fact together with the information about  the cohomology of $L_P$ permits us to calculate the action of   $d_e$  on 
$L_P$ and on $YM_P$. 

Recall that
$$L_P=L^1\otimes \mu _{-1}+L^2+L^3\otimes\mu _1+...$$

We  describe the differential $d_e$  on $L^1\otimes \mu _{-1}$ using decomposition (\ref{E:decompa}-\ref{E:decompb}). 

It follows from  (\ref{E:decompa}) that $L^1\otimes \mu _{-1}$ has  $W^*$ as factor -representation, i.e there exists a 
surjective homomorphism $\phi :L^1\otimes \mu _{-1}\to W^*$.	We conclude from Schur's lemma that 
$d_e$ maps $L^1\otimes \mu _{-1}$ onto $W^*\subset L^2$ and coincides with $\phi$ up to a constant factor. From the 
information about the cohomology of $L_P$ we infer that the constant factor does not vanish. Taking into account that
the $\rH^i(L_P,d_e)=0$ for $i>1$ we obtain that the
 complex $L^1/\Ker d_e\rightarrow L^2\rightarrow \dots$ is acyclic. If we truncate $L^1/\Ker d_e$  term, the resulting 
complex will have cohomology equal to $d_e(L^1/\Ker d_e)=W^*$. This proves that the embedding $W^*\subset YM_P$ 
is a quasiisomorphism. 
 
To derive from this statement that the embedding
of vector bundles $\W^*\subset \YM$ generates
isomorphism of hypercohomology we notice
that this embedding induces a homomorphism of spectral sequences calculating the hypercohomology. It is easy to 
check  that the above statement implies isomorphism of $E_2$
terms, hence isomorphism of hypercohomology.

From K\"unneth theorem we can conclude that the embedding $\Sym ^j W\subset \Sym ^j YM _P$ is a quasi-
isomorphism; using spectral sequences we derive isomorphism of hypercohomology of corresponding complexes of 
vector bundles.

We can give a similar analysis of the complex $\TYM^{\bullet}$. Indeed we have a short exact sequence of complexes \[ 0 \rightarrow\TYM^{\bullet}\rightarrow \YM^{\bullet} \rightarrow {\cal L}^2\rightarrow 0,\] where ${\cal L}^2$ is a trivial 
vector bundle over $\mathcal{Q}$ with a fiber $L^2$.  The short exact sequence   gives rise to short exact sequence of 
corresponding $P$-modules and to a long exact sequence of their cohomology:
\begin{align}
&  0 \rightarrow \rH^0(YM_P,d_e) \rightarrow L^2 
\rightarrow \rH^1(TYM _P,d_e)  \rightarrow 0\notag \\
&\rH^i(TYM_P,d_e) = \rH^i(YM_P,d_e)\quad i\geq 2 \label{E:ffdjdfhq}
\end{align}
By definition $\rH^0(TYM_P,d_e)=0$.
Taking into account quasiisomorphism between $W^*$ and $YM_P$ we get an exact sequence of $\Spin(10)$-modules
\[ 0 \rightarrow W^*\rightarrow L^2 \rightarrow \rH^1(TYM_P,d_e)\rightarrow 0\]

It follows from the decomposition (\ref{E:decompb}) that there is  only one $\Spin(10)$-equivariant embedding of  $\Spin
(10)$-modules:$W^*\rightarrow L^2$. Also the module $L^2/W^*$ is isomorphic to $W$.
From this we conclude that $\rH^1(TYM_P,d_e)$ is isomorphic to $W$.
Due to isomorphisms (\ref{E:ffdjdfhq}) the complex $TYM_P$ has no higher cohomology. We see that the embedding $W
\to TYM_P$ is
a quasi-isomorphism. Again using spectral sequences we obtain that the embedding
\begin{equation}\label{E:tquasi}
\W\rightarrow \TYM^{\bullet}
\end{equation} 
generates an isomorphism of hypercohomology.

\section{Construction of the supersymmetric deformations}\label{comutations}
We will show that starting with $G^{\alpha}$ obeying
\begin{equation}\label{E:firsthomology}
\theta_{\alpha} G^{\alpha}=0
\end{equation} on solutions of the SYM equations we can construct a SUSY deformation of SYM.
Most of these deformation are given by the general  formula (\ref {E:fadjghjd}). We will describe solutions that lead to exceptional deformations  corresponding  to (\ref {E;laplacian}) and (\ref {g2}).

First of all we can take \[G^{\alpha}=\chi ^{\alpha}.\] Then $\theta_{\alpha}\chi ^{\alpha}=\Gamma_{\alpha}^{\alpha ij}F_{ij}=0$.

We shall denote
\begin{equation}\label{E:symmetrization}
E_1\circ \cdots \circ E_n=\frac{1}{n!}\sum_{\sigma\in S_n} E_{\sigma(1)}\cdots E_{\sigma(n)}
\end{equation} the symmetrized product of operators.  We shall be also using freely   a generalization to the formula to the $\mathbb{Z}_2$-graded case.
There is another interesting  solution is
\begin{equation}\label{E:secondsporadic}
G^{\alpha}=2\Gamma^{\alpha ijkl}_{\beta}F_{ij}\circ F_{kl}\circ\chi^{\beta}+7F_{ij}\circ F_{ij}\circ \chi^{\alpha}
\end{equation}

Let us  verify that $G^{\alpha}$ (\ref{E:secondsporadic}) satisfies (\ref{E:firsthomology}). We use $\circ$ for symmetrized product of operators (see Section \ref{S:exsupdef}).
It is an immediate corollary of (\ref{E:susyform})  that 
\[
\theta_{\alpha} \Gamma^{\alpha ijkl}_{\beta}F_{ij}\circ F_{kl}\circ\chi^{\beta}=-4\Gamma^{\alpha ijkl}_{\beta}\Gamma^i_{\alpha\gamma}\nabla_j\chi^{\gamma}\circ F_{kl}\circ \chi^{\beta}+\Gamma^{\alpha ijkl}_{\beta}\Gamma^{\beta st}_{\alpha}F_{ij}\circ F_{kl}\circ F_{st}
\]
Let $T^{ijklst}F_{ij}\circ F_{kl}\circ F_{st}$ be equal to $\Gamma^{\alpha ijkl}_{\beta}\Gamma^{\beta st}_{\alpha}F_{ij}\circ F_{kl}\circ F_{st}$. The tensor $T^{ijklst}$  is obtained from $\Gamma^{\alpha ijkl}_{\beta}\Gamma^{\beta st}_{\alpha}$ by symmetrization in $(ij)(kl)(st)$ groups of indices. $T^{ijklst}$ is zero because it defines an $\so_{10}$-invariant polynomial function of degree three on the adjoint representation of $\so_{10}$. It is well known (\cite{BVinbergALOnishchik}) that such invariants for Lie algebras of $\so$-type exist only in degrees divisible by four.

We omit    verification of \[\Gamma^{\alpha ijkl}_{\beta}\Gamma^i_{\alpha\gamma}=\frac{14}{4}\Gamma^{jkl}_{\beta\gamma}=\frac{14}{4}\left(\Gamma_{\beta\delta}^j \Gamma^{\delta \epsilon k}  \Gamma_{ \epsilon \gamma}^l-\frac{1}{2}\delta^{jk}\Gamma^{l}_{\beta\gamma} -\frac{1}{2}\delta^{kl}\Gamma^{j}_{\beta\gamma}+\frac{1}{2}\delta^{jl}\Gamma^{k}_{\beta\gamma}  \right)\]
$\Gamma$-matrix $\Gamma^{jkl}_{\beta\gamma}$ is skew-symmetric in $\beta\gamma$. We infer that 
\begin{equation}\label{E:gsimpl}
\begin{split}
&\theta_{\alpha} \Gamma^{\alpha ijkl}_{\beta}F_{ij}\circ F_{kl}\circ\chi^{\beta}=14\left(\Gamma_{\beta\delta}^j \Gamma^{\delta \epsilon k}  \Gamma_{ \epsilon \gamma}^l-\frac{1}{2}\delta^{jk}\Gamma^{l}_{\beta\gamma} -\frac{1}{2}\delta^{kl}\Gamma^{j}_{\beta\gamma}+\frac{1}{2}\delta^{jl}\Gamma^{k}_{\beta\gamma}  \right)\nabla_j\chi^{\beta}\circ F_{kl}\circ \chi^{\gamma}=\\
&=-14\Gamma^{l}_{\beta\gamma}\nabla_j\chi^{\beta}\circ F_{jl}\circ \chi^{\gamma} 
\end{split}
\end{equation}
The term that contains $\Gamma_{\beta\delta}^j \Gamma^{\delta \epsilon k}  \Gamma_{ \epsilon \gamma}^l$ vanishes by virtue of the Dirac equation. The term with $\frac{1}{2}\delta^{kl}\Gamma^{j}_{\beta\gamma}$ disappears because $\delta^{kl}F_{kl}=0$.
If we compare 
\[
\theta_{\alpha} F_{ij}\circ F_{ij} \circ \chi^{\alpha}=4\Gamma^{j}_{\beta\alpha}\nabla_i\chi^{\beta}\circ F_{ij}\circ \chi^{\alpha}
\]
with (\ref{E:gsimpl})
we conclude that $G^{\alpha}$ (\ref{E:secondsporadic}) satisfies (\ref{E:firsthomology}).

The space of solutions of (\ref{E:firsthomology}) contains a set of trivial solutions  $G^{\alpha}=\theta_{\beta}G^{\alpha\beta}$ with $\Gamma$-traceless $G^{\alpha\beta}$. Indeed by virtue of (\ref{E:hfdfjfhjd}) \[\theta_{\alpha}G^{\alpha}=\Gamma^i_{\alpha\beta}D_iG^{\alpha\beta}=0\]
It is natural to identify solutions of (\ref{E:firsthomology}), that differ by a trivial solution.

Our interest to the group $H$ of equivalence classes of such solutions  has been stimulated by existence of a map from $H$ to infinitesimal  supersymmetric deformations of the equations of motion of SYM
\[\frac{\delta \L_{SYM}}{\delta A_i}+\epsilon \M_i\quad \frac{\delta \L_{SYM}}{\delta \chi^{\alpha}}+\epsilon \M_{\alpha}\]
We omit verification that infinitesimal field redefinition $\N^i,\N^{\alpha}$ transforms
\[\begin{split}&\M_i\rightarrow  \M_i+D_iD_j\N_{\gamma j}-2D_jD_i\N_{\gamma j}+D_jD_j\N_{\gamma i} +\Gamma^i_{\alpha\beta}[\chi^{\alpha},\N^{\beta}_{\gamma}]\\
& \M_{\alpha}\rightarrow \M_{\alpha}-\Gamma^j_{\alpha\beta}[\chi^{\beta},\N_{\gamma j}]-\Gamma^i_{\alpha\beta}D_i\N^{ \beta}_{\gamma}
\end{split}\]

 The following map generalizes the operator $A$ (\ref{E:int}):
\begin{equation}\label{E:deformationG}
 \M_i=\tt_{\alpha}^{i}G^{\alpha} \quad \M_{\alpha}=\tt_{ \alpha\beta}G^{\beta}
\end{equation}
with $\tt_{\alpha}^{i}(\theta_1,\dots,\theta_{16})$ and $\tt_{ \alpha\beta}(\theta_1,\dots,\theta_{16})$ being some ingeniously chosen  noncommutative polynomials in supersymmetry operators. 
Operators  $\tt_{\alpha}^{i}$, $\tt_{\alpha \beta}$ satisfy certain equations that ensure supersymmetry of (\ref{E:deformationG}).

Operators $\tt$, $\tt_{\alpha}^{i}$, and $\tt_{\alpha \beta}$ satisfy an analogue of the descent equation (see e.g. \cite{Zumino}). To simplify notations we shall  work with  functions 
\begin{equation}\label{E:simplification}
\tt^{i}(\lambda)=\tt_{\alpha}^{i}\lambda^{\alpha},\quad \tt_{ \alpha}(\lambda)=\tt_{ \alpha\beta}\lambda^{\beta},\quad \theta(\lambda)=\theta_{\alpha}\lambda^{\alpha} ... 
\end{equation}
in pure spinor variables $\lambda^{\alpha}$.
There are several  equations they satisfy (see Appendix \ref{E:deformationcomplex}). The most relevant to our present needs   are
\begin{equation}\label{E:descent1}
\tt\theta(\lambda)=D_i\tt^{i}(\lambda)+\chi^{\beta}\tt_{\beta}(\lambda)
\end{equation}
\begin{equation}\label{E:descent2}
\begin{split}
&\theta_{\gamma}\tt=D_iB_{\gamma}^{i}+\chi^{\beta}B_{\gamma,\beta}\\
&\theta_{\gamma}^{R_1} \left( \begin{smallmatrix} \tt^{i}(\lambda)\\\tt_{ \beta}(\lambda)\end{smallmatrix}\right) =\left( \begin{smallmatrix} B_{\gamma}^{i}\theta(\lambda)\\ B_{\gamma,\beta}\theta(\lambda) \end{smallmatrix}\right)+d_1\left( \begin{smallmatrix} C_{\gamma}^{i}(\lambda)\\ C_{\gamma}^{\delta}(\lambda) \end{smallmatrix}\right) \\
\end{split}
\end{equation}
Operators $B,C$ likewise $\tt$ are noncommutative polynomials in $\theta_{\alpha}$.
The reader should consult  formulas (\ref{E:supersymmetry1},\ref{E:supersymmetry2}) and (\ref{E:differential1}) for the action  of the  supersymmetry $\theta_{\alpha}^{R_i}i=1,2$  and the operator $d_1$ . We refer the reader to Appendix \ref{E:deformationcomplex} for the details.
Let us introduce an infinitesimal field redefinition
\begin{equation}\label{E:deformationtriv}
 \N_{\gamma i}=C_{\gamma \beta, i}G^{\beta} \quad \N_{\gamma}^{\alpha}=C_{\gamma\beta}^{\alpha}G^{\beta}
\end{equation}
with $\beta$ a pure spinor index.



The immediate  corollary of (\ref{E:descent1}) is that $\M_i$ $\M_{\alpha}$ in (\ref{E:deformationG}) satisfy (\ref{E:gauegeinv}) for any $G^{\alpha}$, that solves (\ref{E:firsthomology}):
\begin{equation}
\nabla_i \M_i+[\chi^{\beta},\M_{\alpha}]=D_i\tt_{\alpha}^{i}G^{\alpha} +\chi^{\alpha}\tt_{\alpha \beta}G^{\beta} =\tt\theta_{\alpha}G^{\alpha}=0
\end{equation}
Supersymmetry of this deformation  follows  from (\ref{E:descent2}):
\[\begin{split} &\theta_{\gamma}\M_i+2\Gamma_{\gamma}^{\beta ij}D_j\M_{\beta}=D_iD_j\N_{\gamma j}-2D_jD_i\N_{\gamma j}+D_jD_j\N_{\gamma i} +\Gamma^i_{\alpha\beta}[\chi^{\alpha},\N^{\beta}_{\gamma}]\\
&-\theta_{\gamma}\M_{\beta}+\Gamma_{\gamma\beta}^i\M_i=-\Gamma^j_{\alpha\beta}[\chi^{\beta},\N_{\gamma j}]-\Gamma^i_{\alpha\beta}D_i\N^{ \beta}_{\gamma}\end{split}\]
Note that the terms involving $B^{i}_{\gamma}\theta_{\beta}G^{\beta},B_{\gamma\alpha}\theta_{\beta}G^{\beta}$ drop out because $G^{\beta}$ satisfies (\ref{E:firsthomology}).

\section{A deformation complex}\label{E:deformationcomplex}
The goal of this Appendix is to interpret equations (\ref{E:descent1}) and (\ref{E:descent2}) as a part of a more general system of equations. A solution of this system  is a cocycle  in a certain bi-complex:
\begin{equation}\label{E:bicomplex}
\begin{array}{lllllll}
\cdots&&\cdots&&\cdots&&\cdots\\
\uparrow^{{d^{II}_1}}&&\uparrow^{d^{II}_1}&&\uparrow^{d^{II}_1}&&\uparrow^{d^{II}_1}\\
 E^{1}_0&\overset{d^I_0}\leftarrow &E^{1}_1&\overset{d^I_1}\leftarrow  &E^{1}_2&\overset{d^I_2}\leftarrow   &E^{1}_3\\
\uparrow^{d^{II}_0}&&\uparrow^{d^{II}_0}&&\uparrow^{d^{II}_0}&&\uparrow^{d^{II}_0}\\
E^{0}_0&\overset{d^I_0}\leftarrow &E^{0}_1&\overset{d^I_1}\leftarrow  &E^{0}_2&\overset{d^I_2}\leftarrow   &E^{0}_3\\
\end{array}
\end{equation}
The horizontal differential $d^I$ is 
defines  a four-term complex
\[R_0\overset{d_0}\leftarrow R_1\overset{d_1}\leftarrow  R_2\overset{d_2}\leftarrow   R_3\]

In the  simplest form the linear spaces $R_0,R_1,R_2,R_3$ of the complex are formed by noncommutative polynomials with constant coefficients
in covariant derivatives of curvature and spinor field  with values in sections of adjoint bundle 
, the fields satisfy equations of motion of SYM. The complex  governs infinitesimal deformations of SYM. The space $R_1$ contains infinitesimal deformations  $(\M_i(A,\chi),\M_{\alpha}(A,\chi))$ of Euler-Lagrange equation
\begin{equation}\label{E:variation}
\frac{\delta\L_{SYM}}{\delta A_i}+h\M_i=0\quad \frac{\delta\L_{SYM}}{\delta \chi^{\alpha}}+h\M^{\alpha}
\end{equation}

 The space $R_2$ contains infinitesimal field redefinitions $(\N_i(A,\chi),\N^{\alpha}(A,\chi))$:
 \begin{equation}\label{E:fieldredef}A_i\rightarrow A_i+h \N_i\quad \chi^{\alpha}\rightarrow   \chi^{\alpha}+h \N^{\alpha}\end{equation}
 The spaces $R_0,R_3$ are spanned by $\mathcal{Y}(A,\chi)$, which, as we already mentioned, in case $R_0$ define  a Lagrangian density $\tr \mathcal{Y}(A,\chi)$. A cohomology of $d_2$ we interpret as an infinitesimal automorphism of a solution  $(\nabla_i,\chi^{\alpha})$.
 
The differential $d_0,d_1,d_2$ have  forms
\begin{equation}\label{E:differential0}
 d_0\left( \begin{smallmatrix}   \M_i\\\M_{\alpha}\end{smallmatrix}\right)=D_i\M_i+[\chi^{\alpha},\M_{\alpha}]
 \end{equation}
$d_0$ can be obtained from 
\begin{equation}\label{E:gauegeinv}
\nabla_i \frac{\delta \L}{\delta A_i}+[\chi^{\alpha}\frac{\delta \L}{\delta \chi^{\alpha}}]=0.
\end{equation} by variation (\ref{E:variation}).
\begin{equation}\label{E:differential1}
\begin{split}
d_1\left( \begin{smallmatrix}   \N_i\\\N^{\alpha}\end{smallmatrix}\right)=
\left( \begin{smallmatrix}  D_iD_j\N_{j}-2D_jD_i\N_{j}+D_jD_j\N_{i} +\Gamma^i_{\alpha\beta}[\chi^{\alpha},\N^{\beta}]\\-\Gamma^j_{\alpha\beta}[\chi^{\beta},\N_{j}]-\Gamma^j_{\alpha\beta}D_j\N^{\beta}\end{smallmatrix}\right)
\end{split}
\end{equation}
$d_1\left( \begin{smallmatrix}   \N_i\\\N^{\alpha}\end{smallmatrix}\right)$ is a variation of equations (\ref{E:relssfsyy1},\ref{E:relssfsyy3}) under infinitesimal field redefinition (\ref{E:fieldredef}).
\begin{equation}\label{E:differential2}
\begin{split}
d_2\mathcal{Y}=\left( \begin{smallmatrix}   D_j\mathcal{Y}\\-[\chi^{\alpha},\mathcal{Y}]\end{smallmatrix}\right)
\end{split}
\end{equation}
$d_2\mathcal{Y}$ stands for a field-dependent infinitesimal gauge transformation. For the first time the non-supersymmetric analogue of the complex $R$ has been analyzed in mathematical literature in \cite{CD}. The variant that contains fermions and scalar bosons has been worked out in \cite{MSch2}.
The differentials $d_i$ commute with  supersymmetries 
\[\begin{split}
\theta^{R_i}_{\alpha}(\mathcal{Y})=\theta_{\alpha}\mathcal{Y}\quad i=0,3
\end{split}\]

\begin{equation}\label{E:supersymmetry1}
\begin{split}
\theta^{R_1}_{\alpha}\left( \begin{smallmatrix}   \M_i\\\M_{\alpha}\end{smallmatrix}\right)=
\left( \begin{smallmatrix}   \theta_{\alpha}\M_i+2\Gamma_{\alpha}^{\beta ij}D_j\M_{\beta}\\-\theta_{\alpha}\M_{\beta}+\Gamma_{\alpha\beta}^i\M_i\end{smallmatrix}\right)
\end{split}
\end{equation}

\begin{equation}\label{E:supersymmetry2}
\begin{split}
\theta^{R_2}_{\alpha}\left( \begin{smallmatrix}   \N_i\\\N^{\beta}\end{smallmatrix}\right)=
\left( \begin{smallmatrix}   \theta_{\alpha}\N_i+\Gamma_{\alpha\gamma}^i\N^{\gamma}\\-\theta_{\alpha}\N^{\beta}-2\Gamma_{\alpha}^{\beta ij}D_i\N_{j}\end{smallmatrix}\right)
\end{split}
\end{equation}
We do not give here the formulas for infinitesimal shift operators in $R_i$, but mention in passing that they act trivially in cohomology as  operators $\mathcal{Y},\M,\N$ have space-time constant coefficients. The algebra generated by operators $\theta^{R_i}_{\alpha}, i=1,2$  is rather complex. However  operators $\theta^{R_i}_{\alpha}, i=0,1,2,3$ anti-commute in cohomology.

In a  more sophisticated version $\Theta$ of the complex $R$, which is used in construction of the bicomplex (\ref{E:bicomplex}), $\mathcal{Y},\M_i, \M_{\alpha},\N_i,\N^{\alpha}$ are noncommutative polynomials in $\theta_{\alpha}$. Operators $d_0,d_1,d_2$  can be adapted to $\Theta$ if we replace $D_i$  and $[\chi^{\alpha},\cdot]$ in the formulas (\ref{E:differential0},\ref{E:differential1},\ref{E:differential2}) by  
\begin{equation}\label{E:Dformula}
D_i=1/8\Gamma^{\alpha\beta}_i\theta_{\alpha} \theta_{\beta}.
\end{equation} 
\begin{equation}\label{E:chiformula}
1/5\Gamma^{\delta\gamma}_i[\theta_{\beta}, D_i]=1/40\Gamma^{\delta\gamma}_i\Gamma^{\alpha\beta}_i[\theta_{\gamma}, \theta_{\alpha} \theta_{\beta}].
\end{equation} 
 
The bicomplex $E=\bigoplus_{ij} E^{j}_{i}$ is the  tensor product \[E^{j}_{i}=\Theta_i\otimes \Ss_j.\] The differential  $d^{II}$ is a right multiplication on $\pm\theta(\lambda)=\pm\theta_{\alpha}\lambda^{\alpha}$:
\begin{equation}\label{E:diffII}
\mathcal{Y},\M_i, \M_{\alpha},\N_i,\N^{\alpha}\rightarrow \mathcal{Y}\theta(\lambda),\M_i\theta(\lambda), \M_{\alpha}\theta(\lambda),\N_i\theta(\lambda),\N^{\alpha}\theta(\lambda)
\end{equation}

The operator $c_0=\tt$ (\ref{E:int}) is an element of $E_0^0$, operators $c_1=(\tt^{i}(\lambda),\tt_{ \alpha}(\lambda))$ (\ref{E:simplification}) belong to $E_1^1$. There are also $c_2=(\widetilde{\tt}^{ i}(\lambda),\widetilde{\tt}^{ \alpha}(\lambda))$ and $c_3=\widetilde{\tt}(\lambda)$-functions of degree two and three in $\lambda$.
The full set of equations on $c=(c_0,c_1,c_2,c_3)$ that generalizes (\ref{E:descent1})  is:
\[d_0^{II}c_0=d^I_0c_1\quad d_1^{II}c_1=d^I_0c_2 \quad d_2^{II}c_2=d^I_0c_3 \quad d_3^{II}c_3=0 \]
 
 The cohomology of the diagonal complex  $Tot(E)^k=\bigoplus_{k=j-i} E_i^j$ (\ref{E:bicomplex}) is intractable, but if we choose to work with all the gauge groups $\U(N)$ simultaneously, i.e.  let $\mathcal{Y},\M_i, \M_{\alpha},\N_i,\N^{\alpha}$ be  noncommutative polynomials in $\bm{\theta}_{\alpha}\in L$ (see Section \ref{2.1} ) then cohomology are finite dimensional:
 \[H^0(Tot(E))=H^3(Tot(E))=\mathbb{C}\quad \rH^1(Tot(E))=V^*+S^* \quad \rH^2(Tot(E))=V+S\]
Collection  $c_0,c_1,c_2,c_3$ is a generator in $\rH^0(Tot(E))$. Note that  $c_3$ satisfies $d_3^{II}c_3=0$. In fact it is a generator of the only nontrivial $\theta(\lambda)$-cohomology class in $U(L)\otimes \Ss$.  It contains a representative that can be conveniently written in terms of graded symmetric product and differentiations (\ref{E:diff}):
\[\Gamma_{\alpha \beta}^m\lambda^{\alpha}            \Gamma_{\gamma \delta}^n\lambda^{\gamma}      \Gamma_{\epsilon \varepsilon}^k\lambda^{\epsilon}          \Gamma_{\mu \nu}^{mnk}        \frac{\sd^5}{\sd \theta_{\beta}\sd \theta_{\delta}\sd \theta_{\varepsilon}\sd \theta_{\mu}\sd \theta_{\nu}} \theta_{1}\circ\cdots \circ \theta_{16}\] 




To carry out the computations of $c_1$ we choose $\theta_{\alpha}$ to be a weight basis. Then $\Gamma_{\alpha\alpha}^i=0$. Operators  $\left( \begin{smallmatrix} I^{i}(\lambda)\\I_{ \beta}(\lambda)\end{smallmatrix}\right)$ (\ref{E:deformationG},\ref{E:descent1}) can be computed by the formula:
\[\begin{split}
&\left( \begin{smallmatrix} I^{i}(\lambda)\\I_{ \beta}(\lambda)\end{smallmatrix}\right)=\sum_{\beta<\alpha_1}\sum _{k>0}\\\\
&\left(\sum_{\alpha_{2k+1}<\cdots<\alpha_1}(-1)^{\alpha_1+\alpha_3+\cdots+\alpha_{2k+1}} \left(\theta^{R_1}_{\alpha_{2k+1}}\cdots \theta^{R_1}_{\alpha_{2}}\left( \begin{smallmatrix} [\theta_{\alpha_{1}},\theta_{\beta}]\lambda^{\beta}\\0\end{smallmatrix}\right) \right)\theta_1\cdots \widehat{\theta}_{\alpha_{2k+1}}\cdots \widehat{\theta}_{\alpha_{1}}\cdots \theta_{16}+\right.\\
&+\left.\sum_{\alpha_{2k}<\cdots<\alpha_1}(-1)^{\alpha_1+\alpha_3+\cdots+\alpha_{2k-1}+\alpha_{2k}+1} \left( \theta^{R_1}_{\alpha_{2k}} \cdots \theta^{R_1}_{\alpha_{2}} \left( \begin{smallmatrix} [\theta_{\alpha_{1}},\theta_{\beta}]\lambda^{\beta}\\0\end{smallmatrix}\right)\right) \theta_1\cdots \widehat{\theta}_{\alpha_{2k}}\cdots \widehat{\theta}_{\alpha_{1}}\cdots \theta_{16}\right)\\
\end{split}
\]

The summands in the formula should be understood as follows. We apply consequently operators $\theta^{R_1}_{\alpha_{i}}$ (\ref{E:supersymmetry1}) to $ \left( \begin{smallmatrix}  [\theta_{\alpha_{1}},\theta_{\beta}]\lambda^{\beta}\\0\end{smallmatrix}\right)$. Then we multiply from the right components of the resulting two-vector on $ \theta_1\cdots \widehat{\theta}_{\alpha_{n}}\cdots \widehat{\theta}_{\alpha_{1}}\cdots \theta_{16}$
The proof is based on the formula
\[\begin{split}
&\theta_1\cdots \theta_{16}\theta_{\beta}=\sum_{\beta<\alpha_1}\sum^n _{k=1}\\
&\left(\sum_{\alpha_{2k+1}<\cdots<\alpha_1}(-1)^{\alpha_1+\alpha_3+\cdots+\alpha_{2k+1}}[\theta_{\alpha_{2k+1}},[\dots ,[\theta_{\alpha_{1}},\theta_{\beta}]]\theta_1\cdots \widehat{\theta}_{\alpha_{2k+1}}\cdots \widehat{\theta}_{\alpha_{1}}\cdots \theta_{16}+\right. \\
&\left.\sum_{\alpha_{2k}<\cdots<\alpha_1}(-1)^{\alpha_1+\alpha_3+\cdots+\alpha_{2k-1}+\alpha_{2k}+1}[\theta_{\alpha_{2k}},[\dots,[\theta_{\alpha_{1}},\theta_{\beta}]]\theta_1\cdots \widehat{\theta}_{\alpha_{2k}}\cdots \widehat{\theta}_{\alpha_{1}}\cdots \theta_{16}\right)\\
&+r_n
\end{split}
\]
where 
\[r_n=
\begin{cases}
\sum_{\beta<\alpha_1}\sum_{\alpha_{2k+1}<\cdots<\alpha_1}(-1)^{\alpha_1+\alpha_3+\cdots+\alpha_{2k-1}}\theta_1\cdots [\theta_{\alpha_{2k}},[\dots ,[\theta_{\alpha_{1}},\theta_{\beta}]]\cdots \widehat{\theta}_{\alpha_{2k-1}}\cdots \widehat{\theta}_{\alpha_{1}}\cdots \theta_{16}& n=2k\\
\sum_{\beta<\alpha_1}\sum_{\alpha_{2k+1}<\cdots<\alpha_1}(-1)^{\alpha_1+\alpha_3+\cdots+\alpha_{2k+1}}\theta_1\cdots [\theta_{\alpha_{2k+1}},[\dots ,[\theta_{\alpha_{1}},\theta_{\beta}]]\cdots \widehat{\theta}_{\alpha_{2k}}\cdots \widehat{\theta}_{\alpha_{1}}\cdots \theta_{16} & n=2k+1\\
\end{cases}\]
It can verified by induction on $n$.

Few words about the proof of existence of $c$. The algebra $U(L)$ is a free left $U(YM)$-module. 
Then $0=\rH_i(YM,U(L))=\rH^{3-i}(YM,U(L))$ for $i\geq 1$. The complex $\Theta$ computes $\rH^{i}(YM,U(L))$. The element $d^{II}c_1$ is a $d^I$-cocycle. By the vanishing result the cocyle $d^{II}c_1$ is $d^I$-trivial: $d^{II}c_1=d^{I}c_2$. We find $c_3$ along the same lines. We leave the proof of non-triviality of $c$ to the reader.

Note that  $\theta_{\alpha}c=(\theta^{R_0}_{\alpha}c_0,\theta^{R_1}_{\alpha}c_1,\theta^{R_2}_{\alpha}c_2,\theta^{R_3}_{\alpha}c_3)$ is a cocycle (the action of the supersymmetries commutes with the differential). $\theta_{\alpha}c$ has a degree in $\theta$ on one greater  then $c$. This means that $\theta_{\alpha}c$ is a trivial cocycle. The second equation in (\ref{E:descent2}) follows from this.

\section{Dimensional reductions of SYM}\label {S:Appendix H}
 
 We have analyzed   supersymmetric deformations of  ten-dimensional SYM theory and of this theory reduced to a point. Similar methods can be applied to $d$-dimensional reduction of ten-dimensional SYM theory for any $d\leq 10$.  Here we discuss some general results in this direction  leaving concrete calculations for the future work.

 We can generalize the results of  Section \ref{S:homologicalapp} introducing the Lie algebra  $T_dYM$. Recall that the Lie algebra $YM$  as a vector space can be represented as a direct sum 
 $\sum _{n\geq 2} L^n$ where $L^2$ has a basis $\bD_1,..., \bD_{10}$. We define $T_dYM$    as a subalgebra $(L^{2})'+\sum _{n\geq 3} L^n$ of the algebra $YM$. Here $(L^{2})'$ is spanned by $\bD_{d+1},...,\bD _{10}.$ (In more invariant way one can say that $(L^{2})'$ is a $(10-d)$-dimensional subspace of $L^2$, such that the restriction of the  inner-product from $L^2$ to $(L^{2})'$ is non-degenerate.)    We will modify below  the arguments of Appendix C  to obtain the following generalization of theorems \ref{T:dfadfqq} and \ref{T:theorem12} :
\begin{proposition}
Every element of $\rH^{\bullet}(L,U(T_dYM))$ where $k\geq 2$ specifies an infinitesimal supersymmetric deformation of  equations of motion of SYM theory reduced to $d$-dimensional space.
\end{proposition}

This statement should be understood in the framework of BV-formalism.   In this formalism we interpret invariance with respect to Lie algebra $\g$ in terms of L$\ity$ action (Section 6
and Appendix A).  The solutions of the equations of motion  correspond to the points of zero locus of the vector field $Q$; the invariance of  equations of motion  with respect to  L$\ity$ action is expressed by the equation
\begin{equation}\label{E:fadafdh}
d_{\g}q+[Q,q]+\frac{1}{2}[q,q]=0.
\end{equation}
where $q$ is a ghost dependent vector field.
Infinitesimal deformations of $Q$ and $q$ obeying this equation are labeled by elements of some homology group that can be interpreted as 
Lie-Hochschild cohomology $\rHH_{\g}^{\bullet}(\cal A,\cal A )$; in our case $\g=\susy _d$ stands for the algebra of supersymmetries in dimension $10$ reduced to the dimension $d$ and for $\mathcal{A}$ we can take the algebra $bv_d$ obtained from $bv$ by means of dimensional  reduction. {\footnote {Notice, that the most interesting deformations  correspond to the elements of $\rHH_{\g}^k(\cal A, \cal A)$ with $k=2$ ; the elements with $k>2$ correspond to deformation of higher terms in L$\ity$ action; these higher terms do not have direct physical meaning.}}
We can use pure spinor formalism also in the  case of reduction to $d$ dimensions; then for $\mathcal{A}
$ we should take $d$-dimensional reduction of Berkovits algebra $B$ that can be defined as the algebra $B_d$ of polynomial functions depending of pure spinor $\lambda$, odd spinor $\psi$ and $x\in \mathbb{R}^{d}$ with the differential defined as the derivation 
\begin{equation}\label{E:differentia}
\lambda^{\alpha}\bigg(\pr{\psi^{\alpha}}+\Gamma_{\alpha\beta}^i\psi^{\beta}\pr{x^{i}}\bigg).
\end{equation}
The algebra $B_d$ is a supercommutative quadratic algebra; its Koszul dual is a universal enveloping algebra of differential graded Lie algebra $\tilde L_d$ defined as an algebra with generators  $\bm{\theta}_{1},\dots, \bm{\theta}_{16}$ of degree  one, generators $s_{1},\dots,s_{16} $ of degree zero,and generators $\varsigma_{1},\dots ,\varsigma_{d}$ of degree $1$. They satisfy relations 
\begin{equation}\label{E:nvpr}
\Gamma^{\alpha\beta}_{i_1,\dots,i_5}[\bm{\theta}_{\alpha},\bm{\theta}_{\beta}]=0.
\end{equation}
\begin{equation}
\label{nv}
[\bm{\theta}_{\alpha},s_{\beta}]=\Gamma_{\alpha\beta}^i\varsigma_i
\end{equation}
(Other commutation relations are trivial.)  
The differential acts by the formulas  $d(\bm{\theta}_{\alpha})=s_{\alpha},d(\bD_i)=\varsigma_i$ for $i\leq d$, $d(\bD_i)=0$ for $i>d$. It is easy to construct a natural embedding of 
$T_dYM$ into $\tilde L_d$ and to prove that this embedding is a quasi-isomorphism; similar statements are true for their universal envelopes. These statements together with the fact that $B_d$ is a Koszul algebra permit us to say that
$B_d$ is dual to $U(T_dYM)$ in the sense of Appendix B. From other side, $U(T_dYM)$ is dual to $C^{\bullet}(T_dYM)$. (Recall that $C^{\bullet}(\cal G)$ denotes the algebra of polynomial functions of ghosts of Lie algebra $\cal G$ with  the differential calculating the cohomology  $\rH^{\bullet}(\cal G,\mathbb{C})$; the duality between
 $C^{\bullet}(\cal G)$ and $U(\cal G)$ follows the isomorphism between Lie algebra cohomology
$\rH^{\bullet}(\cal G,\mathbb{C})$ and Hochschild cohomology $\rHH^{\bullet}(U(\cal G),\mathbb{C})$.) We obtain that $B_d$ is quasi-isomorphic to $C^{\bullet}(T_dYM)$, hence the supersymmetric deformations are governed by Lie-Hochschild cohomology $\rHH_{\susy _d}{\bullet}(C^{\bullet}(T_dYM),C^{\bullet}(T_dYM))$ that can be regarded as cohomology of a complex $$C^{\bullet}(\susy _d)\otimes \Sym (\Pi T_dYM)^*\otimes \Sym T_dYM$$ with appropriate differential.
 To finish the proof we notice that this cohomology
is isomorphic to $\rH^{\bullet} (L,U(T_dYM)$. (This fact can be derived from the remark that
$C^{\bullet}(\susy _d)\otimes \Sym (\Pi T_dYM)^*=\mathbb{C}[[t^1,\dots,t^{16}]]\otimes \Sym(\Pi YM^*)=\Sym(\Pi L^*)$.)

In Section 2.3 we have constructed vector bundles $\W^*$ and $\W$ over $\mathcal{Q}$. The bundle
$\W^*$ is embedded into trivial vector bundle $\cal V$ with the fiber $V=L^2$;the bundle $\W$ is defined as a quotient $\cal V/\W^*$.
 We define a two-step complex of vector bundles
\begin{equation}\label{E:twoste}
{\cal V}\overset {p}\rightarrow  \W.
\end{equation}

This complex is quasi-isomorphically embedded into the complex of vector bundles $\YM$. This
follows from the remark that it is quasi-isomorphic to the bundle $\W^*$ and from the considerations  of Appendix E. We can  consider a more general two-step complex 
\begin{equation}\label{E:twostep}
{\cal V'}\overset {p}\rightarrow  \W.
\end{equation}
where $\cal V'$ is a trivial vector bundle with the fiber $V'=(L^2)'$. It is embedded into $\mathcal{T}_d\YM$ and this embedding is a quasi-isomorphism. This
can be seen as follows. The zero component of the complex $\mathcal{T}_d\YM$ coincides with $V^{'}$. The zero-truncated complex, i.e. the complex with removed first component , coincides with $\TYM$. We know (see ( \ref{E:tquasi})) that it is quasi-isomorphic to $\W$. 
The statement we need follows from this fact.

 Observe that if we work with the  unreduced theory then $V'=0$ and the complex becomes simply the vector bundle $\W$ in agreement with Proposition \ref{P:WW}.  In the opposite extremity we can study the  theory reduced to a point. Then $V'=V$ and we obtain  the complex (\ref {E:twoste}) which is  quasiisomorphic to $\W^*$. Again this agrees with the results of Proposition \ref{P:W}.

We see that in the case at hand we can write down an analog of (\ref{E:delta}) and apply it to the calculation of   $\rH^i(L,U(T_dYM))$.   The main remaining problem is the calculation of hypercohomology 
\begin{equation}\label{E:hypercohomology}
\mathbb{H}^i(\Sym^k(\BEE)(j)).
\end{equation}
 We are planning to solve this problem  using the techniques of \cite{Macaulay2}.  
 
Notice that the calculation of Euler characteristic of the cohomology  we consider can be done by means of methods used in  Section 4.
\section { Representation rings}\label  {S:Appendix I}

In the set of equivalence classes of representations of a  group $G$ or of $\g$-modules (of representations of Lie algebra
$\g$ ) we can introduce a notion of sum ( direct sum) and of product (tensor product).   One can consider also virtual modules (virtual representations) defined as formal differences
of modules (representations). The set   of equivalence classes of virtual modules can be considered as a ring with respect to direct sum and tensor product;  this ring is called representation ring of group $G$ or of Lie algebra $\g$. \footnote {Notice that in the definition of representation ring one can consider only finite-dimensional representations or only unitary representation, etc }If $G$ is a simply connected Lie group and $\g$ denotes its Lie algebra the representation group of $G$ is isomorphic to the
representation ring of $\g$.

One can use another definition of representation ring that is based on consideration of representations in superpaces instead of virtual representations. (If $V$ is a module and $\Pi V$ is obtained from $V$ by means of parity reversion we assume that $V+\Pi V$ is equivalent to zero. In other words the virtual representation $V-W$ can be considered as a representation  $V+\Pi W$.)

To every representation $\rho$ of a group $G$ we can assign its character  $\chi$ as a function on the set of conjugacy  classes of $G$ defined by the formula $\chi (g)= Tr \rho (g)$. It is easy to check that the character of direct sum of representations is a sum of characters and the character of tensor product is a product of characters. Using this remark one can verify that for a compact group $G$ the representation ring
is isomorphic to a subring of class functions (functions on the set of conjugacy classes of $G$.) Equivalently
one can say that the representation ring is  isomorphic to a subring of the ring of $W(G)$-invariant functions on maximal torus where $W(G)$ stands for the Weyl group.

A graded representation $V=\sum V_k$ can be considered as a power series $\sum V_kt^k$ taking values in the representation ring,

One can define some natural operations on the representation ring (the word "natural" means here that these operations are compatible with homomorphisms).

First of all we can define  operations $\lambda ^i$ sending a module $V$ into its exterior power
$\Lambda ^i V$ and operators $\sigma^i$ sending a module into its symmetric power. We understand here the exterior power and the symmetric power in the sense of superalgebra , hence
$$\lambda ^i(-V)=(-1)^i \sigma^i(V)$$
(changing parity we exchange symmetry with antisymmetry).
It is convenient to consider generating functions
$$\lambda_t (V)=\sum \lambda ^i (V)t^i,$$ 
$$\sigma_t (V)=\sum \sigma ^i(V)t^i.$$
(These functions arise naturally from graded representations $\Lambda V=\sum \Lambda^iV$ and
$\Sym V=\sum \Sym ^kV$.)

Then
 $$\lambda _t(V) S_t(V)=1.$$
 We will consider only representation rings of compact Lie groups and corresponding Lie algebras (reductive Lie algebras). Then  it is sufficient to check an identity between natural operations for the group ring of $\U(1)$;  this follows from the fact a character of a representation is determined by its restriction to the maximal torus.  Using this remark we derive the above identity from  the relation 
 $$(1-xt)(1+xt+x^2t^2+...)=1.$$
 
We define Adams operations $\Psi ^i$ in terms of action on characters: operation $\Psi ^i$ transforms a class function $\chi (g)$ into a class function $\chi (g ^i)$.  It is obvious that these operations are homomorphisms of representation ring:
$$\Psi ^i (V+W)=\Psi ^i (V) +\Psi ^i (W),$$ 
$$\Psi ^i (VW)=\Psi ^i (V) \Psi ^i (W).$$
It is clear also that
$$\Psi ^k\Psi ^l (V)=\Psi ^{kl}(V).$$
The generating function $\Psi _t (V)=\sum \Psi ^i (V) t^i$ can be expressed in the form
$$\Psi _t(V)=-t\frac{d}{dt}\log \lambda _t (V)=t\frac{d}{dt}\log \sigma_t (V).$$
The proof of this relation also can be reduced to  the consideration  of the representation ring of $\U(1)$. 
Conversely, this relation allows us to  express $\lambda ^k$ and $\sigma^k$ in terms of Adams operations.
In particular,
$$\sigma_t(V)=\exp(-\sum \frac{\Psi^k(V)t^k}{k}).$$
In terms of characters this formula can be written in the following way:
\begin{equation}
\label{ad}
\sum t^k\chi _{\Sym ^kV}(g)=\exp (-\sum \frac{t^k}{k}\chi_V(g^k))
\end{equation}

\end{document}